\documentclass[12pt]{amsart}
\usepackage[nosort]{cite}
\usepackage{latexsym,amsmath,amssymb,graphicx,epsfig,color,hyperref,supertabular}
\newsavebox{\fmbox}

\input cyracc.def 
\newfont{\twelvecyr}{wncyr10 at 12pt}

\def\sha{\text{\twelvecyr\cyracc{Sh}}}

\def\pplogo{\vbox{\kern-\headheight\kern -15pt
\halign{##&##\hfil\cr&{
\ppnumber}\cr\rule{0pt}{2.5ex}&\ppdate\cr} }} \makeatletter
\def\ps@firstpage{\ps@empty \def\@oddhead{\hss\pplogo}%
  \let\@evenhead\@oddhead 
}
\def\maketitle{\par
 \begingroup
 \def\thefootnote{\fnsymbol{footnote}}
 \def\@makefnmark{\hbox
 to 0pt{$^{\@thefnmark}$\hss}}
 \if@twocolumn
 \twocolumn[\@maketitle]
 \else \newpage
 \global\@topnum\z@ \@maketitle \fi\thispagestyle{firstpage}\@thanks
 \endgroup
 \setcounter{footnote}{0}
 \let\maketitle\relax
 \let\@maketitle\relax
 \gdef\@thanks{}\gdef\@author{}\gdef\@title{}\let\thanks\relax}
\makeatother

\def\ppnumber{\vbox{\baselineskip14pt
\hbox{UCSB Math 2011-08}
\hbox{IPMU 11-0106}
 }}
\def\ppdate{}
\date{} 

\title[Anomalies and the Euler characteristic]{Anomalies and the
Euler characteristic of elliptic Calabi--Yau threefolds}

\author[A. Grassi and D. R. Morrison]{Antonella Grassi and David
R. Morrison}

\address{Department of Mathematics, University of Pennsylvania,
Philadelphia, PA 19104} \email{grassi@math.upenn.edu}

\address{Institute for the Physics and Mathematics of the Universe,
The University of Tokyo, Kashiwa, Chiba 277-8582, Japan, and}

\address{Departments of Mathematics and Physics, University of
California, Santa Barbara, CA 93106}
\email{drm@math.ucsb.edu}

\thanks{Research partially supported by
National Science
Foundation grants DMS-0606578, DMS-1007414, PHY-1066293,
and by World Premier International Research Center Initiative (WPI Initiative), MEXT, Japan.
}

\theoremstyle{plain}
\newtheorem{theorem}{Theorem}

\newtheorem{lemma}[theorem]{Lemma}
\newtheorem{definition}[theorem]{Definition}
\numberwithin{theorem}{section} \numberwithin{equation}{section}

\theoremstyle{remark}
\newtheorem{remark}[theorem]{Remark}
\newcommand{\Tr}{\operatorname{Tr}}
\newcommand{\tr}{\operatorname{tr}}
\newcommand{\adj}{\operatorname{adj}}
\newcommand{\fund}{\operatorname{fund}}

\newcommand{\spinrep}{\operatorname{spin}}
\newcommand{\vect}{\operatorname{vect}}
\newcommand{\ad}{\operatorname{ad}}
\newcommand{\End}{\operatorname{End}}
\newcommand{\trace}{\operatorname{trace}}

\newcommand{\ord}{\operatorname{ord}}

\def\bs1{\boldsymbol\Sigma}

\def\p2{\mathbb P^2}
\def\r{\boldmath  $\mathcal R$}

\def\c{\chi _{top}}
\def\kb{ K_B}
\def\dim{\operatorname{dim}}
\def\rk{\operatorname{rk}}

\def\f{\frac{1}{2}}

\renewcommand{\div}{\operatorname{div}}

\newcommand{\LB}[1]{#1L}
\newcommand{\minusLB}{-L}
\newcommand{\optionLB}[2]{#2L}
\newcommand{\parenLB}{L}

\newcommand{\Sing}{\operatorname{Sing}}
\newcommand{\betaSigma}{\beta_\Sigma} 
\newcommand{\gammaSigma}{\gamma_\Sigma} 
\newcommand{\deltaSigma}{\delta_\Sigma} 

\newcommand{\rhoP}{\rho_P}

\begin{document}
\begin{abstract}
We investigate the
 delicate interplay between the types of singular fibers in
elliptic fibrations  of Calabi--Yau threefolds (used to formulate F-theory) and the
``matter'' representation of the associated Lie algebra.
The main tool is the analysis and the appropriate interpretation of the anomaly formula for six-dimensional supersymmetric theories.
We find that this anomaly formula is geometrically captured by
a relation among codimension two cycles on the base of the elliptic
fibration, and that this relation holds for elliptic fibrations of any dimension.
We introduce a ``Tate cycle'' which efficiently describes this relationship, and which
is remarkably easy to calculate explicitly from the Weierstrass equation
of the fibration.
 We check the anomaly cancellation formula
 in a
number of situations and show how this formula constrains
the geometry (and in particular the Euler characteristic) of
the Calabi--Yau threefold.

\end{abstract}
\maketitle

Traditional compactification of string theory or M-theory on the
product of a Calabi--Yau manifold and Minkowski space leads to a low-energy
effective theory in which---at least semiclassically---the
physics in Minkowski space is encoded by
well-known and clearly understood geometric features of the Calabi--Yau
manifold.

One of the early lessons in the second superstring revolution
was that for the type II theories, as well as M-theory and F-theory,
singularities in the Calabi--Yau manifold could lead to interesting and
important physical effects.  The first of these is non-abelian gauge
symmetry arising from ADE singularities in complex codimension two
\cite{Witten:1995ex,Aspinwall:1995zi};
refinements to the original story show that
non-simply-laced as well as simply-laced gauge groups can occur
\cite{MR1416960, geom-gauge, fiveDgauge, Witten:1997kz, Diaconescu:1998cn}.
The second effect is  charged matter in the
non-abelian gauge sector \cite{geom-gauge,Katz:1996xe}, which
arises from the singular structure in complex codimension three.
(The codimensions here are
codimensions within the Calabi--Yau manifold; in F-theory, one
sees the same effects in the base of an elliptic fibration on the
Calabi--Yau manifold, where the gauge groups are associated with
complex codimension one and the matter with complex codimension two.)

There are many subtleties involved in determining the physics in this
context, and one important guide has been consistency of the induced
low-energy theory.  In this paper, we investigate the conditions imposed
by anomaly cancellation on F-theory models in six dimensions.  There is
a delicate interplay between the types of singular fibers in the
elliptic fibration used to formulate F-theory, and the corresponding
matter representation.  We not only check the anomaly cancellation formula
(based on an intrinsic geometric analysis of matter from singular
fibers, as in \cite{Katz:1996xe,fiveDgauge,LieF}) in a
number of situations, but we also show how this formula constrains
the geometry: given the representation, one can often determine the
allowed set of geometries.

We have found a natural setting for our key computation which applies to
an elliptically fibered Calabi--Yau manifold of arbitrary dimension.  In our
setup, for each such manifold we describe both a gauge {\em divisor}\/ and a
matter {\em cycle}\/ (of codimension two) on the base; the anomaly
cancellation is the result of a tight relationship between these two.

For an elliptic Calabi--Yau threefold, the associated
 six-dimensional quantum field theory is again
a gauge  theory, and in
order for it to be  consistent,  the gauge, gravitational, and mixed anomalies must
vanish.  The vanishing of one of these anomalies can be interpreted as a
formula
for the Euler characteristic of the Calabi--Yau manifold; the
vanishing of the others severely constrains the ``dictionary'' between
singular fibers in the elliptic fibration and the matter representation
of the gauge theory.  Our earlier paper \cite{grouprep} focused
on the Euler characteristic formula.  In this paper, we will concentrate
on the other anomaly cancellation formulas and the correspondence between
singular fibers and the matter representation (in our more general
setting of a ``matter cycle'').  We stress that these are closely related.

Interest in these questions has been revived recently by some work
on the consistency of low-energy 6D supergravity theories
\cite{Kumar:2009us,Kumar:2009ae}.  A systematic application of constraints on
the low-energy theories,  combined with  anomaly cancellation,
 comes fairly close to matching the set of F-theory vacua
\cite{mapping,global6D,Kumar:2010am}.

Moreover, the continued importance of F-theory in constructing string
vacua in four dimensions\footnote{See \cite{Denef:2008wq,Heckman:2010bq} and
references therein
for recent progress in this direction.}
makes a precise understanding
of these issues quite important.  The lessons learned from the present study
of six-dimensional models have immediate applications to four-dimensional
F-theory models, including those with fluxes and
branes.  In addition, there are other
issues in four dimensions, particularly those related to singularities
in complex codimension three in the base of the F-theory elliptic
fibration \cite{codimthree,Denef:2005mm}, whose complete physical
understanding will undoubtedly require an  understanding of
 the complex codimension two
effects investigated in this paper.

\bigskip

The formula for anomalies is perfectly well-defined in terms of quantities
in the quantum field theory, but a complete dictionary between the
geometry
of the Calabi--Yau manifold and the corresponding quantities in the
quantum
field theory is not yet known.  Thus, our purpose will be to not only
verify this formula, but to complete the dictionary with field-theoretic
quantities at the same time.

 When one of the ``type II string
theories'' is
formulated on a ten-manifold
of the form $M^{3,1}\times X$ with $X$ a Calabi--Yau
threefold
and $M^{3,1}$ a flat spacetime of dimension four, the resulting theory has
a low energy
approximation which takes the form of a
four-dimensional quantum field theory with quite realistic physical
properties
(depending on certain properties of the Calabi--Yau threefold).

Elliptic Calabi--Yau
threefolds with a section $\pi:X\to B$ have also
been used in a different way in string theory.
We can ask what happens to the type IIA theory
in the limit in which the
Calabi--Yau metric on $X$ is varied so that the fibers of the map $\pi$
shrink to zero area and the string coupling approaches infinity.  
It turns out that the resulting physical theory has a
low energy approximation which takes the form of a \emph{six-dimensional}
quantum field theory.  This limiting theory can also be described more
directly in terms of the periods $\tau(b)$ of the elliptic curves
$\pi^{-1}(b)$, regarded as a multi-valued function on $B$.  The type IIB
string theory is compactified on $B$ with the aid of this function, using
what are known as D-branes along the discriminant locus of the map $\pi$.
(This latter approach is known as ``F-theory.'')

This six-dimensional quantum field theory includes gravity as well as
 a gauge field theory whose gauge group
is the compact reductive group $G$  defined
in Section \ref{gaugealgebra}.

The gauge theory is on a eight-dimensional manifold  with boundary, $Y$, whose boundary
is $M^{3,1} \times S^1$; $Y$ is equipped  with a principal $G$-bundle
$\mathcal{G}$ (the
``gauge bundle'').
Then the curvature $F$ of the gauge connection
is an $\ad(\mathcal{G})$-valued
two-form, where each fiber $\ad(\mathcal{G})_x$ of $\ad(\mathcal{G})$ is
isomorphic to the Lie
algebra $\mathfrak{g}$ of $G$, with $\mathcal{G}_x$ acting on
$\ad(\mathcal{G})_x$ via the adjoint action of $G$ on $\mathfrak{g}$.
Similarly,  $Y$ is equipped
with a (pseudo-)Riemannian metric, and the curvature $R$ of the
Levi--Civita connection is a two-form taking values in the endomorphisms
of the tangent bundle. (Notations and background can be found in Appendix \ref{app:A}.)

  In
order to be a consistent quantum theory, the ``anomalies'' of this
theory must vanish.
In particular, Schwarz shows \cite{Schwarz:1995zw} that in these model, 
(${\mathcal{N}}{=}1$ supersymmetric theories in
six dimensions
with a reductive gauge group  $G$), the anomalies are characterized by an eight-form, made from
curvatures and gauge field two-forms.

\vskip 0.1in

We now discuss the content of each Section of the paper,
although not in the order they actually appear:

Sections \ref{sec:Chow}, \ref{1}, and \ref{2} are the pivot around which the paper unfolds: we describe in Section~\ref{sec:Chow} the anomaly cancellations coming from physics and recast them in the geometric language of codimension two  cycles and representations in the F-theory set up.  This is summarized in equations
\eqref{condition1}--%
\eqref{condition4} and expressed in geometric terms in Sections \ref{1} and \ref{2}. One upshot is that through the anomalies the representation theory puts constraints on the possible F-theory geometries and  vice versa.

In the following Sections \ref{sec:examples} and \ref{sec:about} we describe in concrete terms how this works for the standard
``generic'' codimension two singularities of elliptic fibrations
from \cite{geom-gauge}
(as already verified in \cite{grouprep}) as well as for some new
codimension two singularities such as the one from \cite{newTate}.
We also introduce some singularities which are considered here
for the first time.  In the sections preceding Section~\ref{sec:Chow} we introduce the language and the concepts which are necessary to properly state and understand the geometric interpretations of the anomaly.

Some key features of the anomaly cancellation are phrased in terms
of the Casimir operators of a (real) reductive Lie algebra $\mathfrak{g}$. These are discussed in Section \ref{sec:Casimir}. Motivated by the anomalies we say that two representations $\rho$ and $\rho'$ are
{\em Casimir equivalent in degrees $2$ and $4$}\/ if
  $\Tr_\rho F^k = \Tr_{\rho'} F^k$ for $k=2$, $4$.  In fact, since
the anomalies we consider only involve these degree $2$ and $4$
quantities, we can freely replace a representation with a Casimir-equivalent
one without affecting the validity of anomaly cancellation,
as described in Table~\ref{tab:casimir-equivalences}.
Several examples of this equivalence will be important for us in this paper.
In the following Section \ref{gaugealgebra} we explain how a gauge algebra can be naturally associated to an elliptic fibration with section, via its Weierstrass equation, Kodaira's classification of singular fibers and Tate's algorithm; this is summarized in Table~\ref{tab:kodaira}.

An elliptic fibration with a base of dimension at least two also
has a ``matter representation''  (i.e., the matter representation appropriate
to F-theory or M-theory), with contributions  from the components
of the discriminant locus
of the elliptic fibration and  from the singular locus  of
the discriminant locus $\Sigma$. The matter representation for six-dimensional theories
gets contributions from rational curves which are components of fibers
in the Calabi--Yau resolution of a Weierstrass elliptic fibration. We describe how to compute this in Section \ref{sec:matter} and introduce the \emph{virtual matter cycle}, which is an element of the Chow group
with coefficients in the representation ring of the gauge algebra
$\mathfrak{g}$; the representations are derived from the branching rules. 
(In
Section~\ref{2}
we show that anomaly cancellation for these
six-dimensional theories follows from a relation among algebraic cycles.  If the virtual matter cycle is rationally equivalent to another cycle, the \emph{Tate cycle}, and some conditions on the
intersections  of the components of the discriminant locus hold, then the anomalies vanish; this is actually verified in last two Sections of the paper.)

 The \emph{Tate cycle} is introduced in equation \ref{tatecycle} in Section \ref{sec:Tatecycle}. The constituents of this cycle are found in  Tables \ref{tab:divisors} and \ref{tab:reps}: the first are derived from the Tate algorithm, while the latter are \emph{mostly} derived from the branching rules, except some which are substituted by their Casimir equivalent: this substitution is crucial for the verification of  anomaly cancellation.
In fact, in Section \ref{sec:about}
for each type of singular point, we compute the contribution to the
Tate cycle and compare it to the contribution to the matter cycle.

We wish to emphasize that the Tate cycle is very easy to compute from
the Weierstrass equation, and gives a very quick route to determining
the matter representation up to Casimir equivalence.  Existing techniques
for determining the matter representation more precisely are much
more complicated, and involve either making a group-theory analysis
of each singular point (as in \cite{Katz:1996xe}),
or constructing a resolution of singularities
explicitly (as in \cite{matter1,Esole:2011sm}).

\section{Casimir operators}\label{sec:Casimir}

Some key features of the anomaly cancellation are phrased in terms
of the Casimir operators of a (real) reductive Lie algebra $\mathfrak{g}$
of compact type.
Recall that these are complex-valued polynomial functions on
$\mathfrak{g}$ which are invariant under the adjoint action of $\mathfrak{g}$
on itself.
By the basic structure theory for reductive Lie
algebras (cf.~\cite{MR0027270,MR0044515}), if we choose
a Cartan subalgebra $\mathfrak{h}\subset \mathfrak{g}_{\mathbb{C}}$
of the complexification $\mathfrak{g}_{\mathbb{C}}$ of $\mathfrak{g}$,
then the Casimir operators can be identified with elements of the ring
$S(\mathfrak{h}^*)^W$, where $S(\mathfrak{h}^*)$ is the ring of
polynomial functions
on $\mathfrak{h}$, and $W$ is the Weyl group of $\mathfrak{g}$.

Useful examples of Casimir operators are given by $F\mapsto
\operatorname{Tr}_\rho F^k$  where
$\rho$ is a finite-dimensional complex
representation of $\mathfrak{g}$,
$k$ is a positive integer,
and
$F\in\mathfrak{g}$.
The anomalies of the six-dimensional theories we study only
involve Casimir operators of degrees $2$ and $4$; examples of the latter
are provided by $F\mapsto
(\operatorname{Tr}_\rho F^2)^2$ and $F\mapsto
\operatorname{Tr}_\rho F^4$.  We often shorten the notation and
simply refer to these operators as $(\operatorname{Tr}_\rho F^2)^2$
and $\operatorname{Tr}_\rho F^4$, respectively.

For any simple non-abelian
Lie algebra $\mathfrak{g}$, the space of Casimir operators
of degree $2$ is one-dimensional, spanned by $\Tr_\rho F^2$ for any
nontrivial representation $\rho$.

Our main result will be formulated with the aid of a specific basis
for this one-dimensional space, given by a particular
Casimir operator of degree $2$  which we will denote by
$F\mapsto \operatorname{tr} F^2$, or simply $\operatorname{tr} F^2$
(using lowercase $\operatorname{tr}$ to distinguish this particular
``normalized'' trace).
The normalization we use is a natural one, introduced in
\cite{MR0459426} (see also \cite{global6D}): if $G$ is the compact simple simply-connected
Lie group whose  Lie algebra is $\mathfrak{g}$, and
if we choose an appropriate generator of $\pi_3(G)\cong \mathbb{Z}$,
then for any group homomorphism $\varphi:\operatorname{SU}(2)\to G$
whose homotopy class is $N$ times the generator, the integral
\[ \frac1{32\pi^2}\int \tr \left(d\varphi(F)\wedge d\varphi(F)\right)\]
evaluates to $N$ when $F$ is the curvature of (any) four-dimensional
$\operatorname{SU}(2)$-instanton with instanton number $1$,
and $d\varphi: \mathfrak{su}(2)\to \mathfrak{g}$ is the induced
Lie algebra homomorphism.
The choice of generator of $\pi_3(G)$ used in this definition (i.e., the sign)
is determined by the additional
requirement that $\Tr_\rho F^2$ should be a positive multiple of
$\tr F^2$ for any nontrivial irreducible representation $\rho$.

\begin{table}
\begin{center}
\begin{tabular}{|c|ccccc|} \hline
$\mathfrak{g}$ & $\mathfrak{e}_6$ & $\mathfrak{e}_7$ & $\mathfrak{e}_8$ &
$\mathfrak{f}_4$ & $\mathfrak{g}_2$ \\ \hline
$\tr F^2$ & $\frac16\Tr_{\mathbf{27}}F^2$ & $\frac1{12}\Tr_{\mathbf{56}}F^2$
& $\frac1{60}\Tr_{\operatorname{adj}}F^2$ & $\frac16\Tr_{\mathbf{26}}F^2$
& $\frac12\Tr_{\mathbf{7}}F^2$
\\
\hline
\end{tabular}
\end{center}
\smallskip
\caption{Normalized traces for the exceptional algebras}\label{tab:normal-trace}
\end{table}

This ``normalized trace'' has been determined for all of the simple
non-abelian algebras,
in terms of traces in familiar representations.
For the exceptional simple algebras,
the normalized trace is described
in Table~\ref{tab:normal-trace}
in terms of the trace in the
irreducible representation of smallest dimension.
(We label irreducible representations of the exceptional
simple algebras by their dimension, in bold face.)
For $\mathfrak{su}(m)$ (resp.\
$\mathfrak{sp}(n)$), we have $\tr F^2= \Tr_{\operatorname{fund}} F^2$
in terms of
the fundamental representation ``$\operatorname{fund}$'', i.e., the standard
representation of complex dimension $m$ (resp.\ $2n$).
For $\mathfrak{so}(\ell)$, we have
$\tr F^2= \frac12 \Tr_{\operatorname{vect}} F^2$ in terms of the vector
representation ``$\operatorname{vect}$'' (i.e., the standard real
representation of real dimension $\ell$).
In addition,
for any positive integer
$k$ we define
 $\tr F^k = \Tr_{\operatorname{fund}} F^k$ when
$\mathfrak{g}\cong \mathfrak{su}(m)$ or
$\mathfrak{sp}(n)$, and
$\tr F^k= \frac12 \Tr_{\operatorname{vect}} F^k$
when $\mathfrak{g}\cong\mathfrak{so}(\ell)$; this agrees with the facts
about $\tr$ stated above in the case $k=2$.

\begin{table}[ht]
\begin{center}

\begin{tabular}{|c|c|c|c|} \hline
 $\mathfrak{g}$ & $\rho$ & $\Tr_\rho F^2$ & $\Tr_\rho F^4 $ \\ \hline
$\mathfrak{su} (m) $,& $\adj$&$2m\tr F^2$&$(m+6)(\tr F^2)^2  $ \\
$m=2$, $3$         &      $\fund$ &$\tr F^2$& $\frac12(\tr F^2)^2$ \\ \hline
 $\mathfrak{su} (m) $,& $\adj$ &$2m\tr F^2$&$6(\tr F^2)^2+2m\tr F^4  $ \\
$m\ge4$         &      $\fund$ &$\tr F^2$& $0(\tr F^2)^2+\tr F^4$ \\
     & $\Lambda^2$ &$(m-2)\tr F^2$&$3(\tr F^2)^2+(m-8)\tr F^4  $  \\
     & $\Lambda^3$ &$\frac12(m^2{-}5m{+}6)\tr F^2$&$(3m{-}12)(\tr F^2)^2+
\frac12(m^2{-}17m{+}54)\tr F^4  $  \\
 \hline
$\mathfrak{su}(8)$ & $\Lambda^4$ & $20\tr F^2$ & $18(\tr F^2)^2 - 16\tr F^4$ \\
\hline
 $\mathfrak{sp} (n) $,&$\adj$ &$(2n+2)\tr F^2$&$3(\tr F^2)^2+(2n+8)\tr F^4$\\
     $n\ge2$          &$\fund$ &$\tr F^2$&$0(\tr F^2)^2+\tr F^4$\\
     &$\Lambda^2_{\text{irr}}$ &$(2n-2)\tr F^2$&$3 (\tr F^2)^2+(2n-8)\tr F^4  $  \\
     &$\Lambda^3_{\text{irr}}$ &$(2n^2{-}5n{+}2)\tr F^2$&
$(6n{-}12) (\tr F^2)^2+(2n^2{-}17n{+}26)\tr F^4  $  \\
\hline
$\mathfrak{sp}(4)$ & $\Lambda^4_{\text{irr}}$ & $14\tr F^2$ & $15(\tr F^2)^2 - 16 \tr F^4$ \\ \hline
$\mathfrak{so}(\ell)$, & $\adj$ &$(2\ell-4)\tr F^2$&$12(\tr F^2)^2+(2\ell-16)\tr F^4  $ \\
$\ell\ge7$,       & $\vect$ &$2\tr F^2$&$0(\tr F^2)^2+2 \tr F^4$ \\
$\ell\ne8$        & $\spinrep_*$ & $\dim(\spinrep_*)(\frac14 \tr F^2)$ &
        $\dim(\spinrep_*)(\frac3{16} (\tr F^2)^2-\frac18 \tr F^4 ) $ \\ \hline
$\mathfrak{so}(8)$ & $\adj$ &$12\tr F^2$
  &$4 \Tr_{\vect} F^4 + 4 \Tr_{\spinrep_+} F^4 + 4 \Tr_{\spinrep_-} F^4$ \\
       & $\vect$ &$2\tr F^2$&$ \Tr_{\vect} F^4$ \\
        & $\spinrep_+$ & $2 \tr F^2$  &$ \Tr_{\spinrep_+} F^4$ \\
        & $\spinrep_-$ & $2 \tr F^2$  &$ \Tr_{\spinrep_-} F^4$ \\ \hline
 $\mathfrak{e}_6$ & $\adj$ &$24\tr F^2$&$18(\tr F^2)^2$ \\
       & $\mathbf{27}$ &$6\tr F^2$&$3 (\tr F^2)^2$\\ \hline
 $\mathfrak{e}_7$ & $\adj$ &$36\tr F^2$&$24(\tr F^2)^2$\\
       & $\mathbf{56}$ &$12\tr F^2$&$6 (\tr F^2)^2$\\ \hline
 $\mathfrak{e}_8$ & $\adj$ &$60\tr F^2$&$36(\tr F^2)^2$\\ \hline
 $\mathfrak{f}_4$ & $\adj$ &$18\tr F^2$&$15(\tr F^2)^2$\\
       & $\mathbf{26}$ &$6\tr F^2$&$3 (\tr F^2)^2$\\ \hline
 $\mathfrak{g}_2$ & $\adj$ &$8\tr F^2$&$10(\tr F^2)^2$ \\
        & $\mathbf{7}$ &$2\tr F^2$&$(\tr F^2)^2$ \\ \hline

\end{tabular}

\end{center}
\smallskip
\caption{Casimir operators of degrees $2$ and $4$.
}
\label{tab:D}
\end{table}

For $\mathfrak{g}= \mathfrak{su}(2)$, $\mathfrak{su}(3)$, $\mathfrak{g}_2$,
$\mathfrak{f}_4$, $\mathfrak{e}_6$, $\mathfrak{e}_7$, or $\mathfrak{e}_8$,
the space of Casimir operators of degree $4$ is one-dimensional, and
$(\tr F^2)^2$ provides a natural basis.  For all other simple
algebras except
$\mathfrak{so}(8)$, the space of Casimir operators of degree $4$ is
two-dimensional, with $(\tr F^2)^2$ and $\tr F^4$ providing
a basis.  For $\mathfrak{so}(8)$, the space of Casimir operators
of degree $4$ is three-dimensional, and we can take
$\Tr_{\vect} F^4$, $\Tr_{\spinrep_+} F^4$,
and $\Tr_{\spinrep_-} F^4$ as a basis for this space, where
$\vect$, $\spinrep_+$ and $\spinrep_-$ are the three $8$-dimensional
representations of $\mathfrak{so}(8)$ (permuted by
triality). The key relation in this latter case is
\[ (\tr F^2)^2 = \frac13\Tr_{\vect} F^4+\frac13\Tr_{\spinrep_+} F^4
+\frac13\Tr_{\spinrep_-} F^4\]
(verified in Appendix~\ref{app:Casimir}).

For any
representation $\rho$ of
a simple algebra
$\mathfrak{g}$, we
can then express $\Tr_\rho F^2$ as a multiple of $\tr F^2$, and
$\Tr_\rho F^4$ as a linear combination of the basis elements chosen above.
For the representations we consider here,
much of this data was worked out by Erler \cite{Erler:1993zy}, based
in part on earlier work of van Nieuwenhuizen \cite{MR1057343}.
We collect the information we need\footnote{We have
included a complete derivation of the data for $\mathfrak{so}(\ell)$
in Appendix~\ref{app:Casimir}, including the case $\ell=8$, since the special
features of that case seem to have been overlooked in the literature.
} in
Table~\ref{tab:D}.  In addition to the adjoint representation, and
the fundamental and vector representations mentioned above, our
analysis includes the anti-symmetric powers $\Lambda^k$ of
the fundamental representation in the case of $\mathfrak{su}(m)$,
the nontrivial  irreducible component $\Lambda^k_{\operatorname{irr}}$
of the $k^{\text{th}}$ anti-symmetric power
of the fundamental representation in the
case of $\mathfrak{sp}(n)$,
and the spinor representations $\spinrep_*$ in the case of $\mathfrak{so}(\ell)$.
(We use the notation $\spinrep_*$ to denote either
the unique spinor representation ``$\spinrep$'' when $\ell$ is odd,
or either of the two spinor representations ``$\spinrep_\pm$'' when $\ell$
is even.)

We say that two representations $\rho$ and $\rho'$ are
{\em Casimir equivalent in degrees $2$ and $4$}\/ if
  $\Tr_\rho F^k = \Tr_{\rho'} F^k$ for $k=2$, $4$.  Since the
the anomalies we consider only involve these degree $2$ and $4$
quantities, we can freely replace a representation with a Casimir-equivalent
one without affecting the validity of anomaly cancellation.
Two examples of this equivalence will be important for us in this paper.
First, for $\mathfrak{su}(m)$, $m\ge6$,
the representation $\rho_m=\Lambda^3 \oplus (\fund)^{\oplus (m^2-7m+10)/2}$
has invariants
\begin{align*}
\Tr_{\rho_m} F^2
&= \left( \frac12(m^2-5m+6) + \frac12(m^2-7m+10) \right) \tr F^2
= (m-4)(m-2)\tr F^2\\
\Tr_{\rho_m} F^4
&= \left( (3m-12) + 0 \right) (\tr F^2)^2
+ \left( \frac12(m^2-17m+54) + \frac12(m^2-7m+10)\right) \tr F^4\\
&= 3(m-4) (\tr F^2)^2 + (m-4)(m-8) \tr F^4
\end{align*}
which are the same as for $(\Lambda^2)^{\oplus (m-4)}$.
Similarly, for $\mathfrak{sp}(n)$, $n\ge3$, the
representation
$\rho_n = \Lambda^3_{\text{irr}} \oplus (\fund)^{\oplus (2n^2-7n+6)}$
has invariants
\begin{align*}
\Tr_{\rho_n} F^2
&= \left( (2n^2-5n+2) + (2n^2-7n+6) \right) \tr F^2
= (2n-4)(2n-2) \tr F^2\\
\Tr_{\rho_n} F^4
&= \left( (6n-12) + 0\right) (\tr F^2)^2
+ \left( (2n^2-17n+26) + (2n^2-7n+6)\right) \tr F^4\\
&= 3(2n-4) (\tr F^2)^2 + (2n-4)(2n-8) \tr F^4
\end{align*}
which are the same as for $(\Lambda^2_{\text{irr}})^{\oplus (2n-4)}$.
We summarize these equivalences for low values of $m$ and $n$ in
Table~\ref{tab:casimir-equivalences}, as well as some analogous
equivalences for $\Lambda^4$.

\begin{table}[ht]
\begin{center}

\begin{tabular}{|c|c|} \hline
$\mathfrak{g}$ & \\ \hline
$\mathfrak{su}(6)$ & $\Lambda^3 \oplus \fund^{\oplus2 } \sim (\Lambda^2)^{\oplus2}$ \\
$\mathfrak{su}(7)$ & $\Lambda^3 \oplus \fund^{\oplus5 } \sim (\Lambda^2)^{\oplus3}$ \\
$\mathfrak{su}(8)$ & $\Lambda^3 \oplus \fund^{\oplus9 } \sim (\Lambda^2)^{\oplus4}$ \\
$\mathfrak{su}(8)$ & $\Lambda^4 \oplus \fund^{\oplus16 } \sim (\Lambda^2)^{\oplus6}$ \\ \hline
$\mathfrak{sp}(3)$ & $\Lambda^3_{\text{irr}} \oplus \fund^{\oplus3} \sim (\Lambda_{\text{irr}}^2)^{\oplus2}$ \\
$\mathfrak{sp}(4)$ & $\Lambda^3_{\text{irr}} \oplus \fund^{\oplus10 } \sim (\Lambda_{\text{irr}}^2)^{\oplus4}$ \\
$\mathfrak{sp}(4)$ & $\Lambda^4_{\text{irr}} \oplus \fund^{\oplus16 } \sim (\Lambda_{\text{irr}}^2)^{\oplus5}$ \\
\hline
\end{tabular}
\end{center}
\smallskip
\caption{Casimir equivalences}\label{tab:casimir-equivalences}

\end{table}

There is also a useful quartic (or biquadratic) Casimir operator
when $\mathfrak{g}=\mathfrak{g}_L\oplus\mathfrak{g}_R$ is the
direct sum of two simple non-abelian Lie algebras, namely,
the product
$\tr_{\mathfrak{g}_L}(F^2_{\mathfrak{g}_L}) \tr_{\mathfrak{g}_R}(F^2_{\mathfrak{g}_R})$
of the normalized traces.  This can be compared with other
possible biquadratic Casimir operators (determined from representations)
in the following way.
Let $\rho$ be a representation of $\mathfrak{g}_L\oplus
\mathfrak{g}_R$.  If $\rho$ is irreducible, then there exist irreducible
representations $\rho_L$ of $\mathfrak{g}_L$ and $\rho_R$ of
$\mathfrak{g}_R$ such that $\rho=\rho_L\otimes\rho_R$.  In this case,
we define the {\em representation-multiplicity
of $\mathfrak{g}_L$ and $\mathfrak{g}_R$ at $\rho$}\/ to be
\begin{equation} \label{eq:repmult}
 \mu_\rho(\mathfrak{g}_L,\mathfrak{g}_R) =
\frac{\Tr_{\rho_L}(F^2_{\mathfrak{g}_L}) \Tr_{\rho_R}(F^2_{\mathfrak{g}_R})}
{\tr_{\mathfrak{g}_L}(F^2_{\mathfrak{g}_L}) \tr_{\mathfrak{g}_R}(F^2_{\mathfrak{g}_R})}.
\end{equation}
By the way we defined the normalized traces for $\mathfrak{g}_L$ and
$\mathfrak{g}_R$, $\mu_\rho(\mathfrak{g}_L,\mathfrak{g}_R)$
 is always a nonnegative integer.  It is zero if
one of $\rho_L$ and $\rho_R$ is trivial.

Still assuming that $\rho$ is irreducible, we can write a formula
for $\mu_\rho(\mathfrak{g}_L,\mathfrak{g}_R)$ directly in terms of
$\rho$, using the fact that $\Tr_{\rho_L\otimes\rho_R}(F^2_{\mathfrak{g}_L})
= \dim (\rho_R) \Tr_{\rho_L}(F^2_{\mathfrak{g}_L})$ and the fact
that $\dim(\rho) = \dim(\rho_L)\dim(\rho_R)$.  The formula is:
\begin{equation} \label{eq:repmultbis}
 \mu_\rho(\mathfrak{g}_L,\mathfrak{g}_R) =
\frac1{\dim{\rho}} \, \frac{\Tr_{\rho}(F^2_{\mathfrak{g}_L}) \Tr_{\rho}(F^2_{\mathfrak{g}_R})}
{\tr_{\mathfrak{g}_L}(F^2_{\mathfrak{g}_L}) \tr_{\mathfrak{g}_R}(F^2_{\mathfrak{g}_R})}.
\end{equation}
We extend this definition linearly, to arbitrary representations
of $\mathfrak{g}_L\oplus\mathfrak{g}_R$, and also to the representation
ring with $\mathbb{Q}$-coefficients
$R(\mathfrak{g}_L\oplus\mathfrak{g}_R)\otimes\mathbb{Q}$.

Note that the formula \eqref{eq:repmultbis} scales linearly
for multiples of an irreducible representation $\rho$:  for
$\rho^{\oplus k}$, each trace in the numerator is multiplied by
$k$, and the dimension in the denominator is also multiplied by
$k$, giving an overall scaling by $k$.

\section{The gauge algebra}\label{gaugealgebra}

We consider a nonsingular
elliptically fibered complex manifold $\pi: X \to B$, where $B$ is smooth, and  with a section; we denote by $E$ the general elliptic fiber of $\pi$.
Without loss of generality, we can assume that $\pi: X \to B$ is a resolution
of singularities
of a Weierstrass model $W\to B$ \cite{MR977771, alg-geom/9704008}.
$W$ can be described  by a Weierstrass
equation which locally can be written as:
\begin{equation}\label{W} y ^2= x^3 + fx +g, \end{equation}
where $f$ and $g$ are sections of appropriate line bundles\footnote{In fact, 
$\mathcal{O}_B(L)$ is the relative dualizing sheaf of the fibration,
and can be identified with the (extended) Hodge bundle.}
$\mathcal{O}_B(4L)$ and $\mathcal{O}_B(6L)$, respectively, on the base $B$.
We denote by $\Delta \subset B$ the ramification locus of $\pi$; $\Delta$ has codimension one, and it is defined by the equation $4f^3+27g^2$ (using
the standard conventions of the F-theory literature).

We are primarily interested in the case in which $X$ is a Calabi--Yau
manifold; this happens when $L=-K_B$ is the anticanonical bundle on the
base $B$.  Since much of our analysis can be formulated without making
this assumption, we shall do so, and only assume $L=-K_B$ when
strictly necessary.

We also assume that the fibration $X \to B$ is \emph{equidimensional}. It is known
in the Calabi--Yau case
\cite{alg-geom/9305003}
that at least for a base of dimension two, any Weierstrass
fibration can be partially  desingularized to a Calabi--Yau variety with
only $\mathbb{Q}$-factorial terminal
 singularities\footnote{Terminal singularities, which include
the familiar ordinary double points in dimension three occurring on conifolds,
are a natural class of singularities in birational geometry.
See for example \cite{nonspherI, Denef:2005mm} in the physics literature.}
in such a way that the  elliptic fibration  of the Calabi-Yau variety is
 equidimensional.%
\footnote{Note that there are examples with a base of dimension three
in which the fibration cannot be made equidimensional while preserving the
Calabi--Yau
condition on the total space \cite{codimthree,Denef:2005mm}.  The physics
in these cases is not completely understood.}
We are assuming here that this equidimensional Calabi--Yau variety is actually
\emph{nonsingular}, which is a nontrivial assumption. In fact,
there are $\mathbb{Q}$-factorial terminal singularities
which do not admit any (locally) Calabi--Yau desingularization, but
these  do not have a direct physical interpretation.

In addition, we assume that at the generic point $P$ of each codimension two
subvariety $\Gamma$ of $B$,
if we restrict the fibration to a
general local curve $C$
passing through $P$ and transverse to each component of the discriminant
 containing $\Gamma$,
then it is
a  \emph{minimal}\/ Weierstrass fibration.\footnote{A two-dimensional 
Weierstrass fibration is 
\emph{minimal at $P$} if one of the following conditions holds:
$\operatorname{ord}_{u=0}(f) <4$, or
$\operatorname{ord}_{u=0}(g) <6$, or
$\operatorname{ord}_{u=0}(4f^3+27g^2)<12$,
where $P$ is locally defined by the equation $u=0$.}  
Non-minimal cases correspond
to extremal transition points in the moduli space \cite{FCY2}, and
 involve light strings in the spectrum.  They can be analyzed from
the point of view of either of the two branches of the moduli space
which are coming together at the extremal transition point.
As a consequence, we do not lose anything essential by excluding them
from the present analysis.

\vskip 0.1in

A few remarks are in order: the existence of a section is a mild assumption.
In fact, if  $X$ is an elliptic fibration without a section, we can consider the associated Jacobian fibration (with section)
 $\pi_{\mathcal{J}}: \mathcal{J}(X) \to B$, where $\mathcal{J}(X)$ is still a
 Calabi-Yau threefold \cite{MR0387292,MR1242006,MR1272978}.
Then the only assumption is again that $\mathcal{J}(X)$  is smooth; in fact $\mathcal{J}(X)$ could  in principle have
  terminal (and not smooth) singularities, even if $B$ is smooth;
  however, we do not know of
 any such example.

In  the first  four columns of Table \ref{tab:kodaira} we list  Kodaira's classification of singular fibers in  the smooth resolution of a two-dimensional Weierstrass fibration.
Note that if we restrict a higher-dimensional fibration to a
general local curve $C$
passing through $P$ and transverse to each component of the discriminant
locus  containing $\Gamma$, the orders of vanishing of $f$, $g$ and $\Delta$ 
at $P$ do not necessarily determine the singular fiber in the original 
higher-dimensional fibration (cf.\ \cite{matter1,Esole:2011sm}).

\vskip 0.1in

Each F-theory compactification on an elliptically fibered manifold
$\pi:X\to B$
has an associated gauge group, which
can be
determined\footnote{For Calabi--Yau fourfolds, we are describing
the gauge group in the absence of flux.  This ``geometric'' gauge group
is determined in exactly the same way as for K3 surfaces or for Calabi--Yau
threefolds, although the actual gauge group may be smaller since some
of the gauge symmetry may be broken by flux.  We shall not discuss that
aspect of those models in this paper.}
 by compactifying on a circle and analyzing the
corresponding singular Calabi--Yau variety in M-theory.  The gauge
group is a compact reductive group $G$ whose component group
$\pi_0(G)$ coincides with the Tate--Shafarevich group $\sha_{X/B}$
of the fibration \cite{triples}, whose fundamental group
$\pi_1(G)$ coincides with the Mordell--Weil group $\operatorname{MW}(X/B)$
of the fibration \cite{pioneG}, and whose Lie algebra $\mathfrak{g}_{X/B}$
is determined by the singular fibers in codimension one as
described below.

Since the gauge algebra is a reductive Lie algebra,
it can be written as the direct sum
of an abelian Lie algebra and a finite number of simple Lie algebras.
The abelian part is given by $\pi_1(G)\otimes \mathbb{R}$, and so its
dimension coincides with the rank of the Mordell--Weil group.
In this paper we will only consider elliptic fibrations
$\pi:X\to B$ whose Mordell--Weil
group has rank $0$.   We denote the gauge algebra by 
$\mathfrak{g}_{X/B} = \bigoplus_\Sigma \mathfrak{g}(\Sigma)$
where the sum is taken over the components of the discriminant locus,
 or simply by $\mathfrak{g}$ if the meaning if clear.  There are some additional anomaly cancellation
conditions which must hold when there is a nontrivial abelian part of
the gauge algebra, and to simplify our discussion we will
not consider those here.\footnote{In general, one would need to 
add the abelian Lie algebra $\operatorname{MW}(X/B)\otimes\mathbb{R}$
to $\bigoplus_\Sigma \mathfrak{g}(\Sigma)$ in order to obtain the
full gauge algebra.}

Each summand of
the non-abelian part of the gauge algebra in an F-theory compactification
is associated to a component $\Sigma$ of the discriminant locus, and the gauge
algebra summand $\mathfrak{g}(\Sigma)$
for any given component depends on the generic type
of singular fiber of the Weierstrass model
along that component as well as the monodromy of the components of
the exceptional divisor when the singular fibers are resolved.  (We sometimes
denote the summand associated to a single component by $\mathfrak{g}$ when
there is no danger of confusion.) 
The singular fibers were classified by Kodaira \cite{MR0184257},
whereas the monodromy was determined in a systematic way by
Tate \cite{MR0393039} as part of
what is called Tate's algorithm.\footnote{In fact,
Tate was interested in a more general situation which also included
elliptic fibrations over fields of characteristic $p$.  Since we
only need to consider characteristic $0$ the algorithm we present here
is a slightly simplified version of the one given by Tate.}
Tate's algorithm has been discussed in the physics literature before
\cite{geom-gauge}, but the conditions which determine monodromy
were not
spelled out in full generality.

The Kodaira classification and Tate's algorithm only depend on
the generic behavior of the elliptic fibration along a particular component
$\Sigma$ of
the discriminant locus, and describe the behavior of the fibration
near that component.
In a sufficiently small open subset of the base $B$, the discriminant
component $\Sigma$ can be described as $\{z=0\}$;
the relevant data for the algorithm will then be (i) the orders of vanishing
along $\Sigma=\{z=0\}$ of the coefficients $f$ and $g$ in the Weierstrass
equation $ y^2=x^3+fx+g$ and of the discriminant $\Delta=4f^3+27g^2$, and (ii)
the quantities
\[
\left.\frac{f}{z^{\operatorname{ord}_\Sigma(f)}}\right|_{z=0}, \quad
\left.\frac{g}{z^{\operatorname{ord}_\Sigma(g)}}\right|_{z=0}, \quad
\text{and} \quad
\left.\frac{\Delta}{z^{\operatorname{ord}_\Sigma(\Delta)}}\right|_{z=0},
\]
which can interpreted  as generically defined meromorphic
sections of appropriate line bundles on $\Sigma$.

\begin{table}
{\footnotesize
\begin{center}
\begin{tabular}{|c|c|c|c|c|c|} \hline
&$\operatorname{ord}_\Sigma(f)$&$\operatorname{ord}_\Sigma(g)$
&$\operatorname{ord}_\Sigma(\Delta)$
&eqn.\ of monodromy cover&$\mathfrak{g}(\Sigma)$\\ \hline\hline
$I_0$&$\ge0$&$\ge0$&$0$& --
& -- \\ \hline
$I_1$&$0$&$0$&$1$&
-- & -- \\ \hline
$I_2$&$0$&$0$&$2$&
-- & $\mathfrak{su}(2)$ \\ \hline
$I_m$, $m\ge3$&$0$&$0$&$m$&$\psi^2+(9g/2f)|_{z=0}$
& $\mathfrak{sp}(\left[{\frac m2}\right])$ or $\mathfrak{su}(m)$\\ \hline
$II$&$\ge1 $&$   1  $&$    2 $&  --
& -- \\ \hline
$III$&$  1 $&$   \ge2 $&$   3 $& --
& $\mathfrak{su}(2)$ \\ \hline
$IV$&$ \ge2 $&$  2  $&$    4 $&$\psi^2-(g/z^2)|_{z=0}$
& $\mathfrak{sp}(1)$ or $\mathfrak{su}(3)$  \\ \hline
$I_0^*$&$\ge2$&$\ge3$&$6$&
$\psi^3+(f/z^2)|_{z=0}\cdot\psi+(g/z^3)|_{z=0}$
& $\mathfrak{g}_2$ or $\mathfrak{so}(7)$ or $\mathfrak{so}(8)$
\\ \hline
$I_{2n-5}^*$, $n\ge3$&$2$&$3$&$2n+1$&$\psi^2+\frac14(\Delta/z^{2n+1})(2zf/9g)^3|_{z=0}$
& $\mathfrak{so}(4n{-}3)$ or $\mathfrak{so}(4n{-}2)$ \\ \hline
$I_{2n-4}^*$, $n\ge3$&$2$&$3$&$2n+2$&$\psi^2+(\Delta/z^{2n+2})(2zf/9g)^2|_{z=0}$
& $\mathfrak{so}(4n{-}1)$ or $\mathfrak{so}(4n)$ \\ \hline
$IV^*$&$\ge3$&$  4$  &$  8$& $\psi^2-(g/z^4)|_{z=0}$
& $\mathfrak{f}_4$ or $\mathfrak{e}_6$ \\ \hline
$III^*$&$  3 $&$   \ge5 $&$   9 $& --
& $\mathfrak{e}_7$ \\ \hline
$II^*$&$ \ge4$&$   5   $&$   10 $& --
& $\mathfrak{e}_8$ \\ \hline
non-min.\ &$\ge4$&$\ge6$&$\ge12$&--&--\\ \hline
\end{tabular}
\end{center}
\smallskip
\caption{Kodaira--Tate classification of singular fibers, monodromy covers,
and gauge algebras}\label{tab:kodaira}
}
\end{table}

Kodaira's classification and Tate's monodromy refinement of it
are presented in Table~\ref{tab:kodaira}.  The type of singularity
is determined by the orders of vanishing of $f$, $g$,
and $\Delta$ along $\Sigma$.  The monodromy is described (in cases where
it is relevant) by defining a certain {\em monodromy cover}\/ of $\Sigma$
by means of a polynomial of degree $2$ or $3$ in an auxiliary variable
$\psi$, which is a meromorphic section of a certain line bundle over $\Sigma$;
the lines bundles are displayed
 in Table~\ref{tab:psi}.
The key question for determining the gauge algebra
is how many irreducible components this monodromy
cover has.  When the polynomial has degree $2$, this amounts to asking
whether the discriminant of the quadratic equation (which is a local meromorphic
function on $\Sigma$) has a square root or not.
One gets a smaller gauge algebra when the square root does not exist and
the monodromy cover is irreducible, and a larger gauge algebra when
the square root does exist and the monodromy cover is reducible.  Both
algebras are listed in the final column of the table.

The defining polynomial of the monodromy cover has
degree $3$ only in the case of Kodaira type $I_0^*$.
In that case, if the cover is irreducible,
the gauge algebra is $\mathfrak{g}_2$; if the cover has two components, the
gauge algebra is $\mathfrak{so}(7)$; and if the cover has three components,
the gauge algebra is $\mathfrak{so}(8)$.

\begin{table}
{
\begin{center}
\begin{tabular}{|c|c|c|c|c|c|} \hline
&$\operatorname{ord}_\Sigma(f)$&$\operatorname{ord}_\Sigma(g)$
&$\operatorname{ord}_\Sigma(\Delta)$
&line bundle for $\psi$\\ \hline\hline
$I_m$, $m\ge3$&$0$&$0$&$m$&$\mathcal{O}_\Sigma(L)$
\\ \hline
$IV$&$ \ge2 $&$  2  $&$    4 $&$\mathcal{O}_\Sigma(3L-\Sigma)$
\\ \hline
$I_0^*$&$\ge2$&$\ge3$&$6$& $\mathcal{O}_\Sigma(2L-\Sigma)$
\\ \hline
$I_{2n-5}^*$, $n\ge3$&$2$&$3$&$2n+1$&$\mathcal{O}_\Sigma(3L-(n{-}1)\Sigma)$
\\ \hline
$I_{2n-4}^*$, $n\ge3$&$2$&$3$&$2n+2$&$\mathcal{O}_\Sigma(4L-n\Sigma)$
\\ \hline
$IV^*$&$\ge3$&$  4$  &$  8$& $\mathcal{O}_\Sigma(3L-2\Sigma)$
\\ \hline
\end{tabular}
\end{center}
\smallskip
\caption{The line bundle of which $\psi$ is a meromorphic section, in
cases with a monodromy cover}\label{tab:psi}
}
\end{table}

The final row of Table~\ref{tab:kodaira}
gives the divisibility conditions which lead to a
non-minimal Weierstrass model (see Section \ref{gaugealgebra}).  In such a case, replacing $(x,y)$
by $(xz^2,yz^3)$ leads to a new Weierstrass equation
\[ (yz^3)^2 = (xz^2)^3 + (xz^2)f + g,\]
or
\[ y^2 = x^3 + x (f/z^4) + (g/z^6)\]
in which
the coefficients $(f,g)$ have been replaced by $(f/z^4,g/z^6)$
(and $\Delta$ has been replaced with $\Delta/z^{12}$).  One can
then apply the Kodaira--Tate algorithm to the new model instead.
Notice that this change affects the canonical bundle, and hence the
Calabi--Yau condition on the total space.

In Appendix \ref{app:Tate}, we present the detailed derivation
of the information in Table~\ref{tab:kodaira}.

In the case of $I^*_{m-4}$, our formulation of the monodromy
condition in terms of $f$, $g$, $\Delta$ and an appropriate power
of the local equation of $\Sigma$ does not seem to have been
stated explicitly in the literature (even in the number theory
literature, where it would also be relevant).  However, it does
seem to be known: one of the steps in Algorithm 7.5.1 of
\cite{MR1228206} appears to rely on this formulation.

Some aspects of the detailed geometry associated with these gauge
algebras are not directly visible in
Table~\ref{tab:kodaira}.  In fact, for each entry in the table,
there is both a
gauge algebra $\mathfrak{g}(\Sigma)$ and
a {\em covering algebra}\/ $\widetilde{\mathfrak{g}}(\Sigma)$,
together with an action
on $\widetilde{\mathfrak{g}}(\Sigma)$ by the monodromy group $\Gamma$
such that the fixed subspace coincides with the gauge algebra:
\[ (\widetilde{\mathfrak{g}}(\Sigma))^\Gamma=\mathfrak{g}(\Sigma).\]
The covering algebra is simply the algebra associated with the same
Kodaira type but with no monodromy.
As we will explain in the next Section, the action of
$\Gamma$ on the covering algebra preserves
both the set of simple roots
and the
 set of positive roots.
The orbits of the action on the set of positive roots of
the covering algebra $\widetilde{\mathfrak{g}}(\Sigma)$  
coincide with the positive
roots of the gauge algebra $\mathfrak{g}(\Sigma)$.

\section{The matter representation and the virtual matter cycle}\label{sec:matter}

An elliptic fibration with a base of dimension at least two also
has a ``matter representation''  (i.e., the matter representation appropriate
to F-theory or M-theory), with contributions (1) from the components $\Sigma$
of the discriminant locus, (2) from the singular locus $\Sing{\Sigma}$ of
each component $\Sigma$, and (3) from the components $\Gamma_\Sigma$
of each
 {\em residual discriminant locus,}
which is the zero-set of $\deltaSigma:=\Delta/z^{\ord_\Sigma(\Delta)}$
for any local
equation $z=0$ defining $\Sigma$ (away from its singular locus).
Note that the components $\Gamma_\Sigma$ include all pairwise
intersections of components $\Sigma\cap\Sigma'$.

If the base of the elliptic fibration has dimension two, the matter
representation is a quaternionic representation of the gauge algebra 
$\mathfrak{g}:=\mathfrak{g}_{X/B}$.
We follow the tradition of describing this matter representation by
means of a complex representation $\rho: \mathfrak{g} \to \mathfrak{gl}(V)$ which determines
the associated quaternionic representation as $\rho \oplus \overline{\rho}$.
This traditional notation has the drawback that if an irreducible
 complex representation $\tau$
is itself quaternionic, there is  no way to write it in
the form $\rho \oplus \overline{\rho}$.  We can, however, formally write
it as $\frac12\tau \oplus \frac12\overline{\tau}$ since $\tau \cong
\overline{\tau}$.  Thus, the complex matter representation should actually
be treated as an element of $R(\mathfrak{g})\otimes \mathbb{Q}$, the
representation ring with rational coefficients.
We refer to an element $\rho\in R(\mathfrak{g})\otimes\mathbb{Q}$
as a {\em pre-quaternionic representation}\/ if
$\rho \oplus \overline{\rho} \in R(\mathfrak{g})$, i.e., if
$\rho \oplus \overline{\rho}$ is an actual (quaternionic) representation,
not just a
$\mathbb{Q}$-linear combination of representations.
(One says that the corresponding six-dimensional theory
contains {\em half-hypermultiplets in representation $\rho$.})

On the other hand,
if the base of the elliptic fibration has dimension greater than two, the
physical matter representation is a complex representation of $\mathfrak{g}$
rather than a quaternionic representation.  The description we give here of
a complex representation whose irreducible constituents are all associated
with algebraic cycles of codimension two on the base is an important ingredient
in the full description of the matter representation in such cases,
which also involves a topologically twisted
gauge theory \cite{Donagi:2008ca,Beasley:2008dc}.
In particular, the geometric description we give here
specifies potential contributions to the matter representation, but an
additional quantization problem must be solved in order to determine
the multiplicity of such contributions (or even if they are present at all).

Let $\rho_{\text{matter}}$ be the complex matter representation of the gauge algebra
$\mathfrak{g}$ of a six-dimensional F-theory model.  Since the gauge fields
transform in the adjoint representation, the gauge and mixed anomalies of the theory
involve the virtual representation
\[ \rho_{\text{virtual}} := \rho_{\text{matter}} - \adj_{\mathfrak{g}} \in R(\mathfrak{g})\otimes \mathbb{Q},\]
which we call the {\em virtual matter representation}.  (We work with
$\mathbb{Q}$-coefficients here because of the possibility of
half-hypermultiplets in $\rho_{\text{matter}}$.)
This combination
is also the one which has a natural geometric interpretation.
As we will show, each irreducible component of
$ \rho_{\text{virtual}}$
is associated with an algebraic cycle of codimension two.
By pairing each component with the corresponding cycle, we get a
{\em virtual matter cycle}
\[ Z_{\text{virtual}} \in A^2(B) \otimes R(\mathfrak{g})\otimes \mathbb{Q},\]
where $A^2(B)$ denotes the Chow group of codimension two cycles on the
base $B$ modulo rational equivalence.\footnote{As we will see in the
construction, the algebraic cycle will itself have rational coefficients
in certain cases.}  To read off the actual
matter representation in six-dimensional F-theory,\footnote{In
four-dimensional F-theory models, the procedure is more complicated.}
 we just take the
degree and add a copy of the adjoint representation:
\[ \rho_{\text{matter}} = \adj_{\mathfrak{g}} \oplus
\operatorname{degree}(Z_{\text{virtual}}).\]

We will show in
section~\ref{2}
that anomaly cancellation for these
six-dimensional theories follows from a relation in this ``Chow group
with coefficients in the representation ring of $\mathfrak{g}$.''
That relation itself, though, is
not restricted to six-dimensional theories, and also holds when
the base of the elliptic fibration has dimension greater than two.

The matter representation for six-dimensional theories
gets contributions from rational curves which are components of fibers
in the Calabi--Yau resolution of a Weierstrass elliptic fibration.
The deformation spaces for such rational curves can be either one-dimensional
or zero-dimensional, and Witten \cite{WitMF} analyzed the contribution
to the matter representation in each case.  We follow Witten's analysis
for one-dimensional deformation spaces,
with some slight modifications in light of \cite{LieF}.  For the case of
zero-dimensional deformation spaces, it is easier in practice to use
the methods of \cite{Katz:1996xe} instead.

We first analyze the contributions from connected rational curves
contracted by $\varphi: X \to W$ which have nontrivial
deformations on $X$.  Such a rational curve $C$ is
rationally equivalent to a connected linear combination of
components of the fiber over a general point $P\in \Sigma$ for
some component $\Sigma$ of the discriminant locus.
That is, $C$ corresponds to a positive root $r$ of the covering algebra
$\widetilde{\mathfrak{g}}(\Sigma)$.

Now by following any closed loop on $\Sigma$, the curve $C$ may be
transported to another rational curve $C'$: this is the monodromy
action, and it comes from an action of the monodromy group $\Gamma$
on the roots of $\widetilde{\mathfrak{g}}(\Sigma)$.
If the orbit of a root $r$ under the
monodromy action contains $d$ elements, then the deformation space
of the corresponding curve $C$ will be a connected $d$-sheeted cover of $\Sigma$.

Thus, if there is no monodromy, each positive root will have $\Sigma$
itself as parameter curve.  By Witten's analysis, there is then
a contribution of $g=g(\Sigma)$ hypermultiplets in the adjoint
representation of $\mathfrak{g}(\Sigma)$ to the matter representation; here $g(\Sigma)$ is the geometric genus of $\Sigma$, that is, the genus of its normalization.

In our virtual matter representation, we subtract the vector multiplet
in the adjoint representation to obtain
$g-1=\frac12\deg((K_B+\Sigma)|_\Sigma) - \frac12 \sum_P \mu_P(\mu_P-1)$ copies  of the adjoint representation, where $\mu_P$ are the multiplicity of all the points in $\Sigma$ including the ``infinitely near points."\footnote{If $P \in B$ is a singular point of $\Sigma$ there is a sequence of blow ups $\phi= \varphi _n \circ \cdots \circ \varphi_1: \ {\tilde B} \to B$ such that $\tilde \Sigma$ the (strict) transform of $\Sigma$  is smooth around $\phi^{-1}(P)$. The infinitely near points $P$ are the points $Q$ mapped to $P$ by $\phi$.} We denote by $S$ the cycle $  \sum_P \mu_P(\mu_P-1)P$, which is supported on the singular locus of the component $\Sigma$.

The non-local contribution to the virtual matter cycle is then
defined to have two parts:
\[ \frac12((K_B+\Sigma)|_\Sigma)-S) \otimes \rho_\alpha,\]
(where $\rho_\alpha$ is the adjoint representation of $\mathfrak{g}(\Sigma)$) 
represents
the truly non-local contribution; the other part
\[ \frac12(S)\otimes \rho_\alpha,\]
is actually part of the local contribution, supported at
the singular locus of $\Sigma$.

If there is monodromy, then it turns out that there is an integer $d>1$
 such that  the
positive roots of the covering algebra can be divided into two classes:
the ones invariant under monodromy with parameter space $\Sigma$,
and the ones not invariant under monodromy with parameter
space $\widetilde{\Sigma}$ (the same parameter space for all such roots), 
which is a connected branched cover
$\widetilde{\Sigma}\to\Sigma$ of degree $d$.
Let $R$ be the ramification divisor
of
the branched cover,
and assume for simplicity that
all branching is simple.\footnote{It is easy to modify the analysis
in the case of non-simple branching, which can only occur for $d>2$.}
By the Riemann--Hurwitz formula, the genus  $\widetilde{g}:=g(
\widetilde{\Sigma})$ satisfies
\[ \widetilde{g} -1 = d(g-1) + \frac12\deg(R),\]
where $g$ is the  (geometric) genus of $\Sigma$.    Thus,
\[ \widetilde{g}-g = (d-1)(g-1) + \frac12\deg(R).\]

Witten's analysis says that each invariant root contributes
a space of dimension $g(\Sigma)-1$
to the virtual matter representation, and each non-invariant root
contributes\footnote{There is one subtlety here, in the case
of ${\widetilde{\mathfrak{g}}}(\Sigma)=\mathfrak{su}(2n+1)$ and $\mathfrak{g}(\Sigma)=
\mathfrak{sp}(n)$: some of the non-invariant roots do not appear
to have a vector multiplet in their spectrum.  However, as argued in
\cite{LieF}, the gauge algebra must be $\mathfrak{sp}(n)$ and
this implies that the ``index'' contribution to the virtual matter
representation must still have dimension $\widetilde{g}({\Sigma})-1$.}
a space of dimension $g(\widetilde{\Sigma})-1$.
To see how these contributions to the virtual matter representation
transform under the gauge algebra,
 we decompose the adjoint representation of $\widetilde{\mathfrak{g}}(\Sigma)$
as a representation of $\mathfrak{g}({\Sigma})$.

In all cases, we find 
\begin{equation} \label{eq:rho0}
 \adj_{\widetilde{\mathfrak{g}}(\Sigma)}
= \adj_{\mathfrak{g}(\Sigma)} \oplus \rho_0^{\oplus (d-1)}
\end{equation}
for some representation $\rho_0$ which is easily calculated
(see Table~\ref{tab:monodromy}).
Both kinds of roots contribute to the adjoint representation of
$\mathfrak{g}(\Sigma)$, but only the non-fixed roots contribute to the
representation $\rho_0$.  Thus, this part of the virtual matter representation
takes the form
\[
 (g-1) \cdot \adj_{\mathfrak{g}(\Sigma)} + (\widetilde{g}-g)\cdot\rho_0 =
 (g-1) \cdot \adj_{\mathfrak{g}(\Sigma)} + ((d-1)(g-1) + \frac12\deg(R))\cdot\rho_0
\]
as an element of $R(\mathfrak{g}(\Sigma))\otimes \mathbb{Q}$.
It is convenient to rewrite this as
\[ (g-1)\cdot (\adj_{\mathfrak{g}(\Sigma)} + (d-1)\cdot \rho_0) + \frac12\deg(R)\cdot\rho_0
\in R(\mathfrak{g}(\Sigma))\otimes \mathbb{Q}.\]
We then define one part of the nonlocal contribution to the
virtual matter cycle to be
\[ \frac12((K_B+\Sigma)|_\Sigma)-S) \otimes \rho_\alpha,\]
where
\begin{equation} \label{eq:rhoalpha}
\rho_\alpha := \adj_{\mathfrak{g}(\Sigma)} \oplus  \rho_0^{\oplus (d-1)},
\end{equation}
and the other part to be
\[ R   \otimes \frac12 \rho_0 + S \otimes \frac12 \rho_\alpha.\]

That is, we formally represent the first term in
this second part of the nonlocal contribution
as if each point in the ramification divisor carried a representation
of $\frac12\rho_0$.  As we will see in Section~\ref{sec:examples},
in all known cases the virtual matter representation associated to a ramification
point can be written in the form $\frac12\rho_0+\rho'$, where $\rho'\in
R(\mathfrak{g})$ is
a pre-quaternionic representation, often empty.
Thus,  this interpretation is sensible: $\rho'$ represents the truly
local matter contribution at such points.

\begin{table}
\begin{center}
\begin{tabular}{|c|c|c|c|c|} \hline
Type & $\mathfrak{g}(\Sigma)$ & $\widetilde{\mathfrak{g}}(\Sigma)$ & $d$ & $\rho_0$ \\ \hline
$I_3$ or $IV$ & $\mathfrak{sp}(1)$ & $\mathfrak{su}(3)$ & $2$ & $\fund^{\oplus 2}$ \\
$I_{2n}$, $n\ge2$ & $\mathfrak{sp}(n)$ & $\mathfrak{su}(2n)$ & $2$ & $\Lambda^2_{\text{irr}}$ \\
$I_{2n+1}$, $n\ge1$ & $\mathfrak{sp}(n)$ & $\mathfrak{su}(2n+1)$ & $2$ & $\Lambda^2_{\text{irr}}\oplus \fund^{\oplus 2}$ \\
$I_0^*$ & $\mathfrak{g}_2$& $\mathfrak{so}(8)$ & $3$ & $\mathbf{7}$ \\
$I_{m-4}^*$, $m\ge4$ & $\mathfrak{so}(2m-1)$ & $\mathfrak{so}(2m)$ & $2$& $\vect$ \\
$IV^*$ & $\mathfrak{f}_4$ & $\mathfrak{e}_6$ & $2$ & $\mathbf{26}$ \\
\hline
\end{tabular}
\smallskip
\caption{Monodromy and the representation $\rho_0$} \label{tab:monodromy}
\end{center}
\end{table}

We summarize the data about monodromy and the matter representation
$\rho_0$ in
Table~\ref{tab:monodromy}, which shows the algebras
$\mathfrak{g}(\Sigma)$ and $\widetilde{\mathfrak{g}}(\Sigma)$, the degree $d$ of
the monodromy cover, and the representation $\rho_0$ (which is
easily calculated from a reference such as \cite{Slansky:1981yr}
or \cite{MR604363}, and was already given in \cite{grouprep}).

In addition to the rational curves which move in families, there are
isolated rational curves.  If all of these curves $C_j$ have been calculated,
then the representations they form can be determined from the intersection
data $D_i\cdot C_j$ which measures the charges of these classes under
the Cartan subalgebra of the gauge algebra \cite{fiveDgauge,LieF},
and hence determines the weights of $\mathfrak{g}$ occurring in the
representation.
(Here, the $D_i$ are the
components of the inverse images of the various components $\Sigma$.)
In practice, though, the method of Katz and Vafa allows the
matter representation
to be determined more quickly.  The equivalence of these two approaches
has not been checked in every case, but see
\cite{matter1,Esole:2011sm}
for recent progress in this direction.

Here is the Katz--Vafa method.  Given a point $P\in \Sigma$
such that $\pi^{-1}(P)$ contains rational
curves which do not move in families, take a general small disk
$D$ in the base
which meets $\Sigma$ at $P$, and consider the family of elliptic curves
over $D$.  Our key assumption  (see Section \ref{gaugealgebra})
implies that this will never have a non-minimal
Kodaira fiber; thus, there is a particular Kodaira fiber over $P$ on
the Weierstrass model, and the inverse image of $D$ on the Weierstrass
model is a surface with a rational double point.\footnote{For clarity,
we stress that here we are referring to ``Kodaira type'' as determined
by orders of vanishing as in Table~\ref{tab:kodaira}, and not assuming
that the fiber in the nonsingular model is one of Kodaira's fibers.
In fact, in codimension at least two on the base, the fiber need
not coincide with Kodaira's fibers \cite{MR615858,matter1,Esole:2011sm}.}

Now consider a deformation of $D$, i.e.,
a family of disks $D_t$ meeting $\Sigma$ at a variable point
$P_t$ with $D_0=D$.  By the Brieskorn--Grothendieck theorem
\cite{MR0437798,MR584445}, after an appropriate base change $t=\tau^\ell$,
the singularities on $\pi^{-1}(D_{\tau^\ell})$ can be simultaneously resolved.
Moreover, there is a versal space $V$ of deformations of $\pi^{-1}(D_0)$ which
can be simultaneously resolved, and the parameter curve $\{\tau\}$ of
our base-changed family will map to $V$.  Let $m$ be the ramification
of that map at the origin.  Then locally, $t=s^k$ where $s$ is the coordinate
on a local disk within $V$ and $k=\ell/m$.

The Kodaira fiber  over $P=P_0$ determines a
covering algebra $\widetilde{\mathfrak{g}}_P$, and the Kodaira
fiber over $P_{s^k}$ determines a generic
covering algebra $\widetilde{\mathfrak{g}}(\Sigma)$,
with an embedding $\widetilde{\mathfrak{g}}(\Sigma)\subset
\widetilde{\mathfrak{g}}_P$.

The Katz--Vafa method says: there is
a complex representation $\rhoP$ of $\widetilde{\mathfrak{g}}(\Sigma)$
such that the adjoint representation of
$\widetilde{\mathfrak{g}}_P$ decomposes under
$\widetilde{\mathfrak{g}}(\Sigma)$ as
\[ \adj_{\widetilde{\mathfrak{g}}_P} = \adj_{\widetilde{\mathfrak{g}}(\Sigma)} \oplus
\rhoP \oplus \overline{\rhoP} \oplus
\mathbf{1}^{\oplus(\operatorname{rank}(\widetilde{\mathfrak{g}}_P)-
 \operatorname{rank}(\widetilde{\mathfrak{g}}(\Sigma)))};\]
the corresponding local contribution to the matter representation is
the charged part\footnote{Recall (see for example \cite{grouprep}) that the charged part is defined as follows: Let $\rho$ be a representation of a Lie algebra
$\mathfrak{g}$,
with Cartan subalgebra $\mathfrak{h}$.
The {\em charged dimension of $\rho$} is
$(\dim \rho)_{ch}= \dim (\rho)- \dim (ker \rho|_{\mathfrak{h}})$.}
 of $\frac1k\cdot (\rhoP|_{\mathfrak{g}(\Sigma)})$.

(As we will see in examples, when $P$ is a ramification point
for monodromy, this representation naturally contains $\frac12\rho_0$
as a summand.)
As a cycle on $\Sigma$, we define the local contribution
to be the charged part of
\[ \frac1k\cdot P \otimes (\rhoP|_{\mathfrak{g}(\Sigma)}).\]
Note that the representation $\rhoP|_{\mathfrak{g}(\Sigma)}$ could
{\em a priori}\/  have a summand which is neutral under the gauge
algebra $\mathfrak{g}(\Sigma)$; we exclude any such summands from
the local contribution to the matter representation.

The Katz--Vafa method is in fact ambiguous in at least one case, as observed
in \cite{matter1}, since there are two inequivalent embeddings of $A_7$
into $E_8$ \cite{MR0047629,MR1104782}.
We will discuss this case further in Section~\ref{sec:examples}.

\bigskip

There are several  ways to extract information from the
virtual matter cycle, which will prove useful in checking anomaly
cancellation.  First, for any component $\Sigma$ of the discriminant
locus which contributes a non-abelian summand $\mathfrak{g}(\Sigma)$
to the gauge algebra, we can restrict the matter representation
or the matter cycle to  $\mathfrak{g}(\Sigma)$:
\[ Z_{\text{virtual}}|_{\mathfrak{g}(\Sigma)},\]
obtaining a cycle which will involve representations of  $\mathfrak{g}(\Sigma)$
only. 

Second, for any pair of distinct components $\Sigma\ne\Sigma'$
of the discriminant locus which each contribute a non-abelian summand
to the gauge algebra, we can compute the representation-multiplicity
of  $\mathfrak{g}(\Sigma)$ and  $\mathfrak{g}(\Sigma')$ at
the cycle $Z_{\text{virtual}}$:
\[ \mu_{Z_{\text{virtual}}}(\mathfrak{g}(\Sigma),\mathfrak{g}(\Sigma')).\]
Here, we have extended the definition of representation-multiplicity
in a natural way to include the case of
an algebraic cycle $Z$ whose coefficients
are representations, i.e., to elements of
$A^2(B)\otimes
R(\mathfrak{g}_L\otimes\mathfrak{g}_R)\otimes\mathbb{Q}$.
In this extension, the representation-multiplicity
$\mu_Z(\mathfrak{g},\mathfrak{g}')$ is an ordinary cycle, i.e.,
a linear combination of subvarieties with numerical coefficients.

\section{The Tate cycle for $\Sigma$}\label{sec:Tatecycle}

For each component $\Sigma$ of the discriminant locus corresponding
to a non-abelian summand of the gauge algebra,
we now introduce another cycle of codimension two which we call the
{\em Tate cycle for $\Sigma$},
since its definition is closely related to Tate's algorithm.
The part of the anomaly cancellation condition involving $\mathfrak{g}(\Sigma)$
will be satisfied if
the restriction of the virtual matter cycle to $\mathfrak{g}(\Sigma)$
is rationally equivalent to the Tate cycle
for $\Sigma$, or more generally, if the two are Casimir equivalent
in degrees $2$ and $4$.
This
 formulation becomes a criterion which can be checked locally,
for many specific kinds of degenerate fibers in elliptic fibrations.
Sections~\ref{sec:examples} and \ref{sec:about}
will be devoted to checking this criterion
in a wide variety of cases.

Tate's algorithm, reviewed in Appendix~\ref{app:Tate}, starts from a
Weierstrass model of an elliptic fibration and, making changes of coordinates
that involve rational functions on the base as well as appropriate
relative blowups, finds expressions in which it is possible to determine
things such as monodromy of components of the exceptional divisor of
the blowups.  Here, we adapt that algorithm and produce certain rational
sections of line bundles on $\Sigma$ which can be used to mimic Tate's
forms of the equations.  The zeros and poles of those rational sections
are then associated to specific representations of the gauge algebra,
producing an algebraic cycle of codimension two with coefficients in
the representation ring $R(\mathfrak{g})\otimes\mathbb{Q}$,
as we previously did with
the virtual matter cycle.  As we will see, it  can be verified that
anomaly cancellation holds whenever the virtual matter cycle is equal to the
Tate cycle for all components $\Sigma$
(and a certain condition holds on intersections between pairs
of components of the discriminant locus).

\begin{table}[t]
\begin{center}

\begin{tabular}{||c||c|c||c||} \hline
Type & $\betaSigma$ & $\gammaSigma$ & $\deltaSigma$ \\ \hline
$I_2$ & -- & $\deltaSigma (4f^2/81g^2)|_\Sigma$  & $\gammaSigma (-9g/2f)^2|_\Sigma$\\ \hline
$I_m$, $m\ge3$ & $(-9g/2f)|_{\Sigma}$ & $\deltaSigma/\betaSigma^2$ & $\betaSigma^2\gammaSigma$\\ \hline
$III$ & -- & $(f/z)|_{\Sigma}$ &  $4\gammaSigma^3$ \\ \hline
$IV$ &  $(g/z^2)|_{\Sigma}$ & -- & $27\betaSigma^2$ \\ \hline
$I_0^*$ & \multicolumn{3}{|c||}{special case} \\ \hline
$I_{2n-5}^*$, $n\ge3$ & $\deltaSigma/\gammaSigma^3$ & $(-9g/2zf)|_{\Sigma}$ & $\betaSigma\gammaSigma^3$ \\ \hline
$I_{2n-4}^*$, $n\ge3$ & $-\deltaSigma/\gammaSigma^2$ & $(-9g/2zf)|_{\Sigma}$ & $-\betaSigma\gammaSigma^2$ \\ \hline
$IV^*$ & $(g/z^4)|_{\Sigma}$ & -- & $27\betaSigma^2$ \\ \hline
$III^*$ &  -- & $(f/z^3)|_{\Sigma}$  &$4\gammaSigma^3$ \\ \hline
$II^*$ & -- & $(g/z^5)|_{\Sigma}$  &$27\gammaSigma^2$ \\ \hline
\end{tabular}
\smallskip
\caption{Main construction} \label{tab:ab}
\end{center}
\end{table}

Our starting point is an elliptic fibration in Weierstrass form
\[y^2=x^3+fx+g.\]
We assume that $\Sigma$ is a  component of the discriminant
locus $\Delta=0$ of multiplicity $m_\Sigma$, and let $z=0$ be a local defining
equation for $\Sigma$ around a nonsingular point of $\Sigma$.  We let $\deltaSigma = (\Delta/z^{m_\Sigma})|_\Sigma$
be the residual discriminant, a section of $(\LB{12}-m_\Sigma \Sigma)|_\Sigma$.
\begin{itemize}
\item For Kodaira fibers of type $I_m$,  by changing coordinates in $x$ we
may assume that the singular point is located at the origin, and the
equation takes the form
\begin{equation} \label{eq:Weier} y^2 = x^3 + u x^2 + v x + w.\end{equation}
Since $v|_\Sigma$ and $w|_\Sigma$ both vanish, we can determine $u|_\Sigma$
from $f$ and $g$ in the following way.  Completing the cube in
\eqref{eq:Weier} gives
\begin{equation} \label{eq:fg}
\begin{aligned}
f &= -\frac13u^2+v\\ g &= \frac2{27}u^3-\frac13uv+w,
\end{aligned}
\end{equation}
and so $(u|_\Sigma)^2 = -3f|_\Sigma$ and
$(u|_\Sigma)^3 = \frac{27}2 g|_\Sigma$.
Taking the ratio, it follows that
 $u|_\Sigma = (-9g/2f)|_\Sigma$, which can be regarded as
a nonzero rational section of
$\LB2|_\Sigma$.  Note that since this is expressed in terms of $f$ and
$g$, it is independent of the choice of coordinates used to obtain
\eqref{eq:Weier}.

We now define a rational section $\gammaSigma$
of $(\LB8-m_\Sigma \Sigma)|_\Sigma$ by
\[ \gammaSigma := \frac{\deltaSigma}{(u|_\Sigma)^2}.\]
When $m>2$, Tate's analysis
shows that the ramification divisor for monodromy is $\div(u|_\Sigma)$.
In this case, we also define $\betaSigma:=u|_\Sigma$, and note that
 there is a natural algebraic cycle on $\Sigma$
which we can identify with
$\frac12\div(\betaSigma)$, rationally equivalent to $\parenLB|_\Sigma$.

\item For Kodaira fibers of type $II$, $III$, or $IV$, all of the vanishing of
$\deltaSigma$ is attributable to either vanishing of $(f/z^k)|\Sigma$ or of
$(g/z^k)|_\Sigma$; we call this $\gammaSigma$ in types $II$ and $III$ when there
is no possibility of monodromy, and $\betaSigma$ in type $IV$ when Tate's analysis
tells us that $\div((g/z^2)|_\Sigma)$ is the ramification divisor for
monodromy. Note that since singular fibers of type $I_1$ and $II$ are
irreducible on the nonsingular model, they make no contribution to the gauge
algebra.  For this reason, those types are
not included in  Tables \ref{tab:ab},
  \ref{tab:divisors}, and \ref{tab:reps}.

\item For Kodaira fibers of type $I_{m-4}^*$, there is an auxiliary equation
\begin{equation}\label{eq:auxeq}
y^2 = x^3 + (f/z^2) x + (g/z^3)
\end{equation}
describing part of a relevant blowup.  The cubic equation in $x$ has
no {\em a priori}\/ factorization
when $m-4=0$, and this is the trickiest case to
characterize: the characterization depends on how many irreducible factors
the right-hand side of
\eqref{eq:auxeq} has after restriction to $\Sigma$.
(We don't define either $\betaSigma$ or $\gammaSigma$
in this case.)
If the right-hand side of \eqref{eq:auxeq} restricted to $\Sigma$
has a linear factor
\[ (x^3 + (f/z^2) x + (g/z^3))|_\Sigma = (x-a)(x^2+dx+e)\]
then $a$ and $d$ are rational sections of $(\LB2-\Sigma)|_\Sigma$,
while $e$ is a rational section of $(\LB4-2\Sigma)|_\Sigma$.
On the other hand, if
the right-hand side of \eqref{eq:auxeq} restricted to $\Sigma$
has three linear factors
\[ (x^3 + (f/z^2) x + (g/z^3))|_\Sigma = (x-a)(x-b)(x-c)\]
then $a$, $b$, and $c$ are all rational sections of $(\LB2-\Sigma)|_\Sigma$.

For $I_{m-4}^*$ with $m-4>0$, the auxiliary equation \eqref{eq:auxeq} has
one double root and one simple root, and by changing coordinates
we can put the double root at the origin, giving a new auxiliary
equation of the form
\[ y^2 = x^3 + (u/z) x^2 + (v/z^2) x + (w/z^3),\]
with $z\ |\ v/z^2$ and $z\ |\ w/z^3$.
We let $\gammaSigma=(u/z)|_\Sigma = (-9g/2zf)|_\Sigma$, a rational section of
$(\LB2-\Sigma)|_\Sigma$, determined analogously to the $I_m$ case by
completing the cube to obtain $f$ and $g$ from $u$, $v$, and $w$.
Tate's algorithm for $I_{m-4}^*$ is quite involved, but
the upshot is that
$\deltaSigma = (-1)^{k+1}\betaSigma\gammaSigma^k$ for some rational section
$\betaSigma$ of $(\optionLB{(2k-12)}{(12-2k)} - (m+2-k)\Sigma)$,
where $k=3$ when $m$ is odd, and $k=2$ when $m$ is even.
(This is explained in detail in Appendix~\ref{app:Tate}.)
The key fact about $\betaSigma$ is that $\div(\betaSigma)$ is the ramification divisor
for monodromy.\footnote{This corrects a statement from \cite{grouprep},
where the ramification divisor was misidentified for $I_{m-4}^*$.}

\item For Kodaira fibers of type $IV^*$, $III^*$, or $II^*$, once again
all of the vanishing of
$\deltaSigma$ is attributable to either vanishing of $(f/z^k)|_\Sigma$ or of
$(g/z^k)|_\Sigma$; we call this
$\gammaSigma$ in types $III^*$ and $II^*$ when there
is no possibility of monodromy,
and
$\betaSigma$ in type $IV^*$ when Tate's analysis
tells us that $\div((g/z^4)|_\Sigma)$ is the ramification divisor for
monodromy.
\end{itemize}

\begin{table}[t]
\begin{center}

\begin{tabular}{||c|c|c|c||} \hline
Type & $\frac12\div(\betaSigma)$ & $\div(\gammaSigma)$ & $\div(\deltaSigma)$ \\ \hline
$I_2$ & -- & $(\LB8-2\Sigma)|_{\Sigma}$ & $(\LB{12}-2\Sigma)|_{\Sigma}$\\ \hline
$I_m$, $m\ge3$ & $\parenLB|_{\Sigma}$ & $(\LB8-m\Sigma)|_{\Sigma}$ & $(\LB{12}-m\Sigma)|_{\Sigma}$\\ \hline
$III$ & -- & $(\LB4-\Sigma)|_{\Sigma}$ &  $(\LB{12}-3\Sigma)|_{\Sigma}$ \\ \hline
$IV$ & $(\LB3-\Sigma)|_{\Sigma}$ & -- & $(\LB{12}-4\Sigma)|_{\Sigma}$ \\ \hline
$I_0^*$ & \multicolumn{2}{|c|}{special case}& $(\LB{12}-6\Sigma)|_{\Sigma}$ \\ \hline
$I_{2n-5}^*$, $n\ge3$ & $(\LB3-(n-1)\Sigma)|_{\Sigma}$ & $(\LB2-\Sigma)|_{\Sigma}$ & $(\LB{12}-(2n+1)\Sigma)|_{\Sigma}$\\ \hline
$I_{2n-4}^*$, $n\ge3$ & $(\LB4-n\Sigma)|_{\Sigma}$ & $(\LB2-\Sigma)|_{\Sigma}$ & $(\LB{12}-(2n+2)\Sigma)|_{\Sigma}$\\ \hline
$IV^*$ & $(\LB3-2\Sigma)|_{\Sigma}$ & -- & $(\LB{12}-8\Sigma)|_{\Sigma}$ \\ \hline
$III^*$ & -- &  $(\LB4-3\Sigma)|_{\Sigma}$ &$(\LB{12}-9\Sigma)|_{\Sigma}$ \\ \hline
$II^*$  & -- & $(\LB6-5\Sigma)|_{\Sigma}$ &$(\LB{12}-10\Sigma)|_{\Sigma}$ \\ \hline
\end{tabular}
\smallskip
\caption{The constituents of the Tate cycle} \label{tab:divisors}
\end{center}
\end{table}

The definitions of $\betaSigma$ and $\gammaSigma$ are summarized in Table~\ref{tab:ab},
which also shows how $\deltaSigma$ is related to these.
Note that in some cases, only one of $\betaSigma$ and $\gammaSigma$ is defined,
and in the case of $I_0^*$, neither one is defined (and there is
correspondingly no description of $\deltaSigma$).
Whenever $\betaSigma$ is defined, it describes the ramification
of a double cover of $\Sigma$.  Thus, either $\sqrt{\betaSigma}$ is well-defined
on $\Sigma$ (when the double cover splits), or $\sqrt{\betaSigma}$ is
well-defined on the double cover; in either case,
 $\frac12\div(\betaSigma)$ is a well-defined algebraic
cycle class on $\Sigma$.  The equivalence classes of this cycle\footnote{Since
$\Sigma$ is a divisor on $B$, any divisor on $\Sigma$ is an algebraic
cycle of codimension two on $B$.}  and the other
cycles $\div(\gammaSigma)$ and $\div(\deltaSigma)$ are displayed
in Table~\ref{tab:divisors}.

\begin{table}[ht]
\begin{center}

\begin{tabular}{||c|c|c|c|c||} \hline
Type & $\mathfrak{g}(\Sigma)$ & $\rho_\alpha$ & $\rho_{\sqrt{\beta}}$
  & $\rho_{\gamma}$\\ \hline
$I_2$ & $\mathfrak{su}(2)$& $\adj$ &  -- & $\fund$ \\ \hline
$I_{3}$ & $\mathfrak{sp}(1)$& $\adj + 2 \cdot \fund$ & $\fund$ & $\fund$
  \\ \hline
$I_3$ & $\mathfrak{su}(3)$& $\adj$& $\fund$ & $\fund$ \\ \hline
$I_{2n}$, $n\ge2$ & $\mathfrak{sp}(n)$ &  $\adj + \Lambda^2_{\text{irr}}$& $\Lambda^2_{\text{irr}}$ &$\fund$
\\ \hline
$I_{2n+1}$, $n\ge1$ & $\mathfrak{sp}(n)$ & $\adj + \Lambda^2_{\text{irr}} +2\cdot  \fund $& $\Lambda^2_{\text{irr}}+\fund$& $\fund$
\\ \hline
$I_m$, $m\ge2$ & $\mathfrak{su}(m)$& $\adj$ & $\Lambda^2$ & $\fund$  \\ \hline
$III$ & $\mathfrak{su}(2)$& $\adj$& -- & $2\cdot\fund$  \\ \hline
$IV$ & $\mathfrak{sp}(1)$& $\adj + 2\cdot\fund$& $3\cdot \fund $ &-- \\ \hline
$IV$ & $\mathfrak{su}(3)$& $\adj$& $3\cdot\fund $ &-- \\ \hline
$I_0^*$ & $\mathfrak{g}_2$& $\adj + 2\cdot \mathbf{7}$ &
 \multicolumn{2}{|c||}{$\mathbf{7} \mapsto \frac12\div(\deltaSigma)$}
\\ \hline
$I_{0}^*$ & $\mathfrak{so}(7)$& $\adj+\vect$&
 \multicolumn{2}{|c||}{$\spinrep \mapsto \frac16\div(\deltaSigma)$, $\vect\mapsto\frac13\div(\deltaSigma)$ }
\\ \hline
$I_{0}^*$ & $\mathfrak{so}(8)$& $\adj$&
 \multicolumn{2}{|c||}{$\vect$, $\spinrep_+$, $\spinrep_-$ $\mapsto \frac16\div(\deltaSigma)$}
\\ \hline
$I_{1}^*$ & $\mathfrak{so}(9)$& $\adj +\vect$& $\vect$&$\spinrep_*$
\\ \hline
$I_{1}^*$ & $\mathfrak{so}(10)$& $\adj$& $\vect$&$\spinrep_*$ \\ \hline
$I_{2}^*$ & $\mathfrak{so}(11)$& $\adj +\vect$& $\vect$ &$\frac12\cdot\spinrep_*$
\\ \hline
$I_{2}^*$ & $\mathfrak{so}(12)$& $\adj$& $\vect$ &$\frac12\cdot\spinrep_*$
\\ \hline
$I_{3}^*$ & $\mathfrak{so}(13)$& $\adj +\vect$& $\vect$ &$\frac14\cdot\spinrep_*+\vect$
\\ \hline
$I_{3}^*$ & $\mathfrak{so}(14)$& $\adj$& $\vect$  &$\frac14\cdot\spinrep_*+\vect$
\\ \hline
$I_{m-4}^*$, $m\ge8$ & $\mathfrak{so}(2m-1)$& $\adj + \vect$& $\vect$ &N/A
\\ \hline
$I_{m-4}^*$, $m\ge8$ & $\mathfrak{so}(2m)$& $\adj$& $\vect$ &N/A \\ \hline
$IV^*$ & $\mathfrak{f}_4$&$\adj+ \mathbf{26}$& $\mathbf{26}$ &-- \\ \hline
$IV^*$ & $\mathfrak{e}_6$ & $\adj$& $\mathbf{27}$ &--   \\ \hline
$III^*$ &  $\mathfrak{e}_7$& $\adj$& -- &$\frac12\cdot\mathbf{56}$  \\ \hline
$II^*$ & $\mathfrak{e}_8$& $\adj$& -- &N/A  \\ \hline
\end{tabular}

\smallskip
\caption{The Tate representations (in $R(\mathfrak{g}(\Sigma))\otimes\mathbb{Q})$}
\label{tab:reps}
\end{center}
\end{table}

One key thing to note: by our assumptions, we cannot have a point
at which the multiplicities of $(f,g,\Delta)$ exceed $(4,6,12)$.
This implies that
 (1) for $I_{m-4}^*$, $m-4\ge4$, $\gammaSigma$ may not vanish, and
(2) for $II^*$, $\gammaSigma$ may not vanish.  Thus, for $I_{m-4}^*$, $m-4\ge4$,
$(\LB2-\Sigma)|_\Sigma$ must be trivial, and for $II^*$, $(\LB6-5\Sigma)|_\Sigma$
must be trivial.

\begin{table}[t]
\begin{center}

\begin{tabular}{|c|c|c|} \hline
 $\mathfrak{g}(\Sigma)$ & $\rho$ & $\Gamma_\rho$ \\ \hline
 $\mathfrak{su} (2) $& $\adj$ &$\frac12(\minusLB+\Sigma)|_{\Sigma}$\\
                   & $\fund$ &$(\LB8-2\Sigma)|_{\Sigma}$ \\ \hline
 $\mathfrak{su} (3) $& $\adj$ &$\frac12(\minusLB+\Sigma)|_{\Sigma}$\\
                   & $\fund$ &$(\LB9-3\Sigma)|_{\Sigma}$ \\ \hline
 $\mathfrak{su} (m) $,& $\adj$ &$\frac12(\minusLB+\Sigma)|_{\Sigma}$   \\
$m\ge4$         &      $\fund$ &$(\LB8-m\Sigma)|_{\Sigma}$ \\
     & $\Lambda^2$ &$\parenLB|_{\Sigma}$  \\
 \hline
 $\mathfrak{sp} (n) $,&$\adj$ &$\frac12(\minusLB+\Sigma)|_{\Sigma}$\\
     $n\ge2$          &$\fund$ &$(\LB8-2n\Sigma)|_{\Sigma}$\\
     &$\Lambda^2_{\text{irr}}$ &$\frac12(\LB{}+\Sigma)|_{\Sigma}$\\
\hline
$\mathfrak{so}(\ell)$, & $\adj$ &$\frac12(\minusLB+\Sigma)|_{\Sigma}$\\
$7\le \ell\le14$,       & $\vect$ &$\frac12(\optionLB{(4-\ell)}{(\ell-4)}+(6-\ell)\Sigma)|_{\Sigma}$\\
& $\spinrep_*$ & $\frac1{\dim(\spinrep_*)}(\LB{32}-16\Sigma)|_{\Sigma}$ \\ \hline
$\mathfrak{so}(4n{-}1)$, & $\adj$ &$\frac12(\minusLB+\Sigma)|_{\Sigma}$\\
$n\ge4$       & $\vect$ &$(\LB{\frac72}-(n{-}\frac12)\Sigma)|_{\Sigma}$\\ \hline
$\mathfrak{so}(4n)$, & $\adj$ &$\frac12(\minusLB+\Sigma)|_{\Sigma}$\\
$n\ge4$       & $\vect$ &$(\LB4-n\Sigma)|_{\Sigma}$\\ \hline
$\mathfrak{so}(4n{+}1)$, & $\adj$ &$\frac12(\minusLB+\Sigma)|_{\Sigma}$\\
$n\ge4$       & $\vect$ &$(\LB{\frac52}-(n{-}\frac12)\Sigma)|_{\Sigma}$\\ \hline
$\mathfrak{so}(4n{+}2)$, & $\adj$ &$\frac12(\minusLB+\Sigma)|_{\Sigma}$\\
$n\ge4$       & $\vect$ &$(\LB3-n\Sigma)|_{\Sigma}$\\ \hline
 $\mathfrak{e}_6$ & $\adj$ &$\frac12(\minusLB+\Sigma)|_{\Sigma}$\\
       & $\mathbf{27}$ &$(\LB3-2\Sigma)|_{\Sigma}$\\ \hline
 $\mathfrak{e}_7$ & $\adj$ &$\frac12(\minusLB+\Sigma)|_{\Sigma}$\\
       & $\mathbf{56}$ &$\frac12(\LB4-3\Sigma)|_{\Sigma}$\\ \hline
 $\mathfrak{e}_8$ & $\adj$ &$\frac12(\minusLB+\Sigma)|_{\Sigma}$\\ \hline
 $\mathfrak{f}_4$ & $\adj$ &$\frac12(\minusLB+\Sigma)|_{\Sigma}$\\
       & $\mathbf{26}$ &$\frac12(\LB5-3\Sigma)|_{\Sigma}$\\ \hline
 $\mathfrak{g}_2$ & $\adj$ &$\frac12(\minusLB+\Sigma)|_{\Sigma}$\\
        & $\mathbf{7}$ &$(\LB5-2\Sigma)|_{\Sigma}$\\ \hline

\end{tabular}

\end{center}
\smallskip
\caption{%
Cycles associated to representations.
}\label{tab:reverse}
\end{table}

For each gauge algebra, we  associate a representation (with $\mathbb{Q}$
coefficients) to each of
these cycles, as specified in\footnote{These representations have
not been chosen arbitrarily.  Rather, as we will see in Section~\ref{2},
they are precisely the representations we need in order to cancel
anomalies.} Table~\ref{tab:reps}.
(Note that we do not assign a representation to the excluded cases:
there is no $\rho_\gamma$ for either $I_{m-4}^*$, $m-4\ge4$, or $II^*$.)
The case of $I_0^*$ is not completely described in the Table, but
follows our earlier discussion in the various cases.
First, if the cubic equation is irreducible (so that the gauge algebra
is $\mathfrak{g}_2$),
then each point of the ramification locus  of the cover is
associated to $1/2$ of a $7$-dimensional representation.
Second, if the cubic equation has one linear and one quadratic factor
(so that the gauge algebra is $\mathfrak{so}(7)$),
the linear factor determines a bundle equivalent to $\frac16\div(\deltaSigma)$,
and the zeros of that are identified with the spinor representation,
while the ramification points of the quadratic factor correspond
to the vector representation.  And third, if the cubic equation
factors completely, then each factor determines a divisor equivalent
to $\frac16\div(\deltaSigma)$, and the zeros of each of those
correspond to one of
the $8$-dimensional representations of $\mathfrak{so}(8)$.

We now define the {\em Tate cycle for $\Sigma$}\/ to be
\begin{equation}\label{tatecycle} Z_{\text{Tate},\Sigma}=
 \frac12(\minusLB+\Sigma)|_\Sigma \otimes \rho_{\alpha}
+ \frac12\div(\betaSigma) \otimes \rho_{\sqrt{\beta}}
+  \div(\gammaSigma) \otimes \rho_\gamma.
\end{equation}
Aside from the first term, this cycle has a representative
(with $\mathbb{Q}$-coefficients)
which is localized at the zeros of $\deltaSigma$.
Each component $\Gamma_\Sigma$
of $\{\deltaSigma=0\}$ thus has an associated representation,
namely, its coefficient in $Z_{\text{Tate},\Sigma}$.

A very simple manipulation with
Tables~\ref{tab:divisors} and \ref{tab:reps}, collecting terms by irreducible
representation instead of by irreducible cycle,  now shows that the
Tate cycle is rationally equivalent to a cycle of the form
\[ Z_{\text{Tate},\Sigma}\sim\sum_\rho \Gamma_{\rho}\otimes \rho,\]
where the cycle classes $\Gamma_\rho$ are given in Table~\ref{tab:reverse},
and, remarkably,
depend only on the gauge algebra $\mathfrak{g}(\Sigma)$, not on its particular
geometric realization.

Let us compare the ``Tate representations'' of this paper
with the representations described
in \cite[Table A]{grouprep}.  There are some minor differences, but
for the most part the representation denoted $\rho_0$ in the earlier
paper corresponds to the representation $\rho_0$ from \eqref{eq:rho0}
and also coincides with the new $\rho_{\sqrt{\beta}}$ when
there is monodromy;
the representation $\rho_1$ in the earlier paper corresponds to
$\rho_{\sqrt{\beta}}$ in this paper when there is no monodromy;
and the representation $\rho_2$ in the earlier paper corresponds
to $\rho_\gamma$ in this paper.

\section{Six-dimensional anomalies} \label{sec:Chow}

The anomaly of a supersymmetric
 six-dimensional theory (with no abelian local factor in its gauge
group) consists of a pure gravitational
anomaly which is a quartic Casimir in the gravitational
curvature, a pure gauge anomaly which is a quartic Casimir in
the gauge curvature, and a mixed anomaly which is a  product of quadratic
Casimirs in the gravitational and gauge curvatures
\cite{Green:1984bx,Erler:1993zy}.  Each
of these anomalies must vanish.

Without a Green--Schwarz term, such a  theory is typically
anomalous.  The total anomaly
 (in a suitable normalization \cite{Schwarz:1995zw}) is:
\begin{equation} \kappa \cdot A \cdot tr R^4+ B \cdot  (\tr R^2)^2 + \frac{1}{6} \tr R^2 \sum
X_i ^{(2)} -\frac{2}{3} \sum X_{i} ^{(4)} + 4\sum _{i < j
 }  Y_{ij}
\end{equation}
 where
 $\kappa$ is a nonzero proportionality constant,
 $A= (n_V- n_H + 273 - 29 n_T)$ and $B= (\frac{9-n_T}{8})$.
Here $n_T$ is the number of tensor
multiplets,
$n_V$ the number of
vector multiplets,  $ n_H$ the number of
hypermultiplets,\footnote{The geometric interpretation of these various
multiplets is given in Section \ref{1} below.} 
and
\begin{align*}
X_i ^{(n)} &= \Tr_{\adj} F_i ^n - \sum_\rho {n_\rho} \Tr _\rho F^n _i\\
Y_{ij}&= \sum _{\rho,\sigma} {n_{\rho\sigma}} \ \Tr _{\rho} F^2 _i
\  \Tr _{\sigma} F^2 _j,
\end{align*}
where
$R$ is the curvature of the Levi-Civita connection,
and $F$ is the curvature of the
gauge connection.
In these formulas,
$\Tr_{adj}$ means the trace in the adjoint representation, $\Tr
_\rho  $ denotes the trace in the representation $\rho$ of the
simple algebra $\mathfrak{g}_i$ (see Appendix \ref{app:A}), $n_\rho$ is the
multiplicity of the representation $\rho$ of $\mathfrak{g}_i$ in the matter
representation,\footnote{In the physics literature one says that
there are ``$n_\rho$ hypermultiplets in the representation
$\rho$.''} and $n_{\rho\sigma}$ is the multiplicity of the representation
$\rho\otimes\sigma$ of $\mathfrak{g}_i \oplus \mathfrak{g}_j$
in the matter representation.

We can rewrite these expressions in terms of the total virtual
representation
\[ \rho_{\text{virtual}} = -\adj_{\mathfrak{g}}+\sum_\rho n_\rho \cdot\rho \]
as follows.
First,\footnote{%
Note that since no irreducible component of the
 ``adjoint'' term in
  $\rho_{\text{virtual}}$ is  charged under two
different non-abelian summands of $\mathfrak{g}$, this term has
no effect on $Y_{ij}$.
} note that
\[ Y_{ij} = \mu_{\rho_{\text{virtual}}}(\mathfrak{g}_i,\mathfrak{g}_j)
\tr_{\mathfrak{g}_i} (F_i^2)\tr_{\mathfrak{g}_j} (F_j^2),\]
that is, $Y_{ij}$ can be expressed in terms of the representation-multiplicity.
Next we define
\begin{align*}
 X^{(n)}(\rho_{\text{virtual}}) &=- \sum X_i^{(n)}\\
Y(\rho_{\text{virtual}}) &= \sum_{i\ne j}Y_{ij} = 2\sum_{i<j}Y_{ij}.
\end{align*}
(The awkward sign in the first definition is due to the sign
in the original definition of $X^{(n)}_i$, and is designed
so that $X^{(n)}(\sum_\rho n_\rho\cdot\rho) = \sum_{i,\rho} n_\rho
\Tr_\rho F_i^n $.)
This lets us rewrite the anomaly in the form
\begin{equation} \kappa \cdot A \cdot tr R^4+ B \cdot  (\tr R^2)^2 - \frac{1}{6} \tr R^2
X ^{(2)}(\rho_{\text{virtual}}) +\frac{2}{3}  X ^{(4)}(\rho_{\text{virtual}})
+ 2Y(\rho_{\text{virtual}}).
\end{equation}

To ensure an anomaly-free six-dimensional theory,
a term of Green--Schwarz type  must be included
\cite{Green:1984sg, Sagnotti:1992qw, Schwarz:1995zw}.
In the Calabi--Yau threefold case,
Sadov \cite{Sadov:1996zm} derived the form of the Green--Schwarz term
by reducing from ten dimensions, for all F-theory models which in type IIB
language admit only D7-branes and orientifold O7-planes.  (In the Kodaira
language used in this paper, this corresponds to allowing only
 singular fibers of types
$I_m$ and $I_{m-4}^*$.)  This was extended to the general case in
\cite{grouprep}, as we review here.

The Green--Schwarz term in the action takes the form
\[ \int \Psi \cdot \left(\frac12 K_B \otimes \tr R^2
+ 2 \sum_\Sigma (\Sigma\otimes\tr_{\mathfrak{g}(\Sigma)} F^2) \right)\]
where $\Psi$ is a $2$-form field in the effective six-dimensional
theory labeled by an element of
 $H^2(B,Z)$ (with indices suppressed) obtained by dimensional reduction
from a $4$-form field in ten dimensions, and $\cdot$ denotes the
intersection product in $H^2(B,Z)$.  The key expression
\[ D_{\text{gauge}}(F)
:=\sum_\Sigma (\Sigma\otimes\tr_{\mathfrak{g}(\Sigma)} F^2), \]
which we call the {\em gauge divisor},
involves the normalized trace for the corresponding summand of the
Lie algebra; that this normalization
gives the correct linear combination was verified by
Sadov for $\mathfrak{g}(\Sigma)= \mathfrak{u}(m)$,
 $\mathfrak{g}(\Sigma)= \mathfrak{sp}(m)$, and
 $\mathfrak{g}(\Sigma)= \mathfrak{so}(\ell)$.  The normalization for
exceptional groups is a consequence of the study of anomalies made
in \cite{grouprep}.

More generally, for an elliptic fibration which may not be
Calabi--Yau, we propose as Green--Schwarz term
\[ \int \Psi \cdot \left(-\frac12 L \otimes \tr R^2
+ 2 \sum_\Sigma (\Sigma\otimes\tr_{\mathfrak{g}(\Sigma)} F^2) \right)\]
where $L$ is the line bundle used to construct the Weierstrass
model (which coincides with $-K_B$ in the Calabi--Yau threefold case).

Including such a term in the action (as we must, in deriving F-theory from
the type IIB string), there is a contribution to the anomaly of
\begin{align*}
- \frac12\left( -\frac12L \otimes \tr R^2 +2 D_{\text{gauge}}(F)\right)^2
&=- \frac18 L^2 \otimes (\tr R^2)^2
+\sum_\Sigma (L \cdot \Sigma) \otimes (\tr R^2 \tr_{\mathfrak{g}(\Sigma)} F^2) \\
& \quad - 2 \left( \sum_\Sigma \Sigma\otimes (\tr_{\mathfrak{g}(\Sigma)} F^2)
\right)^2
\end{align*}
To obtain an anomaly-free theory, then, requires four conditions:
\begin{align}
n_V-n_H+273-29n_T &= 0 \label{condition1}\\
\frac{9-n_T}8 &= \frac18 L^2 \label{condition2}\\
\frac16 X^{(2)}(\rho_{\text{virtual}}) &=
   \sum (L\cdot \Sigma)\tr_{\mathfrak{g}(\Sigma)} F^2=L\cdot D_{\text{gauge}}(F)\label{condition3}\\
\frac13X^{(4)}(\rho_{\text{virtual}}) + Y(\rho_{\text{virtual}})
&=   \left(
\sum_\Sigma \Sigma\otimes (\tr_{\mathfrak{g}(\Sigma)} F^2)
\right)^2=D_{\text{gauge}}(F)^2.\label{condition4}
\end{align}
Note that:
\begin{itemize}
\item Condition \eqref{condition1} is equivalent to a formula for the topological
Euler characteristic of the total space of the elliptic fibration,
this was the focus of \cite{grouprep} (see Section \ref{1}).
\item Condition \eqref{condition2}
in the Calabi--Yau case
is equivalent to $K_B^2=9-n_T$, and
follows from the fact that $n_T$ is the number of times $\mathbb{P}^2$
must be blown up (or one fewer than the number of times a Hirzebruch
surface $\mathbb{F}_n$ must be blown up) to obtain $B$ (see Section \ref{1}).
\item Conditions \eqref{condition3} and \eqref{condition4} put
specific constraints on the matter representation of the theory.
\end{itemize}
Our main result is a stronger form of the last two conditions,
described in Section \ref{2}.

\section{The Euler characteristic}\label{1}

It is known \cite{FCY2} that if  $X$ is Calabi--Yau the geometric interpretation of the
numbers of multiplets is
\begin{align}\label{mv}
n_V&= \dim(G) , \\  n_T&= h^{1,1}(B)-1 \\
 n_H&=H_{ch} + H_0, \text{ with } H_0= h^{2,1}(X) +1,
\end{align}
where $n_V$,   $n_T$  denote  the number of  vector, tensor multiplets respectively;  $n_H$ is the number of hypermultiplets and $H_{ch}$ denotes the charged hypermultiplets and $H_0$ denotes
the neutral hypermultiplets. 

Moreover, since we are assuming $\rk
MW(X/B)=0$, we have
\begin{equation}\label{eq:rk}
\rk (G) = h^{1,1}(X)- h^{1,1}(B) -1.
\end{equation}

In \cite{grouprep} we provided an algorithm to analyze $H_{ch}$ in terms of the topological Euler characteristic, but we do not derive explicitly this formula from the physics. We shall do that here.

Translated into geometric quantities,
the key formula \ref{condition1} becomes
\begin{equation}\label{eq:firstvanbis}
 h^{2,1} (X)+ 1 +H_{ch} -\dim (G)= 273-29n_T
\end{equation}
and
\begin{equation}\label{eq:secondvanbis}
\f \c (X) = \rk(G) + n_T + 2 - h^{2,1}(X).
\end{equation}
The first statement follows from equation \ref{mv}; the second  holds because $X$ is a Calabi-Yau threefold. In fact the equality  $h^{1,0}(X)=h^{2,0}(X)=0$ and $h^{3,0}(X)=1$ imply:
\begin{align*}\c (X) &=b_0-b_1+b_2-b_3+b_4\\&= 1-0+h^{1,1}(X)-(2+2h^{2,1}(X))+h^{1,1}(X)-0+1\\
&=2h^{1,1}(X)-2h^{2,1}(X).
\end{align*}
Thus, using equation \ref{eq:rk}, we find
$$\f \c (X) = h^{1,1}(X)-h^{2,1}(X)
= 2+n_T+\rk(G)  - h^{2,1}(X).$$

To interpret formula \eqref{eq:firstvanbis} as a formula for the Euler
characteristic, we need a few more geometric facts.

Let $\pi: X \to B$ be an elliptic Calabi--Yau threefold; then
\begin{align} \c (B)= 2 + h^{1,1}(B) \\
10= h^{1,1}(B) + K_B ^2 \\
9= n_T + K_B ^2 \label{eq:tens}
\end{align}

Our assumptions in fact imply that $B$ is either rational or an Enriques surface \cite{Grassi91}: then $\chi (\mathcal O _B)=1$.
The second equality follows from  Noether's formula $K_B^2+ \c (B) = 12 \chi (\mathcal O_B)$  and the third from equation \ref{mv}.
Now if we add equation \ref{eq:tens} with \ref{eq:secondvanbis} we find
$$\f \c (X) + 30 \kb ^2 = 270-30n_T+2+n_T+\rk(G)  - h^{2,1}(X)$$
from which the equation
$$\f \c (X)+ 30 \kb ^2 = 273-29n_T +\rk (G) -( h^{2,1}(X) +1).$$
immediately follows.
 Then the condition \ref{condition1} is then equivalent to
\begin{equation}\label{cor:firstvan} \f \c (X) + 30 \kb ^2= H_{ch} - (\dim(G)- \rk (G)).
\end{equation}

In \cite{grouprep} we defined  and analyzed  the quantity {\r}
${}:=
 \f \c (X) + 30 \kb ^2$ in terms of the matter representation.

\begin{remark}
 In general if  $n_T= h^{1,1}(B)-1$,
then $9-n_T=K_B^2$ implies that either  $h^{1,0}(B)= h^{2,0}(B)= 0$ or $h^{1,0}(B)= 5$ and $ h^{2,0}(B)= 4$. There are other type of surfaces which satisfy these hypothesis, most notably some of general type.
\end{remark}

\section{Anomaly cancellation in the Chow group} \label{2}

It turns out that the last two conditions for anomaly cancellation
hold not only numerically, but as actual algebraic cycles.  That is,
if we use the virtual matter cycle $Z_{\text{virtual}}$ on the
left hand side, and interpret the intersection on the right hand
side as intersection in the Chow group, then we get the stronger
statements
\begin{align}
\frac16 X^{(2)}(Z_{\text{virtual}}) &=
   \sum (L\cdot \Sigma)\tr_{\mathfrak{g}(\Sigma)} F^2=L\cdot D_{\text{gauge}}(F)\label{condition3chow}\\
\frac13X^{(4)}(Z_{\text{virtual}}) + Y(Z_{\text{virtual}})
&=   \left(
\sum_\Sigma \Sigma\otimes (\tr_{\mathfrak{g}(\Sigma)} F^2)
\right)^2=D_{\text{gauge}}(F)^2,\label{condition4chow}
\end{align}
which are to be interpreted as equality of codimension two cycles
on the base $B$
up to rational equivalence.

The coefficients in these cycles are Casimir operators for the
gauge algebra $\mathfrak{g}$; for these relations to be
satisfied, they must hold in each sector of the algebra of
Casimir operators.  In particular, the first statement must
hold as an equality of quadratic Casimir operators
when restricted to each summand $\mathfrak{g}(\Sigma)$:
\[ X^{(2)}(Z_{\text{virtual}}|_{\mathfrak{g}(\Sigma)})
= 6(L\cdot \Sigma) \tr_{\mathfrak{g}(\Sigma)} F^2;\]
the second statement must hold as an equality of quartic Casimir
operators which restricted to each summand $\mathfrak{g}(\Sigma)$:
\[ X^{(4)}(Z_{\text{virtual}}|_{\mathfrak{g}(\Sigma)})
= 3(\Sigma\cdot \Sigma) (\tr_{\mathfrak{g}(\Sigma)} F^2)^2;\]
and the second statement must also hold as an equality of
bi-quadratic Casimir operators when restricted to each
pair of summands $\mathfrak{g}(\Sigma)\oplus\mathfrak{g}(\Sigma')$:
\[ \mu_{Z_{\text{virtual}}}(\mathfrak{g}(\Sigma),\mathfrak{g}(\Sigma'))
= \Sigma \cdot \Sigma'.\]
(In the last equation, we suppressed the generator
$\tr_{\mathfrak{g}(\Sigma)}(F^2)\tr_{\mathfrak{g}(\Sigma')}(F^2)$
of the bi-quadratic Casimirs since that is taken care of in the
definition of the representation-multiplicity $\mu$.)

On the other hand, it is straightforward to verify using
Tables~\ref{tab:D} and \ref{tab:reverse} that
\begin{align*}
 X^{(2)}(Z_{\text{Tate},\Sigma})
&= 6L|_{ \Sigma} \otimes \tr_{\mathfrak{g}(\Sigma)} F^2\\
 X^{(4)}(Z_{\text{Tate},\Sigma})
&= 3\Sigma|_{ \Sigma} \otimes (\tr_{\mathfrak{g}(\Sigma)} F^2)^2
\end{align*}
(and this in fact motivated our definition of the Tate cycles).
Thus, our main anomaly cancellation result is:

\medskip
\noindent
{\bf Main Result.}  {\em The elliptic fibration defines an anomaly-free theory
if it satisfies
\eqref{condition1} and \eqref{condition2},  if
there is a Casimir equivalence in degrees $2$ and $4$
\[Z_{\text{virtual}}|_{\mathfrak{g}(\Sigma)}\sim Z_{\text{Tate},
\Sigma}\]
 for all $\Sigma$, and if there is a rational equivalence of cycles
\[\mu_{Z_{\text{virtual}}}(\mathfrak{g}(\Sigma),\mathfrak{g}(\Sigma'))
=\Sigma\cdot \Sigma'
\]
 for all $\Sigma\ne\Sigma'$.}

\medskip

In the Calabi--Yau case, when $L=-K_B$, the non-local part of
these equations holds, since the contribution of the adjoint representation
is $g-1$ which is calculated by the cycle $\frac12(K_B+\Sigma)|_\Sigma$.
Thus, the only things to check are the local contributions
(including copies of $\rho_\alpha$ associated to singularities
of $\Sigma$), and these
can be checked cycle by cycle.

\medskip
\noindent
{\bf Local Anomaly Cancellation.}
{\em Suppose that $X$ is Calabi--Yau.
Let $\Gamma$ be a subvariety of codimension two, let $\rho_\Gamma$
be the local contribution to the matter
representation associated to $\Gamma$, and let $\Sigma_1$, \dots,
$\Sigma_k$ be the components of the discriminant locus which pass through
$\Gamma$ and which contribute non-abelian summands to the gauge algebra.
If $\rho_\Gamma|_{\Sigma_j}$ is Casimir equivalent in degrees $2$ and
$4$ to the contribution at $\Gamma$
to the Tate representation for $\Sigma_j$, and if for all $i\ne j$
\[\mu_{\rho_\Gamma}(\mathfrak{g}(\Sigma_i),
\mathfrak{g}(\Sigma_j))=
\operatorname{mult}_\Gamma(\Sigma_i,\Sigma_j) ,\]
that is, the representation-multiplicity coincides with the
intersection multiplicity,
then local anomaly cancellation holds at $\Gamma$.
}

\medskip

\noindent
Note that the local contribution to the Tate representation for $\Sigma$
is easily calculated from the Weierstrass equation, since it depends only
on the order of zero or pole along $\Gamma$ of $\betaSigma$ and/or
$\gammaSigma$, as well as the genus drop at a singular point
of $\Sigma$.\footnote{The simplest example of this phenomenon is
an ordinary double point of $\Sigma$ in the $I_m$ case,
which Sadov argued \cite{Sadov:1996zm} is associated to
the symmetric representation $S^2V$ of $\mathfrak{su}(m)$.
Since $S^2V$ is Casimir equivalent to $\adj - \Lambda^2$,
this is accounted for by an appropriate local computation.}

Anomaly cancellation is thus reduced to this kind of local computation.
We carry it out for a wide variety of examples in the next Section.

\section{Examples}\label{sec:examples}

We have seen that the anomaly cancellation can be reduced to
a straightforward property about codimension two cycles on which
the elliptic fibration structure degenerates.  Namely, given such
a cycle and a component $\Sigma$ of the discriminant locus containing
the cycle, one can calculate the local contribution\footnote{We have
already verified the anomaly cancellation condition for the
global contributions to the Tate and virtual matter cycles, so
we can now focus on local contributions only.} to the Tate cycle for $\Sigma$
directly from the Weierstrass equation; one can also calculate intersection
multiplicities of all pairs of components of the discriminant locus
which pass through the cycle.  This data must then be compared with
the contribution of that cycle to the virtual matter representation.

In this Section and the next, we carry out this verification for the standard
``generic'' codimension two singularities of elliptic fibrations
from \cite{geom-gauge}
(as already verified in \cite{grouprep}) as well as for some new
codimension two singularities such as the one from \cite{newTate}.
We also introduce some singularities which are considered here
for the first time.  All of our examples are local, considered in
a neighborhood of a particular codimension two locus defined by
$\{z=t=0\}$.

Let us first consider cases in which the Kodaira type along
$\Sigma=\{z=0\}$ is $I_m$ with $m=2n$ or $m=2n+1$.
The generalized Weierstrass forms proposed in \cite{geom-gauge,newTate} can
all be written in the general form (see \cite[Appendix A]{newTate})
\begin{equation} \label{eq:Im}
 y^2 = x^3 + a_2 x^2 + a_{4,n} z^n x + a_{6,2n}z^{2n},
\end{equation}
with  additional restrictions for  various particular  cases.
The discriminant for such a Weierstrass form is
\begin{equation} \label{disc:Im}
 \Delta = 4 a_2^3 a_{6,2n} z^{2n} - a_2^2a_{4,n}^2z^{2n} + O(z^{3n})
\end{equation}
whenever $n\ge1$.

\medskip

\noindent {\bf Example 1.}
Consider $I_m$, $m\ge3$, with $a_2\ne0$ at the codimension two singular point,
and let $n=[m/2]$.
Then we can rewrite \eqref{eq:Im} as
\begin{equation} \label{eq:Im1}
 y^2 = x^3 + a_2\left(x+\frac{a_{4,[m/2]}}{2a_2}z^{[m/2]}\right)^2
+tz^m + O(z^{m+1})
\end{equation}
which yields a discriminant that satisfies
\begin{equation} \label{disc:Im1}
 \Delta = 4a_2^3tz^m + O(z^{m+1}).
\end{equation}
Thus, we have identified the local contribution to the residual
discriminant with
the codimension two locus $\{z=t=0\}$.
Note that when the coefficients are generic,
the Kodaira fiber along a disk through ${z=t=0}$ has type $I_{m+1}$.

Since the coefficients of $x^1$ and $x^0$ in \eqref{eq:Im1} are
divisible by $z$, $\betaSigma$ coincides with the coefficient
of $x^2$, i.e.,  $\betaSigma=a_2$. Since $\deltaSigma=4a_2^3t$, it follows
that $\gammaSigma=4a_2t$.

\medskip

\noindent {\bf Example 2.}
Consider $I_{2n}$, $n\ge2$, in the form \eqref{eq:Im}
with $a_2=t$ vanishing at the codimension two
singular point.  Thus, we have equation
\begin{equation} \label{eq:Im2}
 y^2 = x^3 + t x^2 + a_{4,n} z^n x + a_{6,2n}z^{2n},
\end{equation}
and discriminant
\begin{equation} \label{disc:Im2}
 \Delta = t^2( 4t a_{6,2n}  - a_{4,n}^2)z^{2n} + O(z^{2n+2}).
\end{equation}
The higher order of vanishing of the error term is important, because
it shows that the total order of vanishing of $\Delta$ in the limit
is at least $2n+2$.  Thus, the special fiber is $I_{2n-4}^*$.

Again $\betaSigma$ is the coefficient of $x^2$, i.e., $\betaSigma=t$; since
$\deltaSigma=t^2( 4t a_{6,2n}  - a_{4,n}^2)$ it
follows that  $\gammaSigma= 4t a_{6,2n}  - a_{4,n}^2$.

\medskip

We now specialize to cases in which the Kodaira type along $\Sigma$
is $I_{2n+1}$, $n\ge1$.
The new generalized Weierstrass
form in this case proposed in \cite{newTate} (with a minor change of
notation) is
\begin{equation} \label{eq:generalized}
 y^2 = x^3 + (\frac14\mu\nu^2+ a_{2,1}z) x^2 + (\frac12\mu \nu \xi
+ a_{4,n+1} z)z^n x
+ (\frac14\mu\xi^2 + a_{6,2n+1}z)z^{2n},
\end{equation}
and the discriminant takes the form
\begin{equation} \label{disc:generalized}
 (\frac14\mu\nu^2+ a_{2,1}z)^2
\left[\mu(\xi^2a_{2,1} - \nu\xi a_{4,n+1}+ \nu^2a_{6,2n+1})
+ (4a_{2,1}a_{6,2n+1} - a_{4,n+1}^2)z\right]z^{2n+1} + O(z^{3n}).
\end{equation}

\medskip

\noindent {\bf Example 3.}
Consider $I_{2n+1}$, $n\ge4$ in the form \eqref{eq:generalized}
with $\mu=t$ vanishing at the singular point.  The equation becomes
\begin{equation} \label{eq:generalized3}
 y^2 = x^3 + (\frac14t\nu^2+ a_{2,1}z) x^2 + (\frac12t \nu \xi
+ a_{4,n+1} z)z^n x
+ (\frac14t\xi^2 + a_{6,2n+1}z)z^{2n},
\end{equation}
and the discriminant is
\begin{equation} \label{disc:generalized3}
 (\frac14t\nu^2+ a_{2,1}z)^2
\left[t(\xi^2a_{2,1} - \nu\xi a_{4,n+1}+ \nu^2a_{6,2n+1})
+ (4a_{2,1}a_{6,2n+1} - a_{4,n+1}^2)z\right]z^{2n+1} + O(z^{2n+4}),
\end{equation}
since $3n\ge2n+4$.
The order of vanishing increases
by at least $3$ at $z=t=0$.

By construction, $\betaSigma=\frac14t\nu^2$; since
\[\deltaSigma=(\frac14t\nu^2+ a_{2,1}z)^2t(\xi^2a_{2,1} - \nu\xi a_{4,n+1}+ \nu^2a_{6,2n+1})\]
 it follows that
$\gammaSigma=t(\xi^2a_{2,1} - \nu\xi a_{4,n+1}+ \nu^2a_{6,2n+1})$.

\medskip

\noindent {\bf Example 4.}
We again
consider $I_{2n+1}$, $n\ge3$, using equation \eqref{eq:generalized},
this time with
$\nu=t$ vanishing at the singular point.  The equation becomes
\begin{equation} \label{eq:generalized4}
 y^2 = x^3 + (\frac14\mu t^2+ a_{2,1}z) x^2 + (\frac12\mu t \xi
+ a_{4,n+1} z)z^n x
+ (\frac14\mu\xi^2 + a_{6,2n+1}z)z^{2n},
\end{equation}
and the discriminant takes the form
\begin{equation} \label{disc:generalized4}
 (\frac14\mu t^2+ a_{2,1}z)^2
\left[\mu(\xi^2a_{2,1} - t\xi a_{4,n+1}+ t^2a_{6,2n+1})
+ (4a_{2,1}a_{6,2n+1} - a_{4,n+1}^2)z\right]z^{2n+1} + O(z^{2n+3}),
\end{equation}
since $3n\ge2n+3$.
The order of vanishing increases by at
least $2$ for this example, and we have $\betaSigma=\frac14\mu t^2$,
$\gammaSigma=\mu(\xi^2a_{2,1}  - t\xi a_{4,n+1}+ t^2a_{6,2n+1} )$.

Note that example 4 includes cases without monodromy, in
the generalized Weierstrass
form from \cite{geom-gauge}:
\[ y^2 + a_1 xy + a_{3,n}z^n y = x^3 + a_{2,1} zx^2 + a_{4,n+1} z^{n+1}x
+ a_{6,2n+1},\]
since, after completing the square, we see that this is the same as equation
\eqref{eq:generalized},
with $\mu=1$, $\nu=a_1$ and $\xi=a_{3,n}$.

\medskip

\noindent {\bf Example 5.}
Consider $I_{m-4}^*$ with $m\ge7$.  The generalized Weierstrass form
from \cite{geom-gauge} can be written as
\begin{equation} \label{eq:Im4star}
y^2=x^3 + a_{2,1}zx^2 + a_{4,[(m+1)/2]}z^{[(m+1)/2]}x + a_{6,m-1}z^{m-1},
\end{equation}
with discriminant
\begin{equation} \label{disc:Im4star}
\Delta = 4a_{2,1}^3a_{6,m-1}z^{m+2}
-a_{2,1}^2 a_{4,[(m+1)/2]}z^{2+2[(m+1)/2]} + O(z^{m+3}).
\end{equation}
For our example, we assume that $a_{2,1}$ does not vanish at $z=t=0$.
Then we can rewrite the equation in the form
\begin{equation} \label{eq:Im4starbis}
y^2=x^3 +
a_{2,1}z\left(x + \frac{a_{4,[(m+1)/2]}}{2a_{2,1}}z^{[(m-1)/2]}\right)^2
+ tz^{m-1} + O(z^m)
\end{equation}
where we have set $a_{6,m-1}=t+(a_{4,[(m+1)/2]}^2/4a_{2,1})z^{2-m+2[(m-1)/2]}
$.  It follows that the
 discriminant takes the form
\begin{equation} \label{disc:Im4starbis}
\Delta = 4a_{2,1}^3tz^{m+2}
+ O(z^{m+3}).
\end{equation}
Since the coefficients of $x^1$ and $x^0$ in \eqref{eq:Im4starbis}
are divisible by $z^3$ and $z^4$,
respectively, $\gammaSigma$ coincides with the coefficient of $x^2$ divided
by $z$, i.e., $\gammaSigma=a_{2,1}$.  Since $\deltaSigma=4a_{2,1}^3t$,
it follows that
\[ \betaSigma=\begin{cases} 4t & \text{if } m \text{ is odd} \\
-4a_{2,1}t & \text{if } m \text{ is even}
\end{cases}
\]

We summarize our first five examples in Table~\ref{tab:group1}.

\begin{table}[ht]
\begin{center}

\begin{tabular}{|c|c|c|l|l|l|l|c|} \hline
&Gen.&Spec.&Eqn.&Disc. & $\betaSigma$ & $\gammaSigma$ &$t$\\ \hline
1&$I_m$, $m\ge4$&$I_{m+1}$&\eqref{eq:Im1}&\eqref{disc:Im1}
&$a_2$&$4a_2t$ &$s$\\
2&$I_{2n}$, $n\ge2$&$I_{2n-4}^*$&\eqref{eq:Im2}&\eqref{disc:Im2}
&$t$&$4ta_{6,2n}-a_{4,n}^2$ & $s^2$\\
3&$I_{2n+1}$, $n\ge4$&$I_{2n-2}^*$&\eqref{eq:generalized3}&\eqref{disc:generalized3}
&$\frac14t\nu^2$&
$t\left(\xi^2a_{2,1}\right.$
 &$s^2$\\
&& &&
&&
$\quad \left.-\nu\xi a_{4,n+1} + \nu^2a_{6,2n+1}\right)$
& \\
4&$I_{2n+1}$, $n\ge3$&$I_{2n-3}^*$&\eqref{eq:generalized4}&\eqref{disc:generalized4}
 &$\frac14\mu t^2$&$\mu\left(\xi^2a_{2,1}\right.$
 &$s$\\
&& & &&&
$\quad \left.-t\xi a_{4,n+1} + t^2a_{6,2n+1}\right)$
&\\
5&$I_{m-4}^*$, $m\ge7$&$I_{m-3}^*$&\eqref{eq:Im4star}&\eqref{disc:Im4star}
&$4\alpha t$&$1$&$s^2$\\

\hline
\end{tabular}

\end{center}
\smallskip
\caption{First group of examples.
In \#5, $\alpha=1$ if $m$ is odd, and $\alpha=-a_{2,1}$ if
$m$ is even.}\label{tab:group1}
\end{table}

We now give 15 additional examples, whose special fibers are of types
$IV$, $I_0^*$, $IV^*$, $III^*$, or $II^*$.  The advantage of these types
is that the equation of the special fiber is precisely the equation
of an ADE singularity, and the example is part of the universal
deformation of that singularity.  (We choose the subspace of the
universal deformation where the singularity corresponding to the
general fiber is retained.)  In all but one case, the rank of the
corresponding Dynkin diagram jumps by $1$ between general and special
fiber, and the deformations we give are universal.
We display this second group of
examples in
Table~\ref{tab:examples2}.

\begin{table}[ht]
\begin{center}

\begin{tabular}{|c|c|c|l|c|} \hline
&Gen. & Spc. & Equation & $t$\\ \hline
6&$III$ & $IV$ & $y^2=x^3+2tzx+z^2$&$s^{1/2}$\\
7&$IV$ & $I_0^*$ & $y^2=x^3+az^2x+tz^2+bz^3$   &$s^2$\\
8&$I_6$ & $IV^*$ &$y^2= x^3 + (tx+z^2)^2$ &  $s^2$\\
9&$I_7$ & $III^*$ & $y^2  x^3 +(- 4tz + 16t^3)x^2+ (z^3 {-} 8t^2z^2)x {+} tz^4 $& $s^2$\\
10&$I_8$ & $II^*$ & $y^2=x^3+(3t^2z+t^5)x^2-(4tz^3+2t^4z^2)x+z^5+t^3z^4$ & $s^2$\\
11&$I_1^*$ & $IV^*$ & $y^2= x^3 + tzx^2 + z^4$&$s$\\
12&$I_2^*$ & $III^*$ &$y^2= x^3 + tzx^2 + z^3x$ &$s^2$\\
13&$I_3^*$ & $II^*$ & $y^2=x^3+z(tx+z^2)^2$&$s^2$\\
14&$IV^*$ & $III^*$ & $y^2=x^3+z^3x+tz^4$&$s^2$ \\
15&$III^*$ & $II^*$ & $y^2=x^3+tz^3x+z^5$&$s^2$\\
16&$I_2{+}I_5$ & $IV^*$ & $y^2=x^3+(4tz+t^4)x^2+2t^2z^2x+z^4$&$s$\\
17 & $I_2{+}I_1^*$ & $III^*$ &$y^2=x^3+tzx^2+z^3x+tz^4$ & $s^2$\\
18 & $I_2{+}IV^*$ & $II^*$ &$y^2=x^3{-}3tz^3x{+}z^5{+}t^3z^4$ &$s^2$ \\
19 & $I_3{+}I_5$ & $III^*$ &$y^2=x^3+(400t^3-15tz)x^2+(480t^4z-45t^2z^2+z^3)x$ & $s^2$  \\
 &  &  &$\quad +144t^5z^2-5t^3z^3$ &  \\
20&$I_2{+}I_3{+}I_3$ & $IV^*$ & $y^2= x^3 -3t^4 x^2 -6t^2z^2x+16t^6z^2+z^4$
&$s$\\
\hline
\end{tabular}

\end{center}
\smallskip
\caption{Second group of examples.
In \#7, $a$ and $b$ are constants.
}
\label{tab:examples2}
\end{table}

Note that the final column of the table
 indicates which substitution $t=s^k$ must
be made in order to get a map to the versal
simultaneous resolution space for the family.
For \textbf{Examples 7-20}, these are computed quite easily by using
the fact that we have a universal deformation of a $D_4$ or
an $E_n$ singularity.
Those singularities have a weighted homogeneous equation, and
have the property that
the homogeneity can be extended to the universal deformation by giving the
deformation parameter $t$ a weight (which is determined from the
other data).  These degrees of homogeneity are displayed in
Table~\ref{tab:homogeneous}.  In each case $t$ has degree $1$ or $2$,
and that restricts the allowed base changes.

\textbf{Example} $\mathbf{6}$ is slightly different: the special fiber
\[ y^2=x^3+z^2\]
has a singularity of type $A_2$
and is weighted homogeneous, but the versal deformation
\[ y^2=x^3-s x^2 + z^2\]
does not coincide with our deformation.  In fact, completing the square
in our example, we get
\[ y^2 = x^3 -  t^2 x^2 + (z+tx)^2,\]
which shows that $s=t^2$.

\begin{table}
\begin{center}
\begin{tabular}{|c|c|c|c|c|c|} \hline
Example &$x$&$y$&$z$&$t$&eqn.\\ \hline
6 & $2$ & $3$ & $3$ & $1$ & $6$ \\
7 & $2$ & $3$ & $2$ & $2$ & $6$ \\
8 & $4$ & $6$ & $3$ & $2$ & $12$ \\
9 & $6$ & $9$ & $4$ & $2$ & $18$ \\
10 & $10$ & $15$ & $6$ & $2$ & $30$ \\
11 & $4$ & $6$ & $3$ & $1$ & $12$ \\
12 & $6$ & $9$ & $4$ & $2$ & $18$ \\
13& $10$ & $15$ & $6$ & $2$ & $30$ \\
14& $6$ & $9$ & $4$ & $2$ & $18$ \\
15& $10$ & $15$ & $6$ & $2$ & $30$ \\
16 & $4$ & $6$ & $3$ & $1$ & $12$ \\
17& $6$ & $9$ & $4$ & $2$ & $18$ \\
18& $10$ & $15$ & $6$ & $2$ & $30$ \\
19& $6$ & $9$ & $4$ & $2$ & $18$ \\
20 & $4$ & $6$ & $3$ & $1$ & $12$ \\
\hline
\end{tabular}
\end{center}
\smallskip
\caption{Degrees of homogeneity degrees for examples 6-20}\label{tab:homogeneous}
\end{table}

\begin{table}[ht]
\begin{center}

\begin{tabular}{|c|c|c|c|c|c|} \hline
&Gen.\ Fib. &  Discriminant&Type of $\Sigma$&$\betaSigma$&$\gammaSigma$ \\ \hline
6&$III$ &$z^3(27z+32t^3)$&$III$&$\Box$ &$2t$\\
7&$IV$ & $z^4(dz^2{+}54btz{+}27t^2)$&$IV$&$t$&$\Box$
   \\
8&$I_6$ & $z^6(27z^2 - 4t^3)$&$I_6$&$t^2$&$-4t^{-1}$
  \\
9&$I_7$
& $z^7(4z^2{-}13 t^2z {+}32 t^4)$&$I_7$&$16t^3$&$\frac18t^{-2}$\\
10&$I_8$ &
$z^8(27z^{2}{+}14t^3z{+}3t^6)$ &$I_8$&$t^5$&$3t^{-4}$ \\
11&$I_1^*$ &$z^7(27z+4t^3)$ &$I_1^*$&$4$&$t$ \\
12&$I_2^*$ &$z^8(4z-t^2)$ &$I_2^*$&$1$&$t$\\
13&$I_3^*$ &$z^9(27z-4t^3)$
  &$I_3^*$&$-4t^{-3}$&$t^2$\\
14&$IV^*$ &$z^8(4z+27t^2)$ &$IV^*$&$t$&$\Box$\\
15&$III^*$ &$z^9(27z+4t^3)$ &$III^*$&$\Box$&$t$\\
16&$I_2{+}I_5$ &$(z{+}2t^3)^2z^5(27z{+}4t^3)$&$I_2$&$\Box$& $64t^{10}$\\
& & &$I_5$&$t^4$&$16t$\\
17 & $I_2{+}I_1^*$ &$4(z+t^2)^2z^7$ &$I_2$&$\Box$&$t^{8}$ \\
&&&$I_1^*$&$4t$&$t$ \\
18 & $I_2{+}IV^*$ &$27(z-t^3)^2z^8$ &$I_2$&$\Box$&$3t^{14}$ \\
&&&$IV^*$&$t^3$&$\Box$\\
19 & $I_3{+}I_5$ & $z^3(z{-}32t^2)^5(4z{-}125t^2)$&$I_3$& $400t^3$&$\frac{131072}5 t^6$\\
 &   & &$I_5$&$16t^3$&$-384t^2$ \\
20&$I_2{+}I_3{+}I_3$
&$27 z^2(z{-}2t^3)^3(z{+}2t^3)^3$&$I_2$&$\Box$&$-192t^{10}$\\
& & &$I_3$& $9t^4$&$\frac{256}3t^7$\\
& & &$I_3$& $9t^4$&$-\frac{256}3t^7$\\
\hline
\end{tabular}

\end{center}
\smallskip
\caption{Calculations for the second group of examples.
We use $\Box$ to indicate an undefined quantity.
In \#7, $a$ and $b$ are constants occurring in the equation, and
$d=4a^3+27b^2$.
In \#16--\#20, the factors in the discriminant and the
($\betaSigma$,$\gammaSigma$) pairs are listed in the same order
as the components of the general fiber in the left column.}
\label{tab:examples2bis}
\end{table}

Most of these  examples have appeared before in the
literature \cite{Katz:1996xe,LieF}, but
examples 10, 13, and 16--20
are new.  For each of these examples, we need to compute the
Weierstrass coefficients $f$ and $g$ and the discriminant
$\Delta=4f^3+27g^2$, which we factor as much as possible.
This factorization allows us to identify the components
$\Sigma_j$ of the discriminant locus which pass through
the point $\{z=t=0\}$, and for each of these, we compute the
quantities $\beta_{\Sigma_j}$ and $\gamma_{\Sigma_j}$ when defined
(following Table~\ref{tab:ab}).
All of these calculations are summarized in Table~\ref{tab:examples2bis}.
We describe these calculations in more detail
in a few of the more challenging
cases below, focusing on the examples which are new.

We begin with example 10, which has defining polynomial
\[y^2=x^3 +
(3t^2z+t^5)x^2
-(4tz^3+2t^4z^2
)x+
z^5+t^3z^4.
\]
We complete the cube using $x=\widetilde{x}-t^2z-\frac13t^5$, to obtain
\begin{align*}
y^2
&= \widetilde{x}^3
+\left(-4tz^3-5t^4z^2-2t^7z-\frac13t^{10}
\right)\widetilde{x}
\\&{}\quad
+\left(z^5+5t^3z^4+\frac{16}3t^6z^3+\frac83t^9z^2+\frac23t^{12}z+\frac2{27}t^{15}
\right),
\end{align*}
i.e., the Weierstrass coefficients are
\begin{align*}
f&= -4tz^3-5t^4z^2-2t^7z-\frac13t^{10}\\
g&=z^5+5t^3z^4+\frac{16}3t^6z^3+\frac83t^9z^2+\frac23t^{12}z+\frac2{27}t^{15} .
\end{align*}
Then the discriminant is:
\begin{align*}
&4\left(-4tz^3-5t^4z^2-2t^7z-\frac13t^{10}
\right)^3\\
&{}\quad
+27\left(z^5+5t^3z^4+\frac{16}3t^6z^3+\frac83t^9z^2+\frac23t^{12}z+\frac2{27}t^{15}
\right)^2\\
&= z^8(27z^{2}+14t^3z+3t^6).
\end{align*}
The relevant discriminant-component is $\{z=0\}$, along which
we have a fiber of type $I_8$.

To calculate the other invariants for this example, we begin with
\begin{align*}
f|_{\{z=0\}}&=-\frac13 t^{10}\\
g|_{\{z=0\}}&=\frac2{27} t^{15}\\
\delta_{\{z=0\}}&=3t^6.
\end{align*}
Thus $\beta_{\{z=0\}}=(-9g/2f)|_{\{z=0\}} = t^5$, and
$\gamma_{\{z=0\}} = \delta_{\{z=0\}}/\beta_{\{z=0\}}^2 = 3t^{-4}$.

Example 13
 is quite straightforward.  We start with the polynomial
\[ y^2 = x^3 + z(tx+z^2)^2 = x^3 + t^2zx^2 + 2tz^3x + z^5\]
and complete the cube using $x = \widetilde{x} - \frac13t^2z$
to obtain the Weierstrass coefficients and discriminant
\begin{align*}
f&=z^2(2tz-\frac13t^4)\\
g&=z^3(z^2-\frac23t^3z+\frac2{27}t^6)
\\ \Delta &
= z^9(27z-4t^3).
\end{align*}
The relevant discriminant-component is $\{z=0\}$, along which we
have a fiber of type $I_3^*$ (since $f$ and $g$ vanish to orders
$2$ and $3$, respectively).

To calculate the other invariants,
\begin{align*}
(f/z^2)|_{\{z=0\}}&=-\frac13 t^{4}\\
(g/z^3)|_{\{z=0\}}&=\frac2{27} t^{6}\\
\delta_{\{z=0\}}&=-4t^3.
\end{align*}
Thus $\gamma_{\{z=0\}}=(-9g/2zf)|_{\{z=0\}} = t^2$, and
$\beta_{\{z=0\}} = \delta_{\{z=0\}}/\gamma_{\{z=0\}}^3 = -4t^{-3}$.

Example 16 is the first case in which we get more than one relevant
discriminant-component.  We begin with the polynomial
\begin{equation} \label{ex:16b}
 {y}^2
= x^3 + (4tz+t^4)x^2 + 2t^2z^2x + z^4.
\end{equation}
We then complete the cube with the substitution
$x=\tilde{x}-\frac43tz-\frac13t^4$ to obtain Weierstrass coefficients
and discriminant:
\begin{align*}
f&=
 -\frac13(10t^2z^2+8t^5z+t^8)\\
g&=\frac1{27}(27z^4+56t^3z^3+78t^6z^2+24t^9z+2t^{12})\\
\Delta&= (z+2t^3)^2z^5(27z+4t^3).
\end{align*}
Since none of the factors of the discriminant divide $f$ or $g$, we
see that the Kodaira fiber types are indeed $I_2$, $I_5$, and $I_1$
along the three components of the discriminant locus.  (Since the third
of these does not contribute to the gauge algebra, we need not consider
it in our computations.)

To compute the invariants along $\{z=-2t^3\}$, we begin with
\begin{align*}
f|_{\{z=-2t^3\}}&= -\frac13(40-16+1)t^8 = -\frac{25}3t^8\\
g|_{\{z=-2t^3\}}&= \frac1{27}(432-448+312-48+2)t^{12}=\frac{250}{27}t^{12}\\
\delta_{\{z=-2t^3\}}&= (-2t^3)^5(-54t^3+4t^3) = 1600 t^{18}.
\end{align*}
Then
\begin{align*}
 u|_{\{z=-2t^3\}}&=(-9g/2f)|_{\{z=-2t^3\}}=5t^4\\
\gamma_{\{z=-2t^3\}}&=\delta_{\{z=-2t^3\}}/(u|_{\{z=-2t^3\}})^2 = 64t^{10}.
\end{align*}

On the other hand,
to compute the invariants along $\{z=0\}$, we begin with
\begin{align*}
f|_{\{z=0\}}&= -\frac13t^8\\
g|_{\{z=0\}}&= \frac2{27}t^{12}\\
\delta_{\{z=0\}}&= (2t^3)^2(4t^3)=16t^9.
\end{align*}
Thus, $\beta_{\{z=0\}}=(-9g/2f)|_{\{z=0\}}=t^4$,
and $\gamma_{\{z=0\}}=\delta_{\{z=0\}}/\beta_{\{z=0\}}^2=16t$.

For example 17, we complete the cube on
\[ y^2 = x^3+tzx^2+z^3x+tz^4\]
and find that
\begin{align*}
f &= z^2(z-\frac13t^2)\\
g &= z^3(\frac23tz+\frac2{27}t^3)\\
 \Delta &=   4(z+t^2)^2\,z^7,
\end{align*}
from which we can easily compute the invariants along the two components
of the discriminant.  For $\{z=-t^2\}$ we have type $I_2$ and
\[ u|_{\{z=-t^2\}} = \frac{-9g|_{\{z=-t^2\}}}{2f|_{\{z=-t^2\}}}
= \frac{-9\cdot (16/27)t^9}{-2\cdot (4/3) t^6} = 2t^3,\]
while $\delta_{\{z=-t^2\}} = 4t^{14}$ so that $\gamma_{\{z=-t^2\}}
= 4t^{14}/(2t^3)^2 = t^8$.

On the other hand, for $\{z=0\}$ we have type $I_1^*$ and
\[ \gamma_{z=0}=(u/z)|_{\{z=0\}} = \frac{-9\cdot (2/27)t^3}{2\cdot (-1/3)t^2}
=t.\]
Since $\delta_{\{z=0\}}=4t^4$, we have $\beta_{\{z=0\}}=4t^4/t^3=4t$.

Examples 18 and 19 are quite similar to these.

Example 20 is the most complicated one.
We start with
\begin{equation} \label{ex:20a}
 y^2 = x^3 -3t^4 x^2 -6t^2z^2x+16t^6z^2+z^4
\end{equation}
and complete the cube with $x=\tilde{x}+t^4$ to obtain
\begin{equation} \label{ex:20b}
 y^2 = \tilde{x}^3 + (-6t^2z^2-3t^8)\tilde{x} + (z^4+10t^6z^2-2t^{12}).
\end{equation}
Thus, the discriminant is
\[ 4(-6t^2z^2-3t^8)^3+27(z^4+10t^6z^2-2t^{12})^2
=27 z^2(z-2t^3)^3(z+2t^3)^3.\]
There are three relevant components of the discriminant locus, with
fibers of types $I_2$, $I_3$, and $I_3$.

We have
\begin{align*}
f|_{\{z=0\}} &= -3t^8\\
g|_{\{z=0\}} &= -2t^{12}\\
\delta_{\{z=0\}} &= -27\cdot64 t^{18},
\end{align*}
and so $u|_{\{z=0\}}= (-9g/2f)|_{\{z=0\}}=-3t^4$ which implies
that $\gamma_{\{z=0\}}=\delta_{\{z=0\}}/(u|_{\{z=0\}})^2=-192t^{10}$.

On the other hand,
\begin{align*}
f|_{\{z=\pm2t^3\}} &= (-24-3)t^8=-27t^8\\
g|_{\{z=\pm2t^3\}} &= (16+40-2)t^{12}=54t^{12}\\
\delta_{\{z=\pm2t^3\}} &= 27(\pm2t^3)^2(\pm4t^3)^3 = \pm 27\cdot 256 t^{15},
\end{align*}
so $\beta_{\{z=\pm2t^3\}} =
(-9g/2f)|_{\{z=\pm2t^3\}} =9t^4$
and $\gamma_{\{z=\pm2t^3\}} =
\delta_{\{z=\pm2t^3\}}/
\beta_{\{z=\pm2t^3\}}^2=\pm\frac{256}3t^{7}$.

\bigskip

Each of our
new examples was constructed using the data from
\cite{gorenstein-weyl}, specialized to the particular geometric situation
we were constructing, and
the starting point for each of these derivations
 was a subdiagram of a Dynkin diagram.
For
this reason, we expect that the matter representations associated to these
examples should be obtained from the algebra inclusion corresponding
to those Dynkin sub-diagrams.  In the one ambiguous case (example 10),
this leads to a precise prediction for the matter representation: it
should be $\fund \oplus \Lambda^2 \oplus \Lambda^3$.

To analyze the matter representation for each of these examples, we take the
inclusion of (simply-laced) Dynkin diagrams $R\subset R'$ which corresponds
to the degeneration in the example, and decompose the adjoint representation
of the Lie algebra $\mathfrak{g}(R')$ under the action of the
Lie algebra $\mathfrak{g}(R)$. The nontrivial constituents
 other than the adjoint
representation of $\mathfrak{g}(R)$, when restricted to the Lie
algebra $\mathfrak{g}$, comprise the corresponding
matter representation.  This is determined by means of
 ``branching rules'' for the adjoint
representation, which are
easily obtained from a reference such as \cite{Slansky:1981yr}
or \cite{MR604363}.
We collect the information we need in Table~\ref{tab:ex:branching}
(most of which was already presented in \cite{grouprep}).
The Table lists a representation $\rho_{R,R'}$ with the property
that $\adj_{\mathfrak{g}(R')}$ decomposes as a representation over
$\mathfrak{g}(R)$ into
\[ \adj_{\mathfrak{g}(R')} = \adj_{\mathfrak{g}(R)} \oplus \rho_{R,R'}
\oplus \overline{\rho_{R,R'}} \oplus
\mathbf{1}^{\oplus (\operatorname{rank}(R') - \operatorname{rank}(R))}.\]
(Note that if $\mathfrak{g}$ is not simply-laced then $\mathfrak{g}(R)$
is the covering algebra and the actual matter representation
is $\rho_{R,R'}|_{\mathfrak{g}}$.)
In the Table, we have included two entries for $A_7\subset E_8$, since there are
known to be two different algebra embeddings, and their representation
theory differs.

\begin{table}[ht]
\begin{center}

\begin{tabular}{|c|c|c|c|} \hline
&$R$ & $ R'$ & $\rho_{R,R'}$ \\ \hline
1&$A_{m-1}$&$A_m$ & $\fund$ \\
2&$A_{2n-1}$ & $D_{2n}$ & $\Lambda^2$\\
3& $A_{2n}$ & $D_{2n+2}$ & $\fund^{\oplus2}\oplus\Lambda^2$\\
4&$A_{2n}$& $D_{2n+1}$&$\Lambda^2$\\
5&$D_m$ & $D_{m+1}$ & $\vect$\\
6&$A_1$&$ A_{2}$ & $\fund$ \\
7&$A_2$&$D_4$ &$\fund^{\oplus 3}$\\
8&$A_5$&$ E_6$ &$\mathbf{1}\oplus\Lambda^3$\\
9&$A_6$&$ E_7$ &$\fund\oplus\Lambda^3$\\
10&$A_7$&$ E_8$ &$\fund\oplus\Lambda^2\oplus\Lambda^3$\\
10$'$&$A_7$&$ E_8$ &$\mathbf{1}\oplus(\Lambda^2)^{\oplus 2}\oplus\frac12\Lambda^4$\\
11&$D_5$&$ E_6$ &$\spinrep_+$\\
12&$D_6$&$ E_7$ &$\mathbf{1}\oplus\spinrep_+$\\
13&$D_7$&$ E_8$ &$\vect\oplus\spinrep_+$\\
14&$E_6$&$ E_7$ &$\mathbf{27}$\\
15&$E_7$&$ E_8$ &$\mathbf{1}\oplus\mathbf{56}$\\
16&$A_1{+}A_4$ & $E_6$ &
$(\mathbf{1}\otimes\fund)\oplus
(\fund\otimes\Lambda^2) $\\
17&$A_1+D_5$& $E_7$ &$(\mathbf{1}\otimes\vect)\oplus(\fund\otimes\spinrep_+)$\\
18&$A_1+E_6$& $E_8$
&$(\mathbf{1}\otimes\mathbf{27})\oplus(\fund\otimes\mathbf{1})\oplus(\fund\otimes\mathbf{27})$\\
19&$A_2+A_4$& $E_7$ &
$(\mathbf{1}\otimes\fund)\oplus
(\fund\otimes\fund)
\oplus
(\fund\otimes\Lambda^2)
$\\
20&$A_1{+}A_2{+}A_2$ & $E_6$ &
$
(\mathbf{1} \otimes
\fund\otimes\fund) \oplus (\fund\otimes\mathbf{1}\otimes\mathbf{1})
\oplus
(\fund \otimes \fund \otimes \fund)
$\\
\hline
\end{tabular}

\end{center}
\smallskip
\caption{Branching rules for the examples.}\label{tab:ex:branching}
\end{table}

In addition to the ambiguity of embedding $A_7$ into $E_8$, there is
a second ambiguity: there are two embeddings of
 $D_6$ into $E_7$ (or of $I_2^*$ into
$III$) one of which yields $\spinrep_+$ in the decomposition and the other
of which yields
$\spinrep_-$.  There is an isomorphism which maps one of these
representations to the other, and it can be realized in the
geometry as well.  That is, depending on how we label the roots
and weights of $D_6$ as geometric objects, either representation
can occur.

\section{Anomaly cancellation}\label{sec:about}

We now verify that each of our examples satisfies anomaly cancellation.  That is,
for each type of singular point, we compute the contributions to the
Tate cycles and compare them to the contribution to the matter cycle;
we also compute intersection numbers for all pairs of discriminant
components meeting at the point.

\begin{itemize}
\item
Example 1 has fibers of type $I_m$ along $\Sigma$
with no monodromy at the special point $t=0$, so it can
represent either $\mathfrak{su}(m)$ or $\mathfrak{sp}([m/2])$.  $\gammaSigma$
has a zero of order $1$ and $\betaSigma$ is nonvanishing, so by
consulting Table~\ref{tab:reps} we see that the Tate representation
is the fundamental representation
(for either gauge algebra, and any parity of $m$).
On the other hand, since no basechange is required for simultaneous
resolution (i.e., $k=1$), Table~\ref{tab:ex:branching} implies that
the matter representation is also the fundamental representation.

\item Example 2 has fibers of type $I_{2n}$ along $\Sigma$ with monodromy at the special
point.  Thus, this example is naturally associated with the gauge algebra
$\mathfrak{sp}(n)$,
and we have $\gammaSigma$ nonvanishing but $\betaSigma$ has a zero of
order $1$.
From Table~\ref{tab:reps} we see that the Tate representation
is $\frac12\Lambda^2_{\text{irr}}$.  On the other hand, $k=2$ so
Table~\ref{tab:ex:branching} shows that we get
the charged part of
$\frac12\Lambda^2|_{\mathfrak{sp}(n)}=\frac12\mathbf{1}\oplus
\frac12\Lambda^2_{\text{irr}}$
for the matter representation.  Discarding the uncharged part of the
representation, we find agreement.

Note that because of the monodromy, we expect
this representation to contain $\frac12\rho_0$
as a summand.  In fact, the representation {\em coincides}\/ with
 $\frac12\rho_0$ in this case.

\item If we make a basechange $t=s^2$ in Example 2, we get an example
without monodromy which is suitable for the gauge algebra $\mathfrak{su}(2n)$.
This time $\betaSigma$ vanishes to order $2$ while $\gammaSigma$ is
nonvanishing, so the Tate representation is $\Lambda^2$.  In this case
$k=1$ so by Table~\ref{tab:ex:branching} the matter representation
is also $\Lambda^2$.

\item Example 3 has fibers of type $I_{2n+1}$ along $\Sigma$ with monodromy at
the special point, so that this example is associated with
the gauge algebra $\mathfrak{sp}(n)$.
Both $\betaSigma$ and $\gammaSigma$ have zeros of order $1$.  Thus,
in Table~\ref{tab:reps} we must add together two representations to
find the Tate representation: $\frac12\Lambda^2_{\text{irr}}
+\frac32 \fund$.
On the other hand, since $k=2$, restricting the appropriate entry
from Table~\ref{tab:ex:branching} to $\mathfrak{sp}(n)$ we find
matter representation
\[(\fund_{\mathfrak{su}(2n+1)} + \frac12\Lambda^2_{\mathfrak{su}(2n+1)})|_{\mathfrak{sp}(n)}
=  \mathbf{1} + \fund + \frac12\cdot\mathbf{1}+\frac12\fund+ \frac12\Lambda^2_{\text{irr}}.
\]
Thus, the charged parts agree.

In this case we can write the representation in the form
$\frac12\rho_0 + \frac12\fund$, and since the fundamental representation
is quarternionic, $\frac12\fund$ indeed defines a pre-quaternionic
representation as expected.

\item Example 4 has fibers of type $I_{2n+1}$ along $\Sigma$ with no local monodromy at
the special point; thus, this example can be associated with either
the gauge algebra $\mathfrak{g}(\Sigma)=\mathfrak{su}(2n+1)$ or the gauge algebra
$\mathfrak{g}(\Sigma)=\mathfrak{sp}(n)$.
$\betaSigma$ vanishes to order $2$ and $\gammaSigma$ is nonvanishing,
so Table~\ref{tab:reps} tells us that the charged part of
the Tate representation is
$\Lambda^2|_{\mathfrak{g}(\Sigma)}$.  On the other hand, since $k=1$,
Table~\ref{tab:ex:branching} shows that the charged part of the matter
representation is also $\Lambda^2|_{\mathfrak{g}(\Sigma)}$.

\item Example 5 has fibers of type $I_{m-4}^*$ along $\Sigma$ with local monodromy at
the special point, so it must represent
the gauge algebra $\mathfrak{so}(2m-1)$.
  $\betaSigma$ has a zero of order $1$ and
$\gammaSigma$ is nonvanishing, so by
consulting Table~\ref{tab:reps} we see that the Tate representation
is $\frac12\vect$.  On the other hand, $k=2$ so by Table~\ref{tab:ex:branching}
we find that the matter representation is also $\frac12\vect$.
As in Example~2, this representation also coincides
with $\frac12\rho_0$.

\item If we make a basechange $t=s^2$ in Example 5, we get an example
without monodromy which is suitable for the gauge algebra $\mathfrak{so}(2m)$.
This time $\betaSigma$ vanishes to order $2$ while $\gammaSigma$ is
nonvanishing, so the Tate representation is the vector representation.
In this case
$k=1$ so by Table~\ref{tab:ex:branching} the matter representation
is also the vector representation.

\item Example 6 has fibers of type $III$ along $\Sigma$ with no local monodromy, and
the gauge algebra is $\mathfrak{su}(2)$.  $\gammaSigma$ has a zero of
order $1$, so by Table~\ref{tab:reps}, the Tate representation is
$2\cdot\fund$.  On the other hand, since $k=\frac12$, from
Table~\ref{tab:ex:branching} we see that the matter representation is
also $2\cdot\fund$.

\item Example 7 has fibers of type $IV$ along $\Sigma$ with local monodromy,
and so is appropriate for the gauge algebra $\mathfrak{sp}(1)$.
$\betaSigma$ has a zero of order $1$, so by Table~\ref{tab:reps}, the Tate representation is
$\frac32\cdot\fund$.  On the other hand, since $k=2$, from
Table~\ref{tab:ex:branching} we see that the matter representation is
also $\frac32\cdot\fund$.  Once again, we have monodromy
and this coincides with $\frac12\rho_0$.

\item If we make a basechange $t=s^2$ in Example 7, we get an example
without monodromy which is suitable for the gauge algebra $\mathfrak{su}(3)$.
This time $\betaSigma$ has a zero of order $2$,
so the Tate representation is $3\cdot \fund$.
In this case
$k=1$ so by Table~\ref{tab:ex:branching} the matter representation
is also $3\cdot\fund$.

\item Example 8 has fibers of type $I_6$ along $\Sigma$ without local monodromy,
and so is appropriate for either the gauge algebra $\mathfrak{g}(\Sigma)=\mathfrak{su}(6)$
or the gauge algebra $\mathfrak{g}(\Sigma)=\mathfrak{sp}(3)$.
$\betaSigma$ has a zero of order $2$, and $\gammaSigma$ has a pole of
order $1$, so the Tate representation is determined as the difference
of two entries in Table~\ref{tab:reps}: it is the charged part of
\begin{equation} \label{eq:Tate8}
\Lambda^2|_{\mathfrak{g}(\Sigma)} - \fund|_{\mathfrak{g}(\Sigma)}.
\end{equation}
  On the other hand, since $k=2$, we see from
Table~\ref{tab:ex:branching} that the matter representation is the
charged part of
$\frac12\Lambda^3|_{\mathfrak{g}(\Sigma)}$, that is,
\begin{equation} \label{eq:matter8}
 \rho_{\text{matter}} = \begin{cases} \frac12\Lambda^3
& \text{if } \mathfrak{g}(\Sigma)=\mathfrak{su}(6)\\
\frac12\Lambda^3_{\text{irr}} + \frac12\fund
& \text{if } \mathfrak{g}(\Sigma)=\mathfrak{sp}(3)
\end{cases}.
\end{equation}
Note that $\Lambda^3$ is a quaternionic representation of $\mathfrak{su}(6)$,
and that both $\Lambda^3_{\text{irr}}$ and $\fund$ are quaternionic
representations of $\mathfrak{sp}(3)$, so this matter representation
is pre-quaternionic.

  To verify anomaly cancellation
in this case, we need to use Casimir equivalence: by the first
and fifth lines in Table~\ref{tab:casimir-equivalences}, the
Tate representation \eqref{eq:Tate8} is Casimir equivalent to
the matter representation \eqref{eq:matter8} for both gauge algebras.

\item Example 9 has fibers of type $I_7$ along $\Sigma$ with local monodromy,
and so is appropriate for
 the gauge algebra $\mathfrak{g}(\Sigma)=\mathfrak{sp}(3)$.
$\betaSigma$ has a zero of order $3$, and $\gammaSigma$ has a pole of
order $2$, so the Tate representation is determined as the difference
of two entries in Table~\ref{tab:reps}: it is
\begin{equation} \label{eq:Tate9}
\frac32 \Lambda^2_{\text{irr}} + \frac32 \fund - 2\fund =
\frac32 \Lambda^2_{\text{irr}} - \frac12 \fund.
\end{equation}
  On the other hand, since $k=2$, we see from
Table~\ref{tab:ex:branching} that the matter representation is
\begin{equation} \label{eq:matter9}
(\frac12\fund_{\mathfrak{su}(7)} + \frac12 \Lambda^3_{\mathfrak{su}(7)})|_{\mathfrak{sp}(3)}
= \fund +\frac12\Lambda^2_{\text{irr}}+ \frac12 \Lambda^3_{\text{irr}}
= \frac12\rho_0 + \frac12\Lambda^3_{\text{irr}}.
\end{equation}
Since $\Lambda^3_{\text{irr}}$ is quaternionic, $\frac12\Lambda^3_{\text{irr}}$
does indeed define a pre-quaternionic representation as expected.

  To verify anomaly cancellation,
we again employ Table~\ref{tab:casimir-equivalences}
to conclude that $\frac12\Lambda^3_{\text{irr}}$ is Casimir
equivalent to $\Lambda^2_{\text{irr}}-\frac32\fund$.  Thus,
the matter representation is Casimir equivalent to
\[ \fund + \frac12\Lambda^2_{\text{irr}} + \Lambda^2_{\text{irr}}-\frac32\fund
= \frac32\Lambda^2_{\text{irr}}-\frac12\fund,\]
i.e., to the Tate representation for $\Sigma$.

\item If we make a basechange $t=s^2$ in Example 9, we get an example
without monodromy which is suitable for the gauge algebra $\mathfrak{su}(7)$.
This time $\betaSigma$ has a zero of order $6$, and $\gammaSigma$ has
a pole of order $4$, so the Tate representation is determined as the difference
$3\cdot \Lambda^2 - 4 \cdot\fund$.
In this case
$k=1$ so by Table~\ref{tab:ex:branching} the matter representation
is  $\fund + \Lambda^3$.
Now the second line in Table~\ref{tab:casimir-equivalences} shows
that these are Casimir equivalent.

\item Example 10 has fibers of type $I_8$ along $\Sigma$ with local monodromy,
and so is appropriate for
 the gauge algebra $\mathfrak{g}(\Sigma)=\mathfrak{sp}(4)$.
$\betaSigma$ has a zero of order $5$, and $\gammaSigma$ has a pole of
order $4$, so the Tate representation is determined as the difference
of two entries in Table~\ref{tab:reps}: it is the charged part of
\begin{equation} \label{eq:Tate10}
\frac52 \Lambda^2_{\text{irr}}- 4\fund
\end{equation}
  On the other hand, since $k=2$, we see from
Table~\ref{tab:ex:branching} that the matter representation takes
one of two forms: it is either (1) the
charged part of
\begin{equation} \label{eq:matter10}
\frac12(\fund + \Lambda^2 + \Lambda^3)|_{\mathfrak{sp}(4)}
= \frac12\cdot\mathbf{1} + \fund + \frac12\Lambda^2_{\text{irr}} +
\frac12\Lambda^3_{\text{irr}},
\end{equation}
which can be written in the form
\begin{equation} \label{eq:matter10bis}
\frac12\rho_0 + \fund + \frac12\Lambda^3_{\text{irr}},
\end{equation}
or (2) the charged part of
\begin{equation} \label{eq:matter10alt}
\frac12(\mathbf{1} + 2\cdot \Lambda^2 + \frac12 \Lambda^4)|_{\mathfrak{sp}(4)}
= \frac74\cdot\mathbf{1} + \frac54\Lambda^2_{\text{irr}} +
 \frac14\Lambda^4_{\text{irr}},
\end{equation}
which can be written in the form
\begin{equation} \label{eq:matter10altbis}
\frac12\rho_0 + \frac34\Lambda^2_{\text{irr}} +
 \frac14\Lambda^4_{\text{irr}}.
\end{equation}
The second case does not leave us with a pre-quaternionic representation,
so does not appear to be valid choice for the matter representation
(although it is Casimir equivalent to the Tate representation, as
follows from the last line of Table~\ref{tab:casimir-equivalences}).

On the other hand, since $\Lambda^3_{\text{irr}}$ is quaternionic,
$\fund +\frac12\Lambda^3_{\text{irr}}$ defines a
pre-quaternionic representation.
Moreover, anomaly cancellation works because $\frac12\Lambda^3_{\text{irr}}$
is Casimir equivalent to $2\cdot \Lambda^2_{\text{irr}}-5\cdot \fund$.

\item If we make a basechange $t=s^2$ in Example 10, we get an example
without monodromy which is suitable for the gauge algebra $\mathfrak{su}(8)$.
This time $\betaSigma$ has a zero of order $10$, and $\gammaSigma$ has
a pole of order $8$, so the Tate representation is determined as the difference
$5\cdot \Lambda^2 - 8 \cdot\fund$.
In this case
$k=1$ so by Table~\ref{tab:ex:branching} the matter representation
is the charged part of either $\fund + \Lambda^2 + \Lambda^3$ or
$\mathbf{1} + 2\cdot \Lambda^2 + \frac12\Lambda^4$.
Now the third line in Table~\ref{tab:casimir-equivalences} shows
that $\Lambda^3$ is Casimir equivalent to $4\cdot\Lambda^2-9\cdot\fund$,
from which the anomaly cancellation follows in the first case;
the fourth line shows that $\Lambda^4$ is Casimir equivalent to
$6\cdot\Lambda^2.-16\cdot\fund$, from which the anomaly
cancellation follows in the second case.  Thus, either choice seems
{\em a priori}\/ possible, but as we stated near the end of
Section~\ref{sec:examples}, our method of
construction of this example strongly suggests that the correct matter
representation is $\fund + \Lambda^2 + \Lambda^3$.

\item Example 11 has fibers of type $I^*_1$ along $\Sigma$ without local monodromy,
and so is appropriate either for the gauge algebra $\mathfrak{so}(9)$ or
for the gauge algebra $\mathfrak{so}(10)$.
$\betaSigma$ is nonvanishing, and $\gammaSigma$ has a zero of order $1$, so by Table~\ref{tab:reps}, the Tate representation is
$\spinrep_*$.
On the other hand, since $k=1$, from
Table~\ref{tab:ex:branching} we see that the matter representation is
also $\spinrep_*$.  (More precisely the matter representation in the case
of $\mathfrak{so}(10)$ can be written as either $\spinrep_+$ or $\spinrep_-$,
which give the same quaternionic representation since they are complex
conjugates of each other.  The representation for $\mathfrak{so}(9)$
is the restriction of either of these to $\mathfrak{so}(9)$, which
is $\spinrep$ in either case.)

\item Example 12 has fibers of type $I^*_2$ along $\Sigma$ without local monodromy,
and so is appropriate either for the gauge algebra $\mathfrak{so}(11)$ or
for the gauge algebra $\mathfrak{so}(12)$.
$\betaSigma$ is nonvanishing, and $\gammaSigma$ has a zero of order $1$, so by Table~\ref{tab:reps}, the Tate representation is
$\frac12\spinrep_*$.
On the other hand, since $k=2$, from
Table~\ref{tab:ex:branching} we see that the charged part of
the matter representation is
also $\frac12\spinrep_*$.
Note that $\spinrep_*$ is quaternionic so that $\frac12\spinrep_*$ is
a pre-quaternionic representation.

\item Example 13 has fibers of type $I_3^*$ along $\Sigma$ with local monodromy,
and so is appropriate for the gauge algebra $\mathfrak{so}(13)$.
$\betaSigma$ has a pole of order $3$, and $\gammaSigma$ has a zero of order $2$,
so the Tate representation is a difference of entries in Table~\ref{tab:reps}:
\[-\frac32\cdot\vect + \frac12\cdot\spinrep_* + 2\cdot\vect
= \frac12\cdot\spinrep_* + \frac12\cdot\vect .\]
On the other hand, since $k=2$, from
Table~\ref{tab:ex:branching} we see that the matter representation is
also
\[
\frac12\cdot\vect+\frac12\cdot\spinrep_* = \frac12\rho_0 + \frac12 \spinrep_*.\]
Since $\spinrep_*$ is quaternionic, $\frac12\spinrep_*$
defines a pre-quaternionic representation as expected.

\item If we make a basechange $t=s^2$ in Example 13, we get an example
without monodromy which is suitable for the gauge algebra $\mathfrak{so}(14)$.
This time $\betaSigma$ has a pole of order $6$, and $\gammaSigma$ has a zero of order $4$,
so the Tate representation is a difference of entries in Table~\ref{tab:reps}:
\[-3\cdot\vect + \spinrep_* + 4\cdot\vect
= \vect+\spinrep_*  .\]
In this case
$k=1$ so by Table~\ref{tab:ex:branching} the matter representation
is also $\vect+\spinrep_*$.

\item Example 14 has fibers of type $IV^*$ along $\Sigma$ with local monodromy,
and so is appropriate for the gauge algebra $\mathfrak{f}_4$.
$\betaSigma$ has a zero of order $1$, so by Table~\ref{tab:reps},
the Tate representation is $\frac12\cdot \mathbf{26}$.
On the other hand, since $k=2$, from
Table~\ref{tab:ex:branching} we see that the
the matter representation is the charged part of
\[\frac12\cdot \mathbf{27}|_{\mathfrak{f}_4}= \frac12\cdot\mathbf{1}
+ \frac12\cdot\mathbf{26},\]
the same as the Tate representation.  Again, as is typical with
monodromy, this representation coincides with $\frac12\rho_0$.

\item If we make a basechange $t=s^2$ in Example 14, we get an example
without monodromy which is suitable for the gauge algebra $\mathfrak{e}_6$.
This time $\betaSigma$ has a zero of order $2$, so by Table~\ref{tab:reps},
the Tate representation is $ \mathbf{27}$.
In this case
$k=1$ so by Table~\ref{tab:ex:branching} the matter representation
is also $\mathbf{27}$,
the same as the Tate representation.

\item Example 15 has fibers of type $III^*$ along $\Sigma$ without local monodromy,
and is appropriate for the gauge algebra $\mathfrak{e}_7$.
$\betaSigma$ is nonvanishing, and $\gammaSigma$ has a zero of order $1$, so by Table~\ref{tab:reps}, the Tate representation is
$\frac12\cdot\mathbf{56}$.
On the other hand, since $k=2$, from
Table~\ref{tab:ex:branching} we see that the matter representation is
also $\frac12\mathbf{56}$.  Since $\mathbf{56}$ is a quaternionic representation,
this is pre-quaternionic.

\item Example 16 has fibers of type $I_2$ along $\Sigma$
and fibers of type $I_5$ along $\Sigma'$, which meet at a common
point.  There is no local monodromy along $\Sigma$, so
either gauge algebra 
$\mathfrak{g}(\Sigma)\oplus\mathfrak{g}(\Sigma')
=\mathfrak{su}(2)\oplus\mathfrak{sp}(2)$
or $\mathfrak{g}(\Sigma)\oplus\mathfrak{g}(\Sigma')
=\mathfrak{su}(2)\oplus\mathfrak{su}(5)$ is possible; we
will make the calculation for the second one.
We need to calculate the Tate representation for each component,
as well as the representation-multiplicity for the pair: the
data for this is contained in Table~\ref{tab:examples2bis}.
Along $\Sigma$,  $\gamma_{\Sigma}$ has a zero of order 10,
so the Tate representation for $\Sigma$
 is $10\cdot\fund$ (a representation of $\mathfrak{su}(2)$).
On the other hand,
along $\Sigma'$, $\beta_{\Sigma'}$ has a zero of order 4 and
$\gamma_{\Sigma'}$
has a zero of order 1, so adding the corresponding entries in
Table~\ref{tab:reps} we see that the Tate representation
for $\Sigma'$ is  $2\Lambda^2 + \fund$
(a representation of $\mathfrak{su}(5)$).

The expression for the discriminant of this example from
Table~\ref{tab:examples2bis} shows that
the components $\Sigma$ and $\Sigma'$ have intersection multiplicity $3$;
thus, we need
\[  \frac{ \Tr_\rho (F_{\mathfrak{g}(\Sigma)}^2)
\Tr_\rho (F_{\mathfrak{g}(\Sigma')}^2)}
{\tr_{\mathfrak{g}(\Sigma)}( F^2) \tr_{\mathfrak{g}(\Sigma')}( F^2)} = 3,\]
for the representation $\rho$ associated to our special point $z=t=0$.

Now we turn to representation theory to compute the actual matter
representation.  As indicated in Table~\ref{tab:ex:branching},
since $k=1$, the matter representation is
\[ \rho=(\mathbf{1}\otimes\fund)\oplus(\fund\otimes\Lambda^2).\]
There are three things to check in order to verify anomaly cancellation:
\begin{enumerate}
\item
When the matter representation is restricted to
$\mathfrak{g}(\Sigma)=\mathfrak{su}(2)$, the
charged part is 10
copies of the fundamental, since $\Lambda^2_{\mathfrak{su}(5)}$
has dimension $10$.  This agrees with the Tate representation for $\Sigma$.
\item
When restricted to
$\mathfrak{g}(\Sigma')=\mathfrak{su}(5)$, the
charged part is 2 copies of $\Lambda^2$ plus a fundamental, since
the fundamental of $\mathfrak{su}(2)$ has dimension $2$.  This agrees
with the Tate representation for $\Sigma'$.
\item
Since
\begin{align*}
\Tr_{\fund_{\mathfrak{su}(2)}} F^2 &= \tr_{\mathfrak{su}(2)}F^2\\
\Tr_{\Lambda^2_{\mathfrak{su}(5)}} F^2 &= 3 \tr_{\mathfrak{su}(5)} F^2,
\end{align*}
we have
$\mu_{\rho}(\mathfrak{g}(\Sigma),\mathfrak{g}(\Sigma')) = 1 \cdot 3 = 3$,
so the
representation-multiplicity coincides with the intersection multiplicity.
\end{enumerate}
Thus, the anomalies cancel.

\item Example 17
 has fibers of type $I_2$ along $\Sigma$
and fibers of type $I_1^*$ along $\Sigma'$, which meet at a common
point.  There is local monodromy along $\Sigma'$, so
the gauge algebra must be $\mathfrak{su}(2)\oplus\mathfrak{so}(9)$.
According to
Table~\ref{tab:examples2bis},
along $\Sigma$,  $\gamma_{\Sigma}$ has a zero of order 8,
so the Tate representation for $\Sigma$
 is $8\cdot\fund$.
On the other hand,
along $\Sigma'$, $\beta_{\Sigma'}$ and
$\gamma_{\Sigma'}$
each have a zero of order 1, so adding the corresponding entries in
Table~\ref{tab:reps} we see that the Tate representation
for $\Sigma'$ is $\frac12\vect\oplus\spinrep_*
= \frac12\rho_0 \oplus \spinrep_*$.  (We expect the $\frac12\rho_0$
summand since this is a branch point for $\mathfrak{so}(9)$.)
Also,
the expression for the discriminant of this example from
Table~\ref{tab:examples2bis} shows that
the components $\Sigma$ and $\Sigma'$ have intersection multiplicity $2$.

Now we turn to representation theory to compute the actual matter
representation.  As indicated in Table~\ref{tab:ex:branching},
since $k=2$, the matter representation is
\[ \rho=
\left.\left(\frac12(\mathbf{1}\otimes\vect)\oplus\frac12(\fund\otimes\spinrep_*)\right)\right|_{\mathfrak{su}(2)\oplus\mathfrak{so}(9)}.\]
The first summand is $\frac12\rho_0$ for $\mathfrak{so}(9)$, and
the second summand is pre-quaternionic, since $\fund$ is a quaternionic
representation of $\mathfrak{su}(2)$.

There are three things to check in order to verify anomaly cancellation:
\begin{enumerate}
\item
When the matter respresentation is restricted to
$\mathfrak{g}(\Sigma)=\mathfrak{su}(2)$, the
charged part is 8
copies of the fundamental, since the spinor representation
of $\mathfrak{so}(9)$ has dimension $16$ and we have half of that.
  This agrees with the Tate representation for $\Sigma$.
\item
When restricted to
$\mathfrak{g}(\Sigma')=\mathfrak{so}(9)$, the
charged part is $\frac12\vect\oplus\spinrep$, since $\frac12\fund$
has dimension $1$ as a representation of $\mathfrak{su}(2)$.
This agrees
with the Tate representation for $\Sigma'$.
\item
To evaluate the representation-multiplicity, we use the scaling property
of that quantity.  Thus, for a half-representation, we get half of
the representation-multiplicity of the corresponding full representation.
In our case, since
\begin{align*}
\Tr_{\fund_{\mathfrak{su}(2)}} F^2 &= \tr_{\mathfrak{su}(2)}F^2\\
\Tr_{\spinrep_{\mathfrak{so}(9)}} F^2 &= 4 \tr_{\mathfrak{so}(9)} F^2,
\end{align*}
we have $\mu_{\rho}(\mathfrak{g}(\Sigma),\mathfrak{g}(\Sigma')) =\frac12( 1 \cdot 4) = 2$,
and we see that the
representation-multiplicity coincides with the intersection multiplicity.
\end{enumerate}
Thus, the anomalies cancel.

A similar analysis applies after we make a basechange $t=s^2$,
obtaining a representation for $\mathfrak{su}(2)\oplus\mathfrak{so}(10)$
with $k=1$.  We omit the details.

\item Example 18
 has fibers of type $I_2$ along $\Sigma$
and fibers of type $IV^*$ along $\Sigma'$, which meet at a common
point.  There is local monodromy along $\Sigma'$, so
the gauge algebra must be $\mathfrak{su}(2)\oplus\mathfrak{f}_4$.
According to
Table~\ref{tab:examples2bis},
along $\Sigma$,  $\gamma_{\Sigma}$ has a zero of order 14,
so the Tate representation for $\Sigma$
 is $14\cdot\fund$.
On the other hand,
along $\Sigma'$, $\beta_{\Sigma'}$ has a zero of order $3$, so that
the Tate representation
for $\Sigma'$ is $\frac32\mathbf{26}
= \frac12\rho_0 \oplus \mathbf{26}$.  (We expect the $\frac12\rho_0$
summand since this is a branch point for $\mathfrak{f}_4$.)
Also,
the expression for the discriminant of this example from
Table~\ref{tab:examples2bis} shows that
the components $\Sigma$ and $\Sigma'$ have intersection multiplicity $3$.

 As indicated in Table~\ref{tab:ex:branching},
since $k=2$, and using the fact that $\mathbf{27}|_{\mathfrak{f}_4}=\mathbf{1}\oplus\mathbf{26}$, the matter representation is
\[ \rho=
(\fund\otimes\mathbf{1})\oplus
\frac12(\mathbf{1}\otimes\mathbf{26})\oplus
\frac12(\fund\otimes\mathbf{26})
.\]
The middle summand is $\frac12\rho_0$ for $\mathfrak{f}_4$, and
the last summand is pre-quaternionic, since $\fund$ is a quaternionic
representation of $\mathfrak{su}(2)$.

There are three things to check in order to verify anomaly cancellation:
\begin{enumerate}
\item
When the matter representation is restricted to
$\mathfrak{g}(\Sigma)=\mathfrak{su}(2)$, the
charged part is 14
copies of the fundamental, since the total dimension of the
representations on the $\mathfrak{f}_4$ side is $2+26=28$
and we have half of that.
  This agrees with the Tate representation for $\Sigma$.
\item
When restricted to
$\mathfrak{g}(\Sigma')=\mathfrak{f}_4$, the
charged part is $\frac32\mathbf{26}$, since the total dimension of
the representations on the $\mathfrak{su}(2)$ side is $1+2=3$
and we have half of that.
This agrees
with the Tate representation for $\Sigma'$.
\item
To evaluate the representation-multiplicity, we again use the scaling property
of that quantity to evaluate it for a half-representation.
We have
\begin{align*}
\Tr_{\fund_{\mathfrak{su}(2)}} F^2 &= \tr_{\mathfrak{su}(2)}F^2\\
\Tr_{\mathbf{26}_{\mathfrak{f}_4}} F^2 &= 6 \tr_{\mathfrak{f}_4} F^2,
\end{align*}
so that $\mu_{\rho}(\mathfrak{g}(\Sigma),\mathfrak{g}(\Sigma')) =\frac12( 1 \cdot 6) = 3$.
Hence, the
representation-multiplicity coincides with the intersection multiplicity.
\end{enumerate}
Thus, the anomalies cancel.

A similar analysis applies after we make a basechange $t=s^2$,
obtaining a representation for $\mathfrak{su}(2)\oplus\mathfrak{e}_6$
with $k=1$.  We omit the details.

\item Example 19
 has fibers of type $I_3$ along $\Sigma$
and fibers of type $I_5$ along $\Sigma'$, which meet at a common
point.  There is local monodromy along both $\Sigma$ and $\Sigma'$, so
the gauge algebra must be $\mathfrak{sp}(1)\oplus\mathfrak{sp}(2)$.
According to
Table~\ref{tab:examples2bis},
along $\Sigma$,  $\betaSigma$ has a zero of
order $3$ and $\gamma_{\Sigma}$ has a zero of order 6,
so the Tate representation for $\Sigma$
 is the sum of two terms from Table~\ref{tab:reps},
totaling $\frac{15}2\cdot\fund = \frac12\rho_0 + \frac{13}2\fund$.
(We expect the $\frac12\rho_0$
summand since this is a branch point for $\mathfrak{sp}(1)$;
the remaining representation $\frac{13}2\fund$ is pre-quaternionic.)
On the other hand,
along $\Sigma'$, $\beta_{\Sigma'}$ has a zero of order $3$ and
$\gamma_{\Sigma'}$ has
 a zero of order 2, so adding the corresponding entries in
Table~\ref{tab:reps} we see that the Tate representation
for $\Sigma'$ is $\frac32\Lambda^2_{\text{irr}}\oplus\frac72\fund
= \frac12\rho_0 \oplus \Lambda^2_{\text{irr}} \oplus \frac52\fund$.
(We expect the $\frac12\rho_0$
summand since this is a branch point for $\mathfrak{sp}(4)$;
the remaining representation $\Lambda^2_{\text{irr}} \oplus \frac52\fund$
is pre-quaternionic.))
Also,
the expression for the discriminant of this example from
Table~\ref{tab:examples2bis} shows that
the components $\Sigma$ and $\Sigma'$ have intersection multiplicity $2$.

Now we turn to representation theory to compute the actual matter
representation.  As indicated in Table~\ref{tab:ex:branching},
since $k=2$, and using the facts that
\begin{align*}
\fund_{\mathfrak{su}(3)}|_{\mathfrak{sp}(1)} &= \mathbf{1}\oplus\fund \\
\fund_{\mathfrak{su}(5)}|_{\mathfrak{sp}(2)} &= \mathbf{1}\oplus\fund \\
\Lambda^2_{\mathfrak{su}(5)}|_{\mathfrak{sp}(2)} &= \mathbf{1}\oplus
\fund\oplus\Lambda^2_{\text{irr}},
\end{align*}
the matter representation is the charged part of
\[
\left.\left(
\frac12(\mathbf{1}\otimes\fund)
\oplus\frac12(\fund\otimes\fund)
\oplus\frac12(\fund\otimes\Lambda^2)\right)
\right|_{\mathfrak{sp}(1)\oplus\mathfrak{sp}(2)},
\]
which equals
\[ \frac32(\mathbf{1}\otimes\fund)\oplus (\fund\otimes\mathbf{1})
\oplus (\fund\otimes\fund)
\oplus \frac12(\mathbf{1}\otimes\Lambda^2_{\text{irr}})
\oplus \frac12(\fund\otimes\Lambda^2_{\text{irr}})
.\]
Subtracting $\frac12\rho_0(\mathfrak{sp}(1))+\frac12\rho_0(\mathfrak{sp}(2))=
(\fund\otimes\mathbf{1})\oplus\frac12(\mathbf{1}\otimes\Lambda^2_{\text{irr}})
\oplus(\mathbf{1}\otimes\fund)
$
leaves a pre-quaternionic representation.

There are three things to check in order to verify anomaly cancellation:
\begin{enumerate}
\item
When the matter representation is restricted to
$\mathfrak{g}(\Sigma)=\mathfrak{sp}(1)$, the
charged part is $1 + \frac12\dim (\Lambda^2_{\text{irr}}) + \dim(\fund_{\mathfrak{sp}(2)})=1+ \frac52+4=\frac{15}2$
copies of the fundamental.
  This agrees with the Tate representation for $\Sigma$.
\item
When restricted to
$\mathfrak{g}(\Sigma')=\mathfrak{sp}(2)$, the
charged part is $\frac32+\dim(\fund_{\mathfrak{sp}(1)})= \frac72$
copies of the fundamental representation plus
$\frac12+\frac12\dim(\fund_{\mathfrak{sp}(1)})=\frac32$
copies of $\Lambda^2_{\text{irr}}$.
This agrees
with the Tate representation for $\Sigma'$.
\item
To evaluate the representation-multiplicity, we use the linearity property
of that quantity, getting contributions from two different irreducible
representations of $\mathfrak{sp}(1)\oplus\mathfrak{sp}(2)$.
We have
\begin{align*}
\Tr_{\fund_{\mathfrak{sp}(1)}} F^2 &= \tr_{\mathfrak{sp}(1)}F^2\\
\Tr_{\fund_{\mathfrak{sp}(2)}} F^2 &= \tr_{\mathfrak{sp}(2)}F^2\\
\Tr_{\Lambda^2_{\text{irr},\mathfrak{sp}(2)}} F^2 &= 2\tr_{\mathfrak{sp}(2)}F^2,
\end{align*}
and so
\begin{align*}
\mu_{\frac12(\fund\otimes\Lambda^2_{\text{irr}})}(\mathfrak{g}(\Sigma),\mathfrak{g}(\Sigma'))&= \frac12
(1\cdot2)=1\\
\mu_{\fund\otimes\fund}(\mathfrak{g}(\Sigma),\mathfrak{g}(\Sigma')) &= 1\cdot 1=1
\end{align*}
so the total representation multiplicity is $1+1=2$.  This
coincides with the intersection multiplicity.
\end{enumerate}
Thus, the anomalies cancel.

A similar analysis applies after we make a basechange $t=s^2$,
obtaining a representation for $\mathfrak{su}(3)\oplus\mathfrak{su}(5)$
with $k=1$.  We omit the details.

\item Finally, example 20 has fibers of type $I_3$ along $\Sigma_+$
and $\Sigma_-$, and fibers of type $I_2$ along $\Sigma$, which all
meet at a common point.  There is no local or global monodromy
for any of these, so the natural gauge algebra to consider is
$\mathfrak{su}(2)\oplus \mathfrak{su}(3)\oplus\mathfrak{su}(3)$.
We need to calculate the Tate representation of each component,
as well as the representation-multiplicity  for each pair
of components.
Along $\Sigma$, $\gammaSigma$ has a zero of order 10, so the
Tate representation for $\Sigma$ is $10\cdot\fund$.
Along $\Sigma_\pm$, $\beta_{\Sigma_\pm}$ has a zero of order 4 and
$\gamma_{\Sigma_\pm}$ has a zero of order 7, so the Tate representation
for $\Sigma_\pm$ is
\[\frac42\fund + 7\cdot \fund = 9\cdot \fund.\]
Also, from the form of the discriminant in Table~\ref{tab:examples2bis}
we see that each pair of components from $\Sigma$, $\Sigma_+$, and
$\Sigma_-$ meets with intersection multiplicity $3$.

Now from Table~\ref{tab:ex:branching} since $k=1$, the matter representation
is
\[\rho=
(\mathbf{1} \otimes
\fund\otimes\fund) \oplus (\fund\otimes\mathbf{1}\otimes\mathbf{1})
\oplus(\fund \otimes \fund \otimes \fund)
\]
as a representation of $\mathfrak{su}(2)\oplus\mathfrak{su}(3)\oplus
\mathfrak{su}(3)$.
There are several things to check:
\begin{enumerate}
\item
When the matter representation is restricted to
 $\mathfrak{g}(\Sigma)=\mathfrak{su}(2)$,
the charged part is $1+9$ copies of the fundamental representation,
$1$ from the middle term and $9$ from the last term.
This agrees with the Tate representation for $\Sigma$.
\item
When restricted to $\mathfrak{g}(\Sigma_\pm)
= \mathfrak{su}(3)$, the charged part is $3+6$ copies of the fundamental
representation, with $3$ copies coming from the first term and $2\cdot3$
copies coming from the last term.  This agrees with the Tate
representation for $\Sigma_\pm$.
\item
All of the representations here are fundamentals of $\mathfrak{su}(m)$,
so all of the ratios
of traces are $1$.  However, $(\Sigma,\Sigma_\pm)$ has three bifundamentals
(coming from the last term), so the representation-multiplicity is $3$.
\item Similarly, $(\Sigma_+,\Sigma_-)$ has three bifundamentals:
one from the first term in $\rho$ and two from the last term,
so the representation-multiplicity is again $3$.
\end{enumerate}
Thus, the anomalies cancel.

\end{itemize}

\section{Discussion} \label{sec:discussion}

As the examples in the past two Sections have shown, the geometry and
representation theory conspire in wonderful ways to ensure anomaly
cancellation in every case.  One of the remarkable things about the
present approach is how easy it is to calculate the various Tate cycles as
well as
the intersection multiplicites of pairs of divisors, starting from
the Weierstrass equation.  It is natural to wonder whether this data
is sufficient to determine the matter representation itself.

That is, suppose we are expecting an actual pre-quaternionic representation
at a given point.  Whatever it is, it will be Casimir equivalent to
some combination of the basic representations (such as $\adj$, $\fund$,
and $\Lambda^2$ in the $\mathfrak{su}(m)$ case), but that combination
would typically have both positive and negative coefficients.  Could it
be true that there is only one ``honest''
representation (up to complex conjugation
of factors) in the Casimir equivalence class?

We already have an example in hand---example 10---which shows that this
is too optimistic.  But one could hope that the number of representations
is small, and that the number of times there are duplications could
be controlled in some fashion.  This is a purely algebraic question
about the representation ring and the sub-semiring generated by
actual representations which deserves to be studied further.  One might
then hope to find an additional piece of geometric information which
would distinguish among the allowed representations in cases of ambiguity.

There are a number of other ways that this work could be usefully
extended.  First, we assumed that the gauge algebra has no abelian summands,
which geometrically corresponds to Mordell--Weil group of rank zero.
It would be very interesting to study cases with abelian summands
allowed.  Some steps in this direction are taken in \cite{park-taylor}.

Second, although our formulation allows components of the discriminant locus
to have singular points, we have not studied this case in detail.  Some
examples in addition to Sadov's appear in \cite{matter1} but it would
be good to have a more systematic treatment.  For example, the representations
of $\mathfrak{su}(m)$ which we studied in detail here are all Casimir equivalent
to combinations of $\fund$ and $\Lambda^2$ alone; allowing $\adj$ gives
a much richer class of representations (and also necessarily implies
that the relevant discriminant-component $\Sigma$ is singular).

And finally, although we made some comments about elliptic fibrations
with higher-dimensional bases or whose total space is not Calabi--Yau,
many of our results clearly extend to these settings and deserve a
more systematic treatment there.

\appendix
\section{Notation and terminology from gauge theory} \label{app:A}

\begin{lemma}\label{gaugebundle}  Let $Y$ be a manifold equipped  with a principal
$G$-bundle $\mathcal{G}$, called the ``gauge bundle'':
\begin{enumerate}
\item Each fiber $\ad(\mathcal{G})_x$ of $\ad(\mathcal{G})$ is
isomorphic to the Lie algebra $\mathfrak{g}$ of $G$, with
$\mathcal{G}_x$ acting on $\ad(\mathcal{G})_x$ via the adjoint
action of $G$ on $\mathfrak{g}$.
\item The curvature $F$ of the
gauge connection is an $\ad(\mathcal{G})$-valued two-form.
\item Similarly, if $Y$ is equipped with a (pseudo-)Riemannian
metric, then the curvature $R$ of the Levi--Civita connection is a
two-form taking values in the endomorphisms of the tangent bundle.
\item  Any representation $\rho$  of the Lie algebra
can be regarded as a homomorphism $\rho:\mathfrak{g}\to \End(V)$
for some (complex) vector space $V$.
As an endomorphism of $V$,
$\rho(F_x)$ can be raised to the $k^{\text{th}}$ power.
\end{enumerate}
\end{lemma}
\begin{proof}  This can be derived,  for example, from \cite{MR584445}.
\end{proof}
\begin{definition} We denote by
$$\Tr_\rho F^k=\trace_V \rho(F)^k$$
the trace of the resulting endomorphism $\rho(F_x)^k$ of $V$.
\end{definition}
\begin{definition}  Similarly,
 $$\tr R^k= \trace _V v (R)^k,$$
 where  $v$ is the ``vector'' representation of the corresponding
orthogonal group. \end{definition}

 These expressions do not depend on the choice of isomorphism
to $\mathfrak{g}$, in fact  they are invariant under the adjoint
action of $G$ on $\mathfrak{g}$ and so is independent of choices.
 The above notation is then well defined.

\medskip

\section{Tate's algorithm}
\label{app:Tate}

Kodaira's analysis identifies the type of singular fiber along each
component of the discriminant locus, and this is almost enough to
identify the gauge group.  However, there is one additional piece
of information, provided by Tate's algorithm, which
specifies the monodromy of the family of Kodaira fibers over the
component.  That information has traditionally been presented in
the form which Tate gave it -- involving  (generic) changes
of coordinates in the Weierstrass (or Tate) model -- but here we formulate
the same information in a more intrinsic form.

Thus, we start with a Weierstrass equation
\begin{equation}\label{eq:Tate00}
y^2=x^3+fx+g
\end{equation}
and the associated discriminant\footnote{We use the normalization of
the discriminant which is common in the F-theory literature.}
\begin{equation}\label{eq:discrim00}
 \Delta=4f^3+27g^2
\end{equation}
and seek conditions
which specify both the type of the fiber (following Kodaira) and
the monodromy (following Tate).  We regard $f$ and $g$ as elements of
the ring $K[[z]]$ of formal power series in $z$ with coefficients
in  some specified field $K$.  (For simplicity, we assume that $K$
contains the complex number field $\mathbb{C}$.)
Typically, $z$ is a local parameter
whose vanishing describes a divisor $\Sigma$ in some algebraic variety
$B$, and $K$ is constructed from the
 field of rational functions on the algebraic variety by localization
and taking residue field.
Note that in $K[[z]]$ we are allowed to divide by any
nonzero element which is not a multiple of $z$; this means that all
results obtained through this algorithm only hold {\em generically}\/
on the algebraic variety $B$, and may fail to hold in some
particular coordinate systems.

To determine the desired  conditions on $f$, $g$, and $\Delta$, we
follow Tate's procedure described in sections 7 and 8 of \cite{MR0393039},
specializing to the case of the coefficient field $K$ having characteristic
zero, which allows us to simplify certain aspects of Tate's procedure.
We follow the numbering of cases given in Tate's paper. We shall have
occasion to use the
discriminant for a generalized Weierstrass model
\[ Y^2 = X^3 + u X^2 + v X + w,\]
which is easily calculated from \eqref{eq:fg} to be
\[ \Delta=4u^3w-u^2v^2-18uvw+27w^2\]
(cf.\ \cite[Section 1]{MR0393039} or
\cite[Appendix I]{grouprep}).

\underline{Case 1.}
If the fiber along $z=0$ is generically nonsingular, then
$z\mathrel{\not |}\Delta$.
This is case $I_0$ in Kodaira's classification, represented by the first
line of Table~\ref{tab:kodaira}.

\underline{Case 2.}
If the fiber along $z=0$ is generically singular, we may assume
that $z\mathrel{|}\Delta$.  In this case, there is a change of coordinates
which translates the singular point in the fiber to the origin, and
this puts the equation in the form
\begin{equation}\label{eq:Tate01}
Y^2 = X^3 + u X^2 + v X + w
\end{equation}
with $z\mathrel{|}v$ and $z\mathrel{|}w$.  (This change of coordinates can be taken to be
$(x,y)=(X+u/3,Y)$ with $u=-9g/2f$ or $u=2f^2/3g$, at least one of which
is guaranteed not to have a pole along $z=0$.)
We write $v=v_1z$ and $w=w_1z$.

The condition to have type $I_m$ is then that $z\mathrel{\not |}u$ (which is easily
seen to be equivalent to $z\mathrel{\not |}f$ or to $z\mathrel{\not |}g$).
In this case  we define
$m=\operatorname{ord}_{z=0}(\Delta)$: this accounts for the next
few lines of Table~\ref{tab:kodaira}.  Note that in this case,
$u \equiv -9g/2f \mod z$.
The discriminant  can be then expanded as follows:
\[ \Delta = 4 u^3 w_1 z + O(z^2).\]
Note that if $m=1$ then the total space is non-singular and there is
nothing further to do.
On the other hand,
if $m \geq 2$ then $z^2\mathrel{|}w$.  If $m=2$ no monodromy is
possible; if $ m \geq 3$ we can
write $w=w_2z^2$ and compute the discriminant again:
\[ \Delta = (4 u^3 w_2 - u^2v_1^2)z^2 + O(z^3).\]
Thus, if $m \geq 3 $ we can write $w = (v_1^2/4u)z^2 + w_3z^3$.
Now, making the coordinate change $(X,Y) = (\widetilde{X}-(v_1/2u)z,\widetilde{Y})$,
we find
a new equation of the form
\begin{equation}\label{eq:Tate02}
\widetilde{Y}^2 = \widetilde{X}^3 + \widetilde{u}\widetilde{X}^2 + \widetilde{v}_2z^2\widetilde{X} + \widetilde{w}_3z^3
\end{equation}
valid for $I_m$ whenever $m\ge3$.  Note that $\widetilde{u}\equiv u \mod z$.

Now let us resolve  the singularities of the Weierstrass equation \ref{eq:Tate02}; in an
appropriate chart the first blowup has an
exceptional divisor determined by the intersection of $z=0$ and the quadratic terms in equation
\eqref{eq:Tate02}, namely,
\[
y_2^2 = \widetilde{u}x_2^2|_{z=0}
\]
and whether this exceptional divisor is reducible or irreducible is
determined by the Tate monodromy relation, that is, is determined by
whether
$\widetilde{u}|_{z=0}=u|_{z=0}=(-9g/2f)|_{z=0}$
has a square root in the field $K$ or not,
or equivalently, whether the monodromy cover defined by
\[ \psi^2 + (9g/2f)|_{z=0} =0\]
is reducible or irreducible\footnote{The divisor is actually  {\em
absolutely irreducible}\/, that is, it is
irreducible even after passing to an
extension field $L$ of $K$.  The point in this case is that when a
quadratic equation in two variables has maximal rank, it is absolutely
irreducible.}.  In Tate's original algorithm, the reducible
case is called \underline{Case 2a}, and the irreducible case is called
\underline{Case 2b}.  This explains the monodromy entry in the fourth
line of Table~\ref{tab:kodaira}.

\underline{Case 3.}  Now we may assume that $z\mathrel{|}u$ (or equivalently that
$z\mathrel{|}f$ and $z\mathrel{|}g$), and in this case it is convenient to return our attention
to the original Weierstrass form \eqref{eq:Tate00}.
If $z^2\mathrel{\not |}g$, we have a Kodaira fiber
of type $II$ (the next line on the Table).  There is no monodromy issue
in this case.

\underline{Case 4.}  We next assume that in addition, $z^2\mathrel{|}g$.
At this point we need to start resolving the singularity; we can write
$f=f_1z$ and $g=g_2z^2$ and on an appropriate chart of the blowup
(with $x=zx_1$ and $y=zy_1$)
we get an equation
\[ y_1^2 = zx_1^3 + f_1x_1 + g_2.\]
The exceptional divisor is defined by
\[ (y_1^2 =  f_1x_1 + g_2)|_{z=0},\]
and this is absolutely irreducible
 if and only if $z\mathrel{\not |}f_1$, i.e., $z^2\mathrel{\not |}f$.
This gives a Kodaira fiber of type $III$, the next line on the Table.
There is again
no monodromy issue.

\underline{Case 5.}  Now assume in addition $z^2\mathrel{|}f$, which implies that
$z^3\mathrel{|}\Delta$.  In this case our blown up equation is
\[ Y^2 = zX^3 +  zf_2X + g_2,\]
and the
exceptional divisor has equation
\[ Y^2 = g_2|_{z=0}.\]
Again we encounter a monodromy issue: the exceptional divisor is reducible
if $g_2|_{z=0}$ has a square root in the field $K$, and irreducible
otherwise.  Since $g_2|_{z=0} = (g/z^2)|_{z=0}$, this can be expressed
as before in terms of a monodromy cover; the one we need this
time is
\[ \psi^2 - (g/z^2)|_{z=0} = 0.\]
We have have Kodaira type $IV$ exactly when $z^3\mathrel{\not |}g$.

\underline{Cases 6 and beyond.}  If we now assume in addition that
$z^3\mathrel{|}g$, we have arrived at a Weierstrass equation of the form
\begin{equation}\label{eq:Tate03}
y^2=x^3+f_2z^2x+g_3z^3.
\end{equation}
The first blowup now leads to an equation of the form
\[ zy_2^2=x_1^3 + f_2 x_1 + g_3,\]
with exceptional divisor
\begin{equation}\label{eq:exceptional03}
 z=0 = x_1^3 + f_2 x_1 + g_3
\end{equation}
and Tate now describes the algorithm as having three branches, depending
on the number of roots of \eqref{eq:exceptional03}.

\underline{First Branch: Case 6.}  If \eqref{eq:exceptional03} has
three distinct roots, then we have type $I_0^*$.  This is the case
in which $\operatorname{ord}_{z=0}(\Delta)=6$, and the further
behavior is determined by the behavior of the splitting field of that
polynomial.  In our Table, we have rephrased this as the behavior
of the monodromy cover defined by
\[ \psi^3 + (f/z^2)_{z=0} \psi + (g/z^3)|_{z=0} =0.\]

\underline{Second Branch: Case 7.}  When \eqref{eq:exceptional03}
has one simple root and one double root,  we will get Kodaira type $I_{m-4}^*$
for $m = \operatorname{ord}_{z=0}(\Delta)-2$.  There is a monodromy
issue here as well, and to fully determine things, we need to execute
several blowups.  Tate describes these
by means of a subprocedure.  To initialize
the subprocedure, we make a change of coordinates similar to Case 1,
but this time putting the double root of \eqref{eq:exceptional03}
at $X=0$.  This coordinate change can be done via $(x,y)=(X+u_1z/3,Y)$,
with $u_1=-9g_3/2f_2=-9g/2zf$ or $u_1=2f_2^2/3g_3$, at least one of
which is guaranteed not to have a pole along $z=0$.  After the coordinate
change, $v$ and $w$ each vanish to an additional order.  We write $v=v_3z^3$
and $w=w_4z^4$, giving an equation
of the form
\begin{equation}\label{eq:start}
 Y^2 = X^3 + u_1 zX^2 + v_3 z^3X + w_4 z^4.
\end{equation}

\underline{Subprocedure}.
We assume that we are in the second branch, so that $z\mathrel{\not |}u_1$.
We claim by induction on the integer $\mu\ge5$ that if
$\operatorname{ord}_{z=0}(\Delta)\ge\mu+2$ then,
possibly after a change of coordinates, the equation \eqref{eq:start}
can be chosen so that $z^{[(\mu+1)/2]} \mathrel{|} v$ and
$z^{\mu-1}\mathrel{|}w$.
The initial step $\mu=5$ of this induction is what we established in
the paragraph above.
In the course of carrying out the induction, we will also exhibit
various blowups of the singularity and arrive at the monodromy
condition.

It is easiest to divide the analysis of the inductive step
into two cases: $\mu$ odd and $\mu$
even.  To handle the first case, we assume $\mu=2n-1$, $n\ge3$, and
write the equation in the form
\begin{equation}\label{eq:nuodd}
Y^2 = X^3 + u_1 z X^2 + v_{n}z^{n} X + w_{2n-2} z^{2n-2}.
\end{equation}
We compute the discriminant in this case as
\begin{equation}\label{eq:nuodd:discriminant}
 \Delta = 4 u_1^3 w_{2n-2} z^{2n+1} + O(z^{2n+2}).
\end{equation}
One coordinate chart in the blowup has $X=z^{n-1}X_{n-1}$, $Y=z^{n-1}Y_{n-1}$
and the equation takes the form
\begin{equation}\label{eq:nuodd:blowup}
Y_{n-1}^2 = z^{n-1}X_{n-1}^3 + u_1 z X_{n-1}^2
+ v_{n}z X_{n-1} + w_{2n-2},
\end{equation}
with exceptional divisor
\begin{equation}\label{eq:nuodd:exceptional}
(Y_{n-1}^2 =  w_{2n-2})|_{z=0}.
\end{equation}
If $\operatorname{ord}_{z=0}(\Delta) > \mu +2$ then we see from
equation \eqref{eq:nuodd:discriminant} that $z\mathrel{|}w_{2n-2}$.  In
this case, the exceptional divisor is a double line, and the
equation already has the form specified in our inductive statement
for $\mu+1$, so we have established the induction in this case.
On the other hand, if $z\mathrel{\not |}w_{2n-2}$
then there is a monodromy issue: \eqref{eq:nuodd:exceptional}
is reducible if and only if $w_{2n-2}|_{z=0}$ is a square.  From
equation \eqref{eq:nuodd:discriminant} we see that we can write
\[ w_{2n-2}|_{z=0}
= \left.\frac{\Delta}{4z^{2n+1}u_1^3}\right|_{z=0}
= \left.\frac14\left(\frac{\Delta}{z^{2n+1}}\right)\left(\frac{-2zf}{9g}\right)^3\right|_{z=0}
\]
and this is the form in which we expressed the monodromy cover
condition $\psi^2=w_{2n-2}|_{z=0}$ in Table~\ref{tab:kodaira}.

We now consider the case in which $\mu$ is even, and let
$\mu=2n$, $n\ge3$.
We write the equation in the form
\begin{equation}\label{eq:nueven}
Y^2 = X^3 + u_1 z X^2 + v_{n}z^{n} X + w_{2n-1} z^{2n-1}.
\end{equation}
We compute the discriminant in this case as
\begin{equation}\label{eq:nueven:discriminant}
 \Delta = (-u_1^2v_{n}^2 +4 u_1^3 w_{2n+3}) z^{2n+2} + O(z^{2n+3}).
\end{equation}
One coordinate chart in the blowup has $X=z^{n-1}X_{n-1}$, $Y=z^{n}Y_{n}$
and the equation takes the form
\begin{equation}\label{eq:nueven:blowup}
zY_{n}^2 = z^{n-1}X_{n-1}^3 + u_1  X_{n-1}^2
+ v_{n} X_{n-1} + w_{2n-1},
\end{equation}
with exceptional divisor
\begin{equation}\label{eq:nueven:exceptional}
0= (u_1  X_{n-1}^2 + v_{n} X_{n-1} + w_{2n-1})|_{z=0}.
\end{equation}
If $\operatorname{ord}_{z=0}(\Delta) > \mu+2$ then we see from
equation \eqref{eq:nueven:discriminant} that
$z\mathrel{|}(-v_{n}^2+4u_1w_{2n-1})$, so
we may write
\[ w_{2n-1} = \frac{v_{n}^2}{4u_1} + w_{2n}z.\]
Then the equation takes the form
\begin{equation}\label{eq:nuevenbis}
Y^2 = X^3 + u_1 z (X + \frac12v_{n}z^{n-1})^2 + w_{2n} z^{2n}.
\end{equation}
Making the coordinate change
$(X,Y)=(\widetilde{X}-\frac12v_{n}z^{n-1},\widetilde{Y})$
then puts the equation in the form specified for $\mu+1$, so the
induction is established.

On the other hand, if $\operatorname{ord}_{z=0}(\Delta) = \mu+2$
then the exceptional divisor \eqref{eq:nueven:exceptional}
is reducible if and only if the discriminant
$ (v_{n}^2-4u_1w_{2n-1})|_{z=0}$
of the quadratic equation \eqref{eq:nueven:exceptional}
is a perfect square.  Note that, by
equation \eqref{eq:nueven:discriminant}, we can write
\[ (v_{n}^2-4u_1w_{2n-1})|_{z=0}
=\left.\frac{-\Delta}{z^{2n+2}u_1^2}\right|_{z=0}
=\left.\left(\frac{-\Delta}{z^{2n+2}}\right)\left(\frac{-2zf}{9g}\right)^2\right|_{z=0}
\]
and this is the form in which we expressed the monodromy cover
condition $\psi^2=(v_{n}^2-4u_1w_{2n-1})|_{z=0}$
in Table~\ref{tab:kodaira}.

\underline{Branch 3 begins. Case 8:} Now suppose that  \eqref{eq:exceptional03}
has a triple root.  In this case, $z\mathrel{|}u_1$ and it is again convenient
to return our attention to the original Weierstrass form \eqref{eq:Tate00}.
We have $z^3\mathrel{|}f$ and $z^4\mathrel{|}g$, and can write the equation in the form
\[ y^2 = x^3 + f_3z^3x+g_4z^4.\]
The key chart for the blowup is
\[ y_2^2 = z^2 x_2^3 + f_3 z x + g_4.\]
When $z\mathrel{\not |}g_4$ we have Kodaira type $IV^*$, and in this case there is
a monodromy issue: the exceptional divisor is described by
\[ (y_2^2 =  g_4)|_{z=0},\]
which is reducible exactly when $g_4|_{z=0} = (g/z^4)|_{z=0}$
is a perfect square.  This leads to the final ``monodromy cover'' entry
in Table~\ref{tab:kodaira}.

\underline{Branch 3 continues. Case 9:}  We now assume\footnote{Note that right
at this point, Tate's paper has a small typographical error concerning
the exponent of $z$, which he calls $\pi$.} $z^5\mathrel{|}g$.  The relevant
chart for the blowup is
\[ z y_3^2 = zx_2^3 + f_3 x_2 + g_5.\]
If $z\mathrel{\not |} f_3$ then we have Kodaira type $III^*$.  There is no monodromy
issue in this case.

\underline{Branch 3 continues. Case 10:}  We now assume
that $z\mathrel{|}f_3$, i.e., $z^4\mathrel{|}f$.  If $z^6\mathrel{\not |}g$, we get Kodaira type $II^*$
with no monodromy.  If $z^6\mathrel{|}g$, the original Weierstrass equation was
not minimal.  We should start over with
\[ y_3^2 = x_2^3 + f_4 x_2 + g_6.\]

\section{Casimir computations for $\mathfrak{so}(\ell)$.}
\label{app:Casimir}

For $\mathfrak{so}(\ell)$ with $\ell=2m$ or $\ell=2m+1$,
the weight lattice has
generators $\varepsilon_1$, \dots, $\varepsilon_m$.  The Weyl group
${\mathcal W}_{\mathfrak{so}(\ell)}$
acts by permutations and sign changes on the $\varepsilon_i$'s, with
${\mathcal W}_{\mathfrak{so}(2m+1)}=\mathfrak{S}_m\rtimes (\mathbb{Z}_2)^{m}$
while
${\mathcal W}_{\mathfrak{so}(2m)}=\mathfrak{S}_m\rtimes (\mathbb{Z}_2)^{m-1}$
where the product of all signs must be $1$ in the second case.
As is easily seen, the ${\mathcal W}$-invariant polynomials are generated by
the elementary
symmetric functions in $\varepsilon_1^2$, \dots, $\varepsilon_m^2$
in the case of $\mathfrak{so}(2m+1)$, while these are supplemented
by $\prod_i \varepsilon_i$ in the case of  $\mathfrak{so}(2m)$
(and in this latter case, $\sigma_m(\varepsilon_1^2,\dots,\varepsilon_m^2)=
\prod_i \varepsilon_i^2$ becomes superfluous).

The vector
representation has weights
\[
\varepsilon_1, \dots, \varepsilon_m,
-\varepsilon_1, \dots, -\varepsilon_m
\]
for $\mathfrak{so}(2m)$, and
\[
\varepsilon_1, \dots, \varepsilon_m, 0,
-\varepsilon_1, \dots, -\varepsilon_m
\]
for $\mathfrak{so}(2m+1)$.  Thus, in both cases, we have
\[ \operatorname{Tr}_{\operatorname{vect}} F^k = \sum (\varepsilon_i)^k
+\sum (-\varepsilon_i)^k = 2\sum \varepsilon_i^k\]
for every even $k$.  Our convention is to set $\operatorname{tr} =
\frac12 \operatorname{Tr}_{\operatorname{vect}}$, so that
\[ \operatorname{tr}(F^k) = \sum \varepsilon_i^k.\]

Now in the case of $\mathfrak{so}(2m+1)$, the spinor representation has
weights $\frac12(\pm \varepsilon_1 \pm \dots \pm \varepsilon_m)$, whereas for
$\mathfrak{so}(2m)$, the even (resp.\ odd) spinor representation has weights
$\frac12(\pm \varepsilon_1 \pm \dots \pm \varepsilon_m)$ where the total number
of minus signs is even (resp.\ odd).

The Casimir operators evaluate as follows:
\[ \operatorname{Tr}_{\operatorname{spin_*}} F^k
= \frac1{2^k}\sum_{a_1,\dots,a_m} ((-1)^{a_1}\varepsilon_1 + \cdots
(-1)^{a_m}\varepsilon_m)^k,
\]
which we only consider for $k$ even.  In the sum $a_i\in \mathbb{Z}_2$,
and in the case of $\mathfrak{so}(2m)$ we impose the condition
$(-1)^{\sum a_i} = \pm 1$ to distinguish the two spinor representations.
In particular, {\em the number of terms in the sum is equal to the
dimension of $\operatorname{spin_*}$}.
Note that these expressions are all invariant under the Weyl group.

Taking a binomial expansion of these expressions in any case with $\ell\ge5$
other than $\ell=8$, the expansion of
$\operatorname{Tr}_{\operatorname{spin_*}} F^2$
can only involve $\sum_i\varepsilon_i^2$ and that of
$\operatorname{Tr}_{\operatorname{spin_*}} F^4$
can only involve $\sum_i\varepsilon_i^4$ and $\sum_{i<j}\varepsilon_i^2
\varepsilon_j^2$.  The reason $\ell=8$ is special is that in that case,
$\prod_i \varepsilon_i$ is also possible for $k=4$.  In any event,
the implication is that aside from the $\ell=8$, $k=4$ case, any terms
with odd powers of $\varepsilon_i$ will cancel out.

It follows that
\[ \operatorname{Tr}_{\operatorname{spin_*}} F^2
= \frac14 \sum_{a_1,\dots,a_m} (\sum_i \varepsilon_i^2)
= \frac14 \dim(\operatorname{spin_*}) \operatorname{tr}(F^2),\]
and that for $\ell\ne8$,
\begin{align}
 \operatorname{Tr}_{\operatorname{spin_*}} F^4
&= \frac1{16} \sum_{a_1,\dots,a_m}\left( \frac{4!}{4!}\sum_i \varepsilon_i^4
+ \frac{4!}{2!2!}
\sum_{i<j} \varepsilon_i^2\varepsilon_j^2\right)
\\ &= \frac1{16} \dim(\operatorname{spin_*})\left( \operatorname{tr}(F^4)
+ 6\cdot\frac12\left((\operatorname{tr}(F^2)^2)-(\operatorname{tr}(F^4))\right)\right)\\
& = \frac1{16} \dim(\operatorname{spin_*})\left(
3(\operatorname{tr}(F^2)^2)-2(\operatorname{tr}(F^4)\right).
\end{align}

The case $\ell=8$ is special, because we get a term
\[\frac{4!}{1!1!1!1!} (-1)^{a_1+a_2+a_3+a_4} \varepsilon_1\varepsilon_2
\varepsilon_3\varepsilon_4
\]
which evaluates to $24\prod_i\varepsilon_i$ for $\operatorname{spin}_+$
and $-24\prod_i\varepsilon_i$ for $\operatorname{spin}_-$.  Thus, for
$\ell=8$ (using the fact that $\dim(\operatorname{spin_*})=8$) we have
\[ \operatorname{Tr}_{\operatorname{spin_\pm}} F^4
= \frac32(\operatorname{tr}(F^2)^2)-(\operatorname{tr}(F^4))
\pm 12(\prod_i\varepsilon_i).
\]
The key relation for us is:
\[ \operatorname{Tr}_{\operatorname{vect}} F^4 +
\operatorname{Tr}_{\operatorname{spin_+}} F^4 +
\operatorname{Tr}_{\operatorname{spin_-}} F^4 =
3(\operatorname{tr}(F^2)^2),
\]
which follows immediately since $\operatorname{Tr}_{\operatorname{vect}} F^4
= 2\operatorname{tr}(F^4)$.

For completeness, we also work out the adjoint representation of
$\mathfrak{so}(\ell)$.  This time the weights are
$\pm\varepsilon_i\pm\varepsilon_j$ for $i<j$, and when $\ell$ is odd, also
$\pm\varepsilon_i$.    We thus have
\begin{align} \operatorname{Tr}_{\operatorname{adjoint}}(F^2) &=
\sum_{i<j}((\varepsilon_i+\varepsilon_j)^2 + (\varepsilon_i-\varepsilon_j)^2 +
(-\varepsilon_i+\varepsilon_j)^2 + (-\varepsilon_i-\varepsilon_j)^2)
\\ &\quad +(\ell-2m) \sum_i(\varepsilon_i^2+(-\varepsilon_i)^2)\\
&= \sum_{i<j}( 4\varepsilon_i^2+4\varepsilon_j^2) +
 2(\ell-2m)\sum_i \varepsilon_i^2 \\
&= 2\sum_{i\ne j} (\varepsilon_i^2+\varepsilon_j^2)
  +  2(\ell-2m)\sum_i \varepsilon_i^2 \\
&= 4(m-1)\sum_{i} \varepsilon_i^2 +  2(\ell-2m)\sum_i \varepsilon_i^2 \\
&= (2\ell-4)\operatorname{tr}(F^2),
\end{align}
and
\begin{align} \operatorname{Tr}_{\operatorname{adjoint}}(F^4) &=
\sum_{i<j}((\varepsilon_i+\varepsilon_j)^4 + (\varepsilon_i-\varepsilon_j)^4 +
(-\varepsilon_i+\varepsilon_j)^4 + (-\varepsilon_i-\varepsilon_j)^4)
\\ &\quad +(\ell-2m) \sum_i(\varepsilon_i^4+(-\varepsilon_i)^4)\\
&= \sum_{i<j} (4\varepsilon_i^4 + 24\varepsilon_i^2\varepsilon_j^2 +
4\varepsilon_j^4)
+2(\ell-2m) \sum_i\varepsilon_i^4\\
&= 24\sum_{i<j}\varepsilon_i^2\varepsilon_j^2
+ 2\sum_{i\ne j}(\varepsilon_i^4 + \varepsilon_j^4)
+2(\ell-2m) \sum_i\varepsilon_i^4\\
&= 12\left((\sum_i \varepsilon_i^2)^2 - \sum_i \varepsilon_i^4\right)
+ 4(m-1)\sum_{i}\varepsilon_i^4
+2(\ell-2m) \sum_i\varepsilon_i^4\\
&= 12(\operatorname{tr}(F^2)^2)
+ (2n-16)\operatorname{tr}(F^4)
\end{align}
Notice that all terms with odd powers of $\varepsilon_i$ and $\varepsilon_j$
canceled explicitly in the second line of the computation.
In particular, this computation is valid for all values of $\ell$ including
$\ell=8$.

\bigskip

{\bf Acknowledgments:} We would like to thank
S. Katz,
V. Kumar,
K. Rubin,
S. Schafer-Nameki,
J. Silverman,
J. Sully,
Y. Tachikawa,
and
especially W. Taylor
for useful discussions.
DRM thanks the Aspen Center for Physics for hospitality during
various stages of preparation of this paper.


\ifx\undefined\bysame
\newcommand{\bysame}{\leavevmode\hbox to3em{\hrulefill}\,}
\fi

\end{document}